\documentclass[reprint, amsfonts, amssymb, amsmath,  showkeys, nofootinbib,pra, superscriptaddress, twocolumn,longbibliography,floatfix]{revtex4-2}

\usepackage{minitoc}
\usepackage[subfigure]{tocloft}
\usepackage{microtype}
\usepackage{graphicx}
\usepackage{subfigure}
\usepackage{booktabs} % for professional tables

\usepackage{MnSymbol}
\usepackage{makecell} 
\usepackage{algorithmic}

\usepackage[ruled,vlined]{algorithm2e}

\usepackage[utf8]{inputenc}
\usepackage{graphics}
\usepackage{selinput}
\usepackage{bbm}
\usepackage{braket}
\usepackage{amsthm}
\usepackage{mathtools}
\usepackage{bm}
\usepackage{physics}
\usepackage[dvipsnames]{xcolor}
\usepackage{graphicx}
\usepackage{adjustbox}
\usepackage{placeins}
\usepackage[T1]{fontenc}
\usepackage{lipsum}
\usepackage{csquotes}
\usepackage{mathbbol} %identity 
\usepackage[normalem]{ulem}
\def\Cbb{\mathbb{C}}

\def\HC{\mathcal{H}}

\def\LC{\mathcal{L}}
\def\EC{\mathcal{E}}

\def\w{\omega}

\usepackage[makeroom]{cancel}
\usepackage[toc,page]{appendix}
\usepackage[colorlinks=true,citecolor=blue,linkcolor=magenta]{hyperref}

\usepackage{tikz}
\tikzset{every picture/.style=remember picture}

\usetikzlibrary{quantikz2}

\newcommand{\poly}{\operatorname{poly}}

\newcommand{\Ebb}{\mathbb{E}}

\newcommand{\Rbb}{\mathbb{R}}

\newcommand{\FC}{\mathcal{F}}
\newcommand{\GC}{\mathcal{G}}

\newcommand{\MC}{\mathcal{M}}

\newcommand{\OC}{\mathcal{O}}

\newcommand{\QC}{\mathcal{Q}}
\newcommand{\RC}{\mathcal{R}}

\newcommand{\XC}{\mathcal{X}}
\newcommand{\YC}{\mathcal{Y}}

\newcommand{\Var}{{\rm Var}}

\renewcommand{\geq}{\geqslant}
\renewcommand{\leq}{\leqslant}

\renewcommand{\Re}{\text{Re}}
\renewcommand{\Im}{\text{Im}}

\DeclareMathOperator*{\argmax}{arg\,max}

\renewcommand{\vec}[1]{\boldsymbol{#1}}  % Bold vectors instead of arrow vectors

\newcommand{\bs}{\textsf{BS}}

\newcommand{\ep}{\epsilon}

\newcommand{\yv}{\vec{y} }
\newcommand{\alphav}{\vec{\alpha} }

\newcommand{\Om}{\Omega }

\newcommand{\phv}{\vec{\phi}}

\newcommand{\omv}{\vec{\omega}}
\newcommand{\Delv}{\vec{\Delta}} 

\newcommand{\delv}{\vec{\delta}} 

\newcommand{\psv}{\vec{\psi}}
\newcommand{\xv}{\vec{x}}
\newcommand{\pvv}{\vec{p}}
\newcommand{\betav}{\vec{\beta}}
\newcommand{\wv}{\vec{w}}

\newcommand{\av}{\vec{a}}
\newcommand{\zv}{\bm{z}}
\newcommand{\uv}{\bm{u}}
\def\be{\begin{equation}}
\def\ee{\end{equation}}

\def\bs{\begin{split}}
\def\es{\end{split}}
\def\bea{\begin{eqnarray}}
\def\eea{\end{eqnarray}}
\def\w{\omega}

\def\U{\mathrm{U}}

\def\U{\mathrm{U}}

\def\zh{\hat{z}}

\newcommand\Lh{\hat{L}}
\theoremstyle{plain}
\newtheorem{theorem}{Theorem}[section]
\newtheorem{proposition}[theorem]{Proposition}
\newtheorem{lemma}[theorem]{Lemma}
\newtheorem{corollary}[theorem]{Corollary}
\theoremstyle{definition}
\newtheorem{definition}[theorem]{Definition}
\newtheorem{assumption}[theorem]{Assumption}

\theoremstyle{remark}

\makeatletter
\renewcommand\onecolumngrid{
\do@columngrid{one}{\@ne}
\def\set@footnotewidth{\onecolumngrid}
\def\footnoterule{\kern-6pt\hrule width 1.5in\kern6pt}
}
\renewcommand\twocolumngrid{
        \def\footnoterule{
        \dimen@\skip\footins\divide\dimen@\thr@@
        \kern-\dimen@\hrule width.5in\kern\dimen@}
        \do@columngrid{mlt}{\tw@}
}
\makeatother
\usepackage{amssymb}
\usepackage{dsfont}
\begin{document}
\doparttoc % Tell to minitoc to generate a toc for the parts
\faketableofcontents % Run a fake tableofcontents command for the partocs

\title{On Dequantization of Supervised Quantum Machine Learning via Random Fourier Features}

\author{Mehrad Sahebi}
\affiliation{Institute of Physics, \'Ecole Polytechnique F\'ed\'erale de
Lausanne (EPFL), CH-1015 Lausanne, Switzerland}
\affiliation{Center for Quantum Science and Engineering, \'Ecole Polytechnique F\'ed\'erale de
Lausanne (EPFL), CH-1015 Lausanne, Switzerland}

\author{Alice Barthe}
\affiliation{$\langle aQa ^L\rangle $ Applied Quantum Algorithms, Universiteit Leiden, Leiden, the Netherlands}
\affiliation{CERN, Europen Organization for Nuclear Research Geneva, Switzerland}

\author{Yudai Suzuki}
\affiliation{Institute of Physics, \'Ecole Polytechnique F\'ed\'erale de
Lausanne (EPFL), CH-1015 Lausanne, Switzerland}
\affiliation{Center for Quantum Science and Engineering, \'Ecole Polytechnique F\'ed\'erale de
Lausanne (EPFL), CH-1015 Lausanne, Switzerland}
\affiliation{Quantum Computing Center, Keio University, Hiyoshi 3-14-1, Kohoku-ku, Yokohama 223-8522, Japan}

\author{Zo\"{e} Holmes}
\affiliation{Institute of Physics, \'Ecole Polytechnique F\'ed\'erale de
Lausanne (EPFL), CH-1015 Lausanne, Switzerland}
\affiliation{Center for Quantum Science and Engineering, \'Ecole Polytechnique F\'ed\'erale de
Lausanne (EPFL), CH-1015 Lausanne, Switzerland}

\author{Michele Grossi}
\affiliation{CERN, Europen Organization for Nuclear Research Geneva, Switzerland}

\date{\today}

\begin{abstract}
In the quest for quantum advantage, a central question is \textit{under what conditions can classical algorithms achieve a performance comparable to quantum algorithms}--a concept known as \textit{dequantization}. Random Fourier features (RFFs) have demonstrated potential for dequantizing certain quantum neural networks (QNNs) applied to regression tasks, but their applicability to other learning problems and architectures remained unexplored. In this work, we derive bounds on the true risk gap between classical RFF models and quantum models for regression and classification tasks with both QNN and quantum kernel architectures. Furthermore, we provide sufficient conditions under which this gap is small and thus the quantum system can be dequantized via the RFF method.  
We support our findings with numerical experiments that illustrate the practical dequantization of existing quantum kernel-based methods. Our findings not only broaden the applicability of RFF-dequantization but also enhance the understanding of potential quantum advantages in practical machine-learning tasks.     
\end{abstract}

\maketitle

\section{Introduction}

Supervised quantum machine learning (QML) models such as
quantum kernel (QK) methods~\cite{schuld2019quantum,havlivcek2019supervised} and quantum neural networks (QNNs)~\cite{benedetti2019parameterizedCorrect, abbas2021power} have recently gained significant attention with the fundamental objective of seeking advantages over classical machine learning methods. 
In parallel, one line of research known as \textit{dequantization} aims to find efficient classical algorithms that can perform as well as QML models for certain tasks~\cite{schreiber2022classical, shin2024dequantizing, landman2022classicallycorrect, sweke2023potential}. These efforts both provide valuable insights into the potential advantages of QML and support the development of improved classical algorithms.

Random Fourier features (RFF) is a classical algorithm which has attracted attention as a potential dequantization method due to its simplicity and thorough study~\cite{rahimi2008weighted,sutherland2015error,avron2017random,li2019towards}.

It was first introduced by Ref.~\cite{rahimi2007random} to approximate shift-invariant kernels, a family of kernels depending only on the difference between the two inputs (e.g. the Gaussian kernel).
The central idea is to construct a feature map by sampling from the Fourier transform of the target kernel. 
This feature map can be used to approximate the kernel or can be directly used in various classical machine learning tasks such as Ridge Regression (RR) and Support Vector Machines (SVMs). 

\begin{table}[t] 
\caption{\textbf{Results summary}: we provide sufficient conditions for dequantization of different quantum machine learning models.}
\vskip 0.2in
\centering
\begin{tikzpicture}
\draw (-3.5, 0) -- (-3.5, -4.4);
\draw (-0.8, 0) -- (-0.8, -4.4);
\draw (-5.5, -1.65) -- (1.9, -1.65);
\draw (-5.5, -3) -- (1.9, -3);

 \node[anchor=center] (russell) at (-2.15,-0.6)
     {\centering\includegraphics[width=0.3\columnwidth]{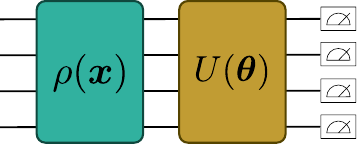}};
\node[anchor=center] (russell) at (0.55,-0.6)
     {\centering\includegraphics[width=0.3\columnwidth]{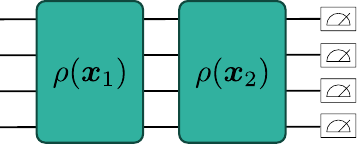}};
\node[anchor=center] (russell) at (-4.6,-2.1)
     {\centering\includegraphics[width=0.155\columnwidth]{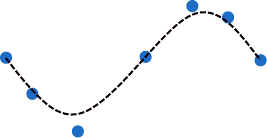}};
\node[anchor=center] (russell) at (-4.6,-3.55)
     {\centering\includegraphics[width=0.155\columnwidth]{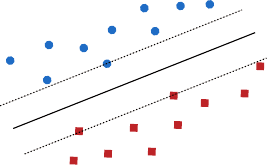}};
 \node[align=center,text width=3] at (-2.45,-1.4)
     {QNN};
\node[align=center,text width=2] at (0.3,-1.4)
     {QK};
\node[align=center,text width=3] at (-5.25,-2.7)
     {Regression};
\node[align=center,text width=2] at (-4.95,-4.2)
     {SVM};

\node[align=center,text width=3] at (-5.2,-0.85)
     {Algorithm};

\node[align=center,text width=2cm] at (-2.15,-2.35)
     {Ref.~\cite{sweke2023potential}};

\node[align=center,text width=2cm] at (0.6,-2.35)
     {Prop.~\ref{prop:1}};

\node[align=center,text width=2cm] at (0.6,-3.7)
     {Prop.~\ref{prop:loose}};
\node[align=center,text width=2cm] at (-2.15,-3.7)
     {Prop.~\ref{prop:tight}};
\end{tikzpicture}
\end{table} \label{tab:main}

RFF is particularly suited for the dequantization of a common class of supervised QML models that are known to have a finite Fourier series representation~\cite{schuld2021supervised, shin2023exponential}. However, among the two most typical supervised QML models, i.e., QNNs~\cite{benedetti2019parameterized, abbas2021power} and QKs~\cite{schuld2019quantum,havlivcek2019supervised}, the use of RFF for dequantization has been limited to QNNs so far. 
For instance, Ref.~\cite{landman2022classicallycorrect} used RFF to dequantize QNN regression tasks and provided bounds on the difference between the optimal RFF and QNN decision functions. 
 Subsequently, Ref.~\cite{sweke2023potential} derived a bound on the \textit{generalization} performance of a QNN and the RFF algorithm for a regression problem.
Nevertheless, the potential of RFF for the dequantization of QKs and for other learning tasks such as SVM classification remains unexplored. For more details about related works see App.~\ref{sec:relatedworks}.

In this work, we extend the RFF-dequantization to a broader class of supervised QML including QK regression, SVMs with QKs (QSVM) and QNNs with SVM loss (QNN-SVM), as summarized in Table~\ref{tab:main}.
Concretely, we adapt RFF to non-shift-invariant kernels and propose sufficient conditions for RFF-dequantization, i.e., for when RFF can efficiently reach a true risk comparable to that of the QML models. 
Finally, we compare the performance of RFF with QSVM on a dataset of particle collisions. 
We further provide numerical experiments to show dequantization in practice with our RFF algorithm performing similarly to a QML model on a real-world dataset.

\section{Preliminaries} \label{sec:framework}
We denote vectors with boldface small letters ($\xv$) and matrices with capital letters ($F$). 
Transpose and complex conjugate transpose of $\xv$ are represented as $\xv^T$ and $\xv^\dagger$.

$\Ebb[\cdot]$, $\expval{\cdot,\cdot}$ and $\Tr[\cdot]$ denote expected values, inner product and trace, respectively. 
For the matrix and vector norm, we use the notation $\norm{\cdot}_p$, which corresponds to $l_1$, $l_2$ (Frobenius), and infinity (Spectral) norm for $p=1,2,\infty$, respectively. Note that the matrix norms here are not induced norms but Schatten norms.  
$\norm{\cdot}_\HC$ represents the norm of a Hilbert space $\HC$. For the quantum systems, we use the Dirac notations. That is, $\ket{\cdot}$ represents a complex-valued vector while $\bra{\cdot}$ denotes its complex conjugate. Then, $\braket{\cdot}$ represents the inner product of two complex-valued vectors.  
\subsection{Classical Framework}
\paragraph*{Supervised Learning.} 
Let $P:\XC \times \YC \rightarrow \Rbb$ and $\FC$ be some unknown distribution and a set of functions mapping $\XC \subset \Rbb^d$ to $\YC$, called hypothesis class, respectively.
Given a set of independent identically distributed (i.i.d.) samples $\MC = \{\xv_i,y_i\}_{i=1}^m$ from $P$ and a (positive) loss function $\LC:\YC \times \YC \rightarrow [0,\infty)$, supervised learning aims to  minimize the  \textit{true risk} defined as $\RC(f) = \Ebb_P[\LC(f(X),Y)]$ over $\FC$.
However, since $P$ is unknown, learning algorithms work with the sample mean of the true risk over $\MC$, known as \textit{empirical risk}.
Some learning algorithms map the input data to a Hilbert space $\HC_0$ called \textit{feature space} using a \textit{feature map} $\phi: \XC \rightarrow \HC_0$, so that the structure of data can be captured better. 

\paragraph*{Kernels.} 
A kernel function $k:\XC \times \XC \rightarrow \Rbb$, or simply a kernel, is a symmetric and real-valued function used for measuring the similarity between a pair of data points $\xv_1$ and $\xv_2$.
Any kernel can be written as the inner product of a feature map $\phi$ in a Hilbert space i.e. $k(\xv_1, \xv_2) = \expval{\phi(\xv_1), \phi(\xv_2)}$. 
For any function $f$ that can be represented by the feature map of a kernel, i.e. $f(\cdot) = \expval{ {v}, \phi(\cdot)} \ \text{for some }\ {v}$, the RKHS norm of $f$ is defined as $\norm{f}_k := \min \{\norm{{v}}_{\HC_0} \ s.t. \ f(\cdot) = \expval{ {v}, \phi(\cdot)} \}$.

When a kernel's output depends solely on the difference between the two inputs, it is called a \textit{shift-invariant (stationary)} kernel. 
\begin{definition}[Shift-invariant Kernel] \label{def:sikernel}
    A \textit{shift-invariant} kernel is a kernel $k$ on $\XC \subset \Rbb^d$ such that $k(x,y) = k(x-y)$.
    Shift-invariant kernels are also called stationary kernels. On the other hand, the kernels that are not shift invariant are called \textit{non-stationary} kernels.
\end{definition} 
A common example of a shift-invariant kernel is the Gaussian kernel $k(\xv_1,\xv_2) = \text{exp}(-{\norm{\xv_1-\xv_2}^2}/ {2\sigma^2})$. However, as we will see, most QKs are non-stationary.  

\paragraph*{Support Vector Machines.} 
SVMs are a widely used machine learning method for classification tasks. The main idea is to find a hyperplane $\wv$ in the feature space that maximizes the margin. That is, one seeks to minimize the regularized \textit{hinge loss}, 
\be \label{eq:hinge}
\min_{\wv \in \HC_0} \frac{\lambda}{2} \norm{\wv}_{\HC_0} +  \frac{1}{m} \sum_{i=1}^m\max\{0, 1-y_i \expval{\wv,\phv(\xv_i)}\}\, ,
\ee
where $\lambda$ is a regularization hyperparameter that controls the separability of data (soft or hard margin).
The dual problem can be solved by only dealing with the inner products of the input feature map pairs i.e. a corresponding kernel $k(\xv_i, \xv_j) = \expval{\phv(\xv_i), \phv(\xv_j)}$.
This allows for the application of various kernel functions in SVMs without having access to or even knowing the underlying feature map. For more information on SVMs see App.~\ref{app:SVM}.

 \paragraph*{Linear Regression.} Linear Ridge Regression seeks to minimize the \textit{mean squared error loss} between the data labels $y_i$ and the trainable labels $\tilde{y}_i(\wv)  = \expval{\wv,\phv(\xv_i)}$. In other words, it seeks to solve the optimization problem
 \begin{equation} \label{eq:reg_mse}
     \min_{\wv \in \HC_0} \frac{\lambda}{2} \norm{\wv}_{2} +  \frac{1}{m} \sum_i (y_i - \expval{\wv,\phv(\xv_i)})^2 \, ,
 \end{equation}
 where again the real positive parameter $\lambda$ serves as regularization.
 Similar to a kernel SVM, a kernel function can be used to solve the corresponding dual problem instead.  
 
\paragraph*{Random Fourier Features.} According to Bochner's theorem (Thm.~\ref{th:boch}), every bounded, normalized, continuous and shift-invariant kernel can be written as 
\be \label{eq:kp}
k(\xv_1, \xv_2)  = \Ebb_{\omv \sim p} [\psv(\xv_1,\omv)^T \psv(\xv_2,\omv)] ,
\ee
where $\psv(\xv,\omv)= \bigl(\cos(\omv\cdot \xv),  \sin(\omv\cdot \xv) \bigr)^T$ and $p$ is a probability distribution obtained by normalizing the Fourier transform of the kernel. RFF algorithm estimates this expectation using samples from the distribution. More precisely, given $D$ i.i.d. samples, $\omv_1, \cdots, \omv_D$,  from the distribution $p$, we can generate the  \textit{random Fourier feature map},
\begin{align} \label{eq:clrff2}
\phv_{D}(\xv) = \frac{1}{\sqrt{D}}
\begin{pmatrix}
\cos(\omv_1 \cdot\xv) \\  \sin(\omv_1 \cdot\xv) \\ \vdots \\ \cos(\omv_D \cdot\xv) \\  \sin(\omv_D \cdot\xv) \, 
\end{pmatrix}.
\end{align}
The feature map can be used with high probability to construct the approximate kernel $\hat{k}(\xv_1,\xv_2) = \expval{\phv_{D}(\xv_1), \phv_{D}(\xv_2)}$ with a point-wise error no more than $\epsilon$, if we choose $D \in \Omega(\frac{d}{\epsilon^2} \log(\frac{\sigma_p}{\epsilon}))$ where $\sigma_p = \Ebb_p[\norm{\omv}_2^2]$~\cite{rahimi2007random}.
For more details on RFF, see App.~\ref{app:RFF}.

\subsection{Quantum Framework}
\paragraph*{Quantum Kernels.} 
QK methods evaluate a kernel on a quantum computer for all pairs of data to construct the Gram matrix,
followed by a classical optimization~\cite{ schuld2019quantum,  havlivcek2019supervised}.
 This work only considers fidelity QKs defined as 
\be \label{eq:DefKernel}
k(\bm{x}_1, \bm{x}_2) = | \bra{0} \U^\dagger(\bm{x}_2) \U(\bm{x}_1)\ket{0}|^2,
\ee
where the encoding unitary $U$ has the form 
$\U(\bm{x})=\prod_{l=1}^L W_l V_l(\bm{x})$ and $\ket{0}$ represents the initial quantum state. 

\paragraph*{Quantum Neural Networks.} QNNs are parametrized quantum models where trainable parameters are optimized to minimize the loss for a learning problem. 
They typically consist of layers of data encoding followed by parametrized gates (unitary transformations) and then the measurement of an observable $O$. 
More concretely, the output of a QNN can be written as
\be \label{eq:pqc}
f_{\bm{\theta}}(\bm{x}) = \Tr[U(\bm{x}, \bm{\theta})\ketbra{0}{0} U^\dagger(\bm{x}, \bm{\theta})O] \, .
\ee
Here $
U(\bm{x}, \bm{\theta}) = \prod_{l=1}^L W_l(\bm{\theta}) V_l(\bm{x}) 
$ is an $L$-layered circuit where $W_l(\bm{\theta})$ is a parameterized unitary with a vector of trainable parameters $\bm{\theta}$ and the $V_l(\bm{x})$ are fixed unitaries that encode the input $\bm{x}$ into the circuit.
In the training phase of QNNs, the output function $f_{\bm{\theta}}(\bm{x})$ is evaluated on a quantum computer, 
as queried by an optimization performed on a classical computer.
These models have been used for both regression~\cite{mitarai2018quantum, benedetti2019parameterized} and SVM classification~\cite{blance2021quantum,innan2023enhancing}. 

\paragraph*{RFF Dequantization.}
Throughout this work, we employ the following common encoding strategy, which we call a \textit{Hamiltonian encoding} \cite{schuld2021supervised, shin2023exponential}. 
\begin{assumption} \label{ass:Hencoding}
    We assume that the QML model uses a Hamiltonian encoding, where a gate in the $l$-th layer that encodes $j$-th element of data $\xv$ is represented as $V_l(\xv) = \exp(iH^{(j)}_l\xv_j)$ for a Hermitian matrix $H^{(j)}_l$. 
\end{assumption}
It is known that QNNs using this type of encoding have the following Fourier representation~\cite{schuld2021effect}: 
\begin{align}    \label{eq:pqcft}
f_{\bm{\theta}}(\bm{x}) &= \sum_{\bm{\omega} \in \Omega} c_{\bm{\omega}}(\bm{\theta}) e^{i\bm{\omega}\bm{x}}   \\
&= a_{\bm{0}} + \sum_{\bm{\omega} \in {\Omega_+}} a_{\bm{\omega}}(\bm{\theta}) \cos(\bm{\omega}\bm{x}) + b_{\bm{\omega}}(\bm{\theta}) \sin(\bm{\omega}\bm{x}) \nonumber,
\end{align}
where \textit{the frequency support} of the encoding, $\Omega\subseteq \Rbb^d$, depends on the encoding Hamiltonians only and always contains $0$.
Also, it is symmetric around zero, i.e., $\forall \omv \in \Omega, \ -\omv \in \Omega$. For ease of notation, we denote a subset containing only positive frequencies as $\Omega_+$ without frequency zero. A detailed description of $\Omega$ is presented in App.~\ref{app:QK}.

Note that the output of a QNN in Eq.~\eqref{eq:pqcft} can be written as the inner product between a vector $\bm{\nu}= [a_0, a_{\omv_1}, b_{\omv_1}, \cdots,  a_{\omv_{|\Omega_+|}}, b_{\omv_{|\Omega_+|}}]^T$ of Fourier coefficients and the QNN feature map defined as 
\begin{align} \label{eq:pqcfm}
\phv_{\text{QNN}}(\xv) =
\begin{pmatrix}
1  \\ \cos(\omv_1 \cdot \xv) \\  \sin(\omv_1 \cdot \xv) \\ \vdots \\ \cos(\omv_{|\Omega_+|} \cdot\xv) \\  \sin(\omv_{|\Omega_+|}\cdot \xv) 
\end{pmatrix}.
\end{align}
This means that the hypothesis class of a QNN can be represented by a cosine feature map. Consequently, RFF presents itself as a natural method for the dequantization of such models. This was originally highlighted in Ref.~\cite{landman2022classicallycorrect} which used the random Fourier feature map (Eq.~\eqref{eq:clrff2}) to dequantize QNN for regression tasks, using the techniques in Ref.~\cite{sutherland2015error}. 

Subsequently, Ref.~\cite{sweke2023potential} employed the learning errors provided by Ref.~\cite{rudi2017generalization} and bounded the gap between the true risk of the RFF empirical risk minimizer and the QNN true risk minimizer for ridge regression problems. 
These bounds depend on the properties of the sampling distribution $p$ and are then used to derive sufficient conditions under which RFF can match the performance of QNNs with efficient resources. 
We will discuss these conditions later in Sec.~\ref{sec:regression}, but before that, let us follow Ref.~\cite{sweke2023potential} and formalize RFF-dequantization as follows. For more details about related works see App.~\ref{sec:relatedworks}. 

\begin{definition}[RFF Dequantization] \label{def:rff_dequantization}
Consider a supervised learning task with a training dataset $\{\xv_i,y_i\}_{i=0}^m  \  \text{s.t.} \ \xv\in \XC \subset \Rbb^d$. 
Given a quantum model (either a QNN or QK) with hypothesis class $\QC$, $f_q$ denotes the true risk minimizer of the quantum model i.e. $f_q=\min_{f \in \QC} \RC[f]$.
We say the task is \textit{RFF-dequantized}, if there exists a distribution $p:\Omega\rightarrow \Rbb$ such that for some $\epsilon > 0$, with $m, D \in \OC(\poly(d, \epsilon^{-1}))$, the following is true with high probability:
\be 
\RC[f_{D,p}] - \RC[f_q] \leq \epsilon.
\ee
Here, $f_{D,p}$ is the decision function of the task trained with a $D$ dimensional RFF feature map (Eq.~\ref{eq:clrff2}) sampled from $p$ and $\RC$ denotes the true risk. 
\end{definition}
Informally, we say the task is RFF-dequantized, if the performance of RFF with polynomial resources is not much worse than the best possible QML model.

\section{RFF Dequantization of QML models}
In the context of the dequantization of QML models via RFF, there are two main issues that can be explored. The first is \textit{whether RFF can approximate quantum kernels}. This is motivated by the observation that, if a good approximation can be found, a similar performance is achievable using the approximate kernel~\cite{sutherland2015error}. Since QKs are non-stationary in general, the original RFF approximation algorithm does not apply to them directly. However, we can use similar ideas.    
Namely, we express a QK in the form of Eq.~\eqref{eq:kp} for some $\psv$ and then replace the expectation with the sample mean.

\begin{lemma} \label{lem:QKYaglom}
    A QK $k_Q$ as in Eq.~\eqref{eq:DefKernel} using a Hamiltonian encoding (Assumption~\ref{ass:Hencoding}) can be written as
\be \label{eq:qkft}
k_Q(\bm{x}_1,\bm{x}_2) = \sum_{\bm{\omega}, \bm{\nu} \in {\Omega}} F_{\bm{\omega}\bm{\nu}} e^{i(\bm{\omega}\cdot\bm{x}_1 - \bm{\nu} \cdot\bm{x}_2)},
\ee
where $F$ is a positive semi-definite, unit-trace ($\sum_{\bm{\omega}} F_{\bm{\omega}\bm{\omega}} = 1$) matrix and $F_{-\bm{\omega}, -\bm{\nu}} = F^*_{\bm{\omega} \bm{\nu}}$, and $\Omega$ is a set of vector frequencies determined by the eigenvalues of encoding Hamiltonians.  
\end{lemma} 

Although $F$ is diagonal if and only if the kernel is shift invariant, the diagonal elements of $F$ always form a valid probability mass function. That is, they are all non-negative (the positive semi-definite property of $F$) and sum to one (the unit-trace property of $F$). This motivates the following definition.   
\begin{definition}[Kernel distribution]\label{def:FT}
    Given a Hamiltonian-encoded QK in Eq.~\eqref{eq:qkft}, we call $F$ its Fourier transform and $\Omega$ its frequency support.
    We also define the diagonal distribution (probability mass function) $q: \Omega \rightarrow \Rbb$ of the kernel as the diagonal elements of its Fourier transform, i.e., $q_{\omv} = F_{\omv \omv}$. 
\end{definition}

With the property in Lem.~\ref{lem:QKYaglom}, we propose a RFF algorithm for approximating non-stationary kernels with a discrete Fourier spectrum (Alg.~\ref{alg:ApproxKernelEVD}) in App.~\ref{app:approx}. 
The main idea of this scheme is to sample frequencies from eigenvalues and eigenvectors of $F$. We show that $D \in \bm{\Omega}\bigl(\frac{|\Omega|^2 d}{\epsilon^2}\log(\frac{1}{\epsilon^2})\bigr)$ frequency samples are required to reach $\epsilon$ point-wise error with high probability, which is consistent with the bound for shift-invariant kernels introduced in Ref.~\cite{rahimi2007random}. However, the number of frequencies in $\Omega$ scales exponentially with $d$, the dimension of input data. This suggests prohibitive computational costs to obtain $F$ even for relatively small system sizes.

Nonetheless, the ability to efficiently approximate the kernel is only a sufficient but not necessary condition to outperform a given QML model using RFF. 
This leads us to the second, and arguably more important, issue - the possibility of \textit{dequantization} (Def.~\ref{def:rff_dequantization}). That is, \textit{can a classical algorithm using random Fourier features have a generalization performance comparable to that of QML methods?} We will thus proceed to discuss the dequantization of QML models with RFF methods for regression and SVM classification.

We begin with a brief review of the sufficient conditions for QNN regression established in Ref.~\cite{sweke2023potential} and introduce a new condition for QK regression (Sec.~\ref{qk_regression}).
Then, we provide sufficient conditions for SVM classification tasks using QKs (Sec.~\ref{QSVM}) and QNN-SVM (Sec.~\ref{PQC_SVM}). This is summarized in Table~\ref{tab:main}.

\subsection{Dequantization for Regression} \label{sec:regression} 
\subsubsection{QNN Regression by Ref.~\texorpdfstring{\cite{sweke2023potential}}{Lg}}\label{sec:PQCregression}

Bounds on the difference between the true risk of RFF ridge regression and QNN regression were derived in Ref.~\cite{sweke2023potential}. These were then used to provide the following set of sufficient conditions for regression using QNNs to be RFF-dequantizable with a distribution $p$.

\begin{itemize}
    \item \textbf{Sampling Efficiency:} The distribution $p$ is easy to sample from. 
    \item \textbf{Concentration:} The maximum of the distribution $p_{\max}$ scales inversely polynomial with the problem size i.e., $p_{\max}^{-1} \in \OC(\poly(d))$. 
    \item \textbf{Alignment:} $p_{\omv} \propto |{c}_{\omv}|$ where ${c}_{\omv}$ are the Fourier coefficients of the best quantum model $f_q$ (Eq.~\eqref{eq:pqcft}).
\end{itemize}

The sampling efficiency condition is clearly essential for the RFF algorithm to be viable. Given this, we do not explicitly mention it in our technical results shown hereafter.

The concentration condition suggests that $p$ should \textbf{not} be very ``flat'' or close to uniform over $\Omega$. Namely, if the sampling distribution is very close to uniform, $p_{\max}^{-1} = |\Omega|$ is exponential in $d$.

The alignment condition indicates that a frequency that is more prominent in $f_q$, should be proportionally more probable to be in the RFF feature map. We note that this condition appears differently in Ref.~\cite{sweke2023potential}. The original form of this condition is $\norm{f_q}_{k_p} \in \OC(\poly(d))$ where $k_p$ is the kernel whose Fourier transform is the distribution $p$ (defined in Eq.~\eqref{eq:kp}). Yet, as the distribution $p$ can be regarded a hyperparameter of RFF, we can minimize $\norm{f_q}_{k_p}$ over choices of $p$. By doing so, we obtain the alignment condition as presented above; see App.~\ref{app:alignment} for more details. 
Fig.~\ref{fig:summary} provides a visual sketch of the interplay of the concentration and alignment conditions. 

\begin{figure}
\vskip 0.2in
\includegraphics[width=\columnwidth]{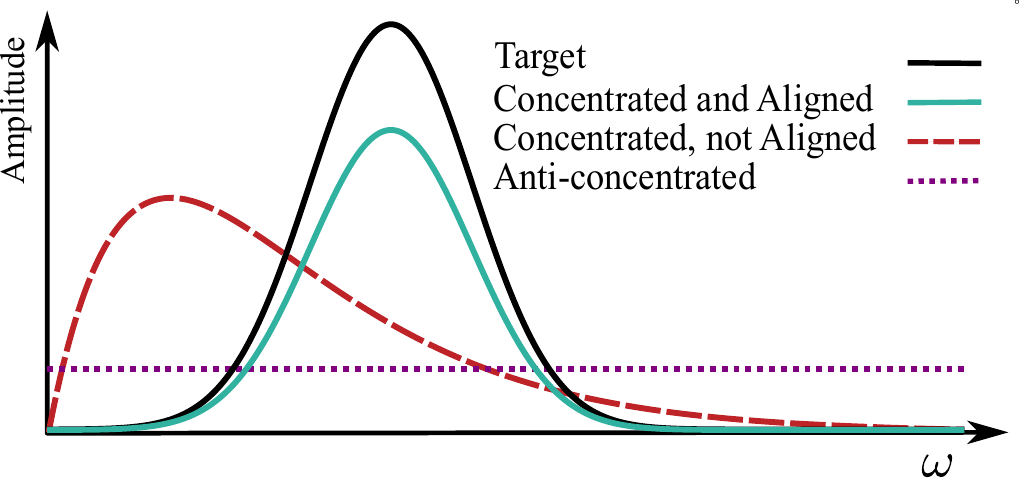}
\caption{\textbf{Alignment and concentration conditions}. For QK methods,
given a distribution $q_{\omv}$ corresponding to the diagonal of the Fourier transform $F$ of the target QML model (Black), the model is dequantizable if the probability distribution $p_{\omv}$ used for RFF is aligned to the target and concentrated (Green).
On the other hand, when $p_{\omv}$ is not aligned (Red) or concentrated (Purple), it might not be dequantizable. }
\label{fig:summary}
\vskip -0.2in
\end{figure}

\subsubsection{QK Regression} \label{qk_regression}
Here, we provide a set of sufficient conditions for RFF to dequantize (Def.~\ref{def:rff_dequantization}) QK regression. 
Our complete proof is provided in App.~\ref{app:regression}. This proof contains a substantial advancement from the previous results. Namely, we extend the scope of RFF to non-stationary kernels to encompass QKs which have non diagonal Fourier transforms (Def.~\ref{def:FT}).

\begin{proposition}[Sufficient conditions for RFF dequantization of QK regression] \label{prop:1}
    Consider a regression task using a quantum kernel $k_q$ with Fourier transform $F$. 
    Then, the linear ridge regression with the RFF feature map (Eq.~\eqref{eq:clrff2}) sampled according to distribution $p$ dequantizes the QK ridge regression if:  
    \begin{itemize}
        \item \textbf{Concentration:} $p_{\max}^{-1} \in \OC(\poly(d))$ 
        \item \textbf{Alignment:} $\norm{\sqrt{P}^{-1} \sqrt{F}}_\infty \in \OC(\poly(d))$ where  $P$ is a diagonal matrix with $p$ as its diagonal
        \item \textbf{Bounded RKHS norm:} $\norm{f_q}_{k_q} \in \OC(\poly(d))$ where $f_q$ is the true risk minimizer of kernel regression.
    \end{itemize}
\end{proposition}

Indeed, the conditions we have found here for the dequantization of QK regression are relatively similar to those for QNN regression in Sec.~\ref{sec:PQCregression}. In particular, the concentration conditions are identical in both cases. In contrast, the alignment condition has a different form and there is an additional condition on the RKHS norm that does not depend on the sampling distribution $p$. Below, we discuss these conditions more in depth.

To gain an intuition for the alignment condition in this case, let us start considering the case where the QK is shift invariant. If $k_q$ is shift invariant then $F$ is diagonal and then, as shown in Lem.~\ref{lem:appFeP} of the App.~\ref{app:regression}, a sampling distribution completely aligned with the kernel ($P=F$) minimizes $\norm{\sqrt{P}^{-1} \sqrt{F}}_\infty$ and the form of the condition reduces to the $p_{\omv} \propto |{c}_{\omv}|$ alignment condition for QNN regression. In this manner, the alignment condition here can be viewed as a generalization of the alignment condition in Ref.~\cite{sweke2023potential} to non-stationary kernels. 

We stress that the RKHS norm condition applies to the norm of the \textit{optimal decision function} $f_q$ of the QML model with respect to the QK $k_q$. Intuitively, when this norm is large, the model is too complex, as sketched in Fig.~\ref{fig:norm}, and the generalization may be poor~\cite{huang2021power}. 
In other words, QML models that do not satisfy the first condition are likely to have a bad performance in the first place and not be of interest for dequantization. Thus, this final condition can be seen as a condition to ensure a good performance for the QML model in the first place.   

\begin{figure}
\begin{center}
\centerline{\includegraphics[width=0.9\columnwidth]{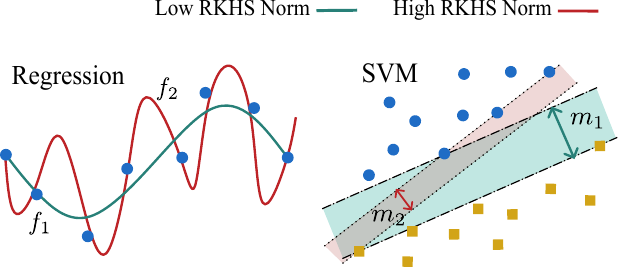}}
\caption{\textbf{Bounded RKHS norm condition.} 
The RKHS norm describes the model complexity and is thus associated with the generalization performance.
In regression, a well-behaved output $f_1$ (green) has the small RKHS norm, whereas $f_2$ that is more complex has a large norm. In SVM classification, the inverse of the norm is equal to the margin.
}
\label{fig:norm}
\end{center}
\vskip -0.2in
\end{figure}

\subsection{Dequantization for SVM Classification}  \label{sec:svm}
We now proceed to extend the results of RFF dequantization to SVM classification. That is, we consider what changes when optimizing the hinge loss (Eq.~\eqref{eq:hinge}) instead of the mean square error loss (Eq.~\eqref{eq:reg_mse}). 
Similarly to regression models, we consider two types of QML models: QSVM and QNN-SVM. The former solves the dual SVM problem with a kernel matrix evaluated on a quantum computer, whereas the latter solves the primal SVM problem with the output of a QNN as a parametrized, trainable decision boundary.   
To dequantize the QML models, we introduce an RFF-based SVM algorithm inspired by Ref.~\cite{rahimi2008weighted} and then propose a set of sufficient conditions for dequantization.

\subsubsection{Quantum SVM} \label{QSVM}

First, we discuss our proposed RFF algorithm for SVM, Alg.~\ref{alg:mainSVM}, which is inspired by Ref.~\cite{rahimi2008weighted}. This algorithm takes a distribution over frequencies as input and samples from it to construct a random cosine feature map in Eq.~\eqref{eq:clrff2}. The SVM optimization problem is then solved with this feature map (See Alg.~\ref{alg:mainSVM}).

The \textit{search radius} $C$, which appears as a constraint in optimization, plays a pivotal role in our theoretical analysis. 
However, in practice, it just needs to be chosen as a sufficiently large number. The generalization analysis of this algorithm is presented in App.~\ref{app:QKSVM}.

\begin{algorithm}[t]
\caption{RFF-SVM} \label{alg:mainSVM}
\begin{algorithmic}
\STATE {\bfseries Input:} A training dataset $\{\xv_i,y_i\}_{i=1}^m$, an integer
$D$, a scalar $C$ as search radius, a scalar $\lambda$ as regularization parameter and a probability distribution $p$ on a set $\Omega$.
\STATE{\bfseries Output:} A function  $$\hat{f}(\xv) = \sum_{j=1}^D \betav_j \cos(\omv \cdot \xv) + \betav_{j+D} \sin(\omv \cdot \xv)$$
\noindent \hrule
\vspace{5pt}
\STATE \textbf{1:} Draw $\omv_1,\cdots,\omv_D$ i.i.d. from $p$
\STATE \textbf{2:} Construct the feature map: 
$\zv(\xv_i)= \sqrt{D} \phv_D(\xv)$ with $\phv_D(\xv)$ defined in Eq.~\eqref{eq:clrff2}.
\STATE {\textbf{3:} With $\omv$'s fixed, solve the empirical risk minimization problem:
\begin{align*}
\min_{\betav \in \Rbb^D} \quad &  \frac{\lambda}{2}\norm{\betav}_2^2 + \frac{1}{m}\sum_{i=1}^m \max\{0,1-\betav^T \zv(\xv_i)y_i\} \\ 
s.t & \quad \norm{\betav}_\infty \leq C/D
\end{align*}}
\end{algorithmic}
\end{algorithm}

The following sufficient conditions for RFF-dequantization of QSVM are derived in App.~\ref{app:main results}.

\begin{proposition} [Sufficient Conditions for RFF dequantization of QSVM] \label{prop:loose}
Let $f_q$ be the true risk minimizer of QSVM trained with a QK $k_q$, whose frequency support and diagonal distribution (Def.~\ref{def:FT}) are denoted as $\Omega$ and $q$, respectively. 
Then, Alg.~\ref{alg:mainSVM} with sampling probability $p: \{1,\cdots,|\Omega|\} \rightarrow \Rbb$ can dequantize the QML model, if: 
\begin{itemize}
    \item \textbf{Concentration:} $\sum_{\omv \in \Omega} \sqrt{q_{\omv}} \in \OC(\poly(d))$
    \item \textbf{Alignment:} $p_{\omv} = \frac{\sqrt{q_{\bm{\omega}}}}{\sum_{\bm{\nu}}\sqrt{q_{\bm{\nu}}}}$ 
    \item \textbf{Bounded RKHS norm: }$\norm{f_q}_{k_q} \in \OC(\poly(d))$. 
\end{itemize} 
\end{proposition}
Prop.~\ref{prop:loose} implies that RFF can dequantize QSVM problems under a broadly similar set of conditions as regression.
The biggest difference between the cases of SVM and regression is found in the square root appearing in the first two conditions; however, they do still intuitively capture concentration and alignment respectively. 

To illustrate how the first condition captures concentration, let us consider two examples of distributions: delta function and uniform distribution. 
For the delta distribution on frequency zero, the measure is $1$ showing maximum concentration. For a uniform distribution, the concentration measure is $\sum_i \sqrt{|\Omega|}^{-1} = \sqrt{|\Omega|}$ which is exponential in $d$ showing  anti-concentration.   Moreover, this concentration measure is closely related to $1/2$-R\'enyi entropy of the diagonal distribution defined as $H_{1/2} = 2\log(\sum_i \sqrt{q_i})$. In particular, this entropy is an upper bound for the tail of the distribution~\cite{verstraete2006matrix2} further emphasizing its role as a concentration measure. 

The second condition here implies that the optimal sampling distribution for the dequantization of QSVM should be proportional to the square root of the diagonal distribution. This arises from the fact that we use bounds provided by Ref.~\cite{rahimi2008weighted}, which only works for Lipschitz loss functions such as SVMs.
On the other hand, the bounds provided by Ref.~\cite{rudi2017generalization} are used in Sec.~\ref{sec:regression}, because the regression task deals with the non-Lipschitz mean square loss. Nonetheless, the condition here can still be viewed as an alignment condition. We further suspect, but have not been able to prove, that the alignment condition in Prop.~\ref{prop:1} would also suffice in this case. 

Finally, similarly to in Prop.~\ref{prop:1}, we require the RKHS norm of the true risk minimizer to be polynomially bounded.  In fact, this condition on the RKHS norm has a more immediate interpretation in the current SVM case as this norm is proportional to the inverse of the SVM margin. Thus this condition states that we need the data in the feature space to be separable with a margin that decreases at worst inverse polynomially in $d$. 

\subsubsection{QNN-SVM} \label{PQC_SVM}
Lastly, we move on to the dequantization of QNN-SVM. The following proposition is proved as Prop.~\ref{prop:tightapp} of the App.~\ref{app:regression}.

\begin{proposition} [Sufficient Conditions for RFF dequantization of QNN-SVM] \label{prop:tight}
 Let $f_q$ be a true risk minimizer of QNN-SVM. Then, Alg.~\ref{alg:mainSVM} with sampling probability $p: \Omega \rightarrow \Rbb$ can dequantize the QML model, if
 \begin{itemize}
     \item \textbf{Alignment:} $p_{\omv} \propto |{c}_{\omv}| $  where ${c}_{\omv}$ are the Fourier coefficients of the best quantum model $f_q$ (Eq.~\eqref{eq:pqcft})
    \item  \textbf{Bounded Fourier Sum:}  $ \sum_{\omv \in \Omega} |{c}_{\omv}| \in \OC(\poly(d))$ 
 \end{itemize}
\end{proposition}

 This proposition shows that for QNN-SVMs, despite the difference in loss function and the learning bounds used, we have the same alignment condition as Ref.~\cite{sweke2023potential}. The second condition is a bound on the 1-norm of the vector $c$ of Fourier coefficients of $f_q$. Compared to the bound on the RKHS norm (equal to the 2-norm of $c$) which appears in Props.~\ref{prop:1} and~\ref{prop:loose}, it is related, but stricter:
 $$ \norm{f_q}_{k_q}=\norm{c}_2 \leq \norm{c}_1\,.$$

\medskip 
In App.~\ref{app:CondReg}, we show that these conditions also work as sufficient conditions for the dequantization of QNN regression. Moreover, we compare this alternative set of conditions to the conditions proposed by Ref.~\cite{sweke2023potential}. We show that the bounded Fourier sum condition is a broader condition than the concentration condition of Ref.~\cite{sweke2023potential}. In other words, whenever the concentration condition of Ref.~\cite{sweke2023potential} is satisfied, the bounded Fourier sum condition is also satisfied. 

Overall, Props.~\ref{prop:1},~\ref{prop:loose} and~\ref{prop:tight} suggest that the possible advantage of QML models with Hamiltonian encodings lies where the QML model performs well and has an anti-concentrated Fourier spectrum, or at least a spectrum complex enough such that it is hard to sample from its corresponding Fourier spectrum. In all other cases, it is \textit{theoretically} possible to dequantize QNN and QK regression and classification tasks. 

That said, it is important to stress that these theoretical guarantees are in several ways not of practical use. In particular, our definition of dequantization (Def.~\ref{def:setup}), following that in Ref.~\cite{sweke2023potential}, requires only the \textit{existence} of a sampling strategy $p$ where RFF performs as well as the QML model, not that this distribution is \textit{known}. Moreover, the sufficient conditions require this distribution to be aligned with the Fourier transform of the QK or the decision function of the QNN. Computing the Fourier transform of a QML model in general either requires an exponential (in the data dimensionality $d$) number of calls to a quantum computer or exponential (in the number of qubits $n$) classical memory. Even if the Fourier transform is known, there is no guarantee that the aligned distribution is easy to sample from. An interesting question is whether QML models can dequantized using  RFF distributions that are independent or loosely dependent on the kernel and easy to sample from. To probe this question we provide some numerical results in the next section.   

Finally, many implementations of QML are affected by the barren plateau phenomenon~\cite{mcclean2018barren, larocca2024review} or, more generally, exponential concentration~\cite{thanasilp2022exponential}. Recently Ref.~\cite{cerezo2023does} conjectures that provably barren-plateau-free models with classical data inputs can effectively be dequantized. In App.~\ref{app:bp} we discuss the connection between barren plateaus, exponential concentration and RFF-dequantization for QNNs and QKs.

\section{Numerical experiments}\label{subsec:numerics}
In this section, we investigate when QML models can be dequantized using RFF \textit{in practice}. There are two key features of practical implementations of RFF and QNNs/QKs that are not captured by our theoretical analysis above. Firstly, as just discussed, we do not have access to the optimal RFF sampling distribution. Therefore, we look at the performance of RFF with simple, easy-to-sample-from, distributions that are loosely dependent on, or inspired by, the QK. Secondly, our analysis above compared the performance of RFF to the performance of QNNs/QKs assuming access to the exact loss/kernel values. However, in practice, as these will be measured on a quantum computer, we only ever have estimates of these values computed from a finite number of shots. Since the loss functions of QNNs and QC values often concentrate exponentially to some fixed data independent point~\cite{mcclean2018barren, larocca2024review, thanasilp2022exponential}, the unavoidability of shot noise can have a significantly detrimental effect on the quantum model performance.

For concreteness, we compare the performance of QSVMs with RFF-SVMs using data from high-energy physics. In particular, we use the dataset of Ref.~\cite{wozniak2023quantum} which consists of simulated proton-proton collisions.
It includes a mix of Standard Model (SM) and Beyond Standard Model (BSM) processes, with numerous features simulated using Monte Carlo, whose dimension is reduced to 16, 32, and 64 with auto-encoders. Ref.~\cite{wozniak2023quantum} use this dataset for anomaly detection with QSVM. Here, 
we use the dataset for two-class classification with SVM, using a QK inspired by the same work that is detailed in Fig.~\ref{fig:kernel} of the App.~\ref{app:dataset}. More details about the dataset and implementation are provided in App.~\ref{app:dataset}.

\begin{figure}
\vskip 0.2in
\begin{center}
\centerline{\includegraphics[width=\columnwidth]{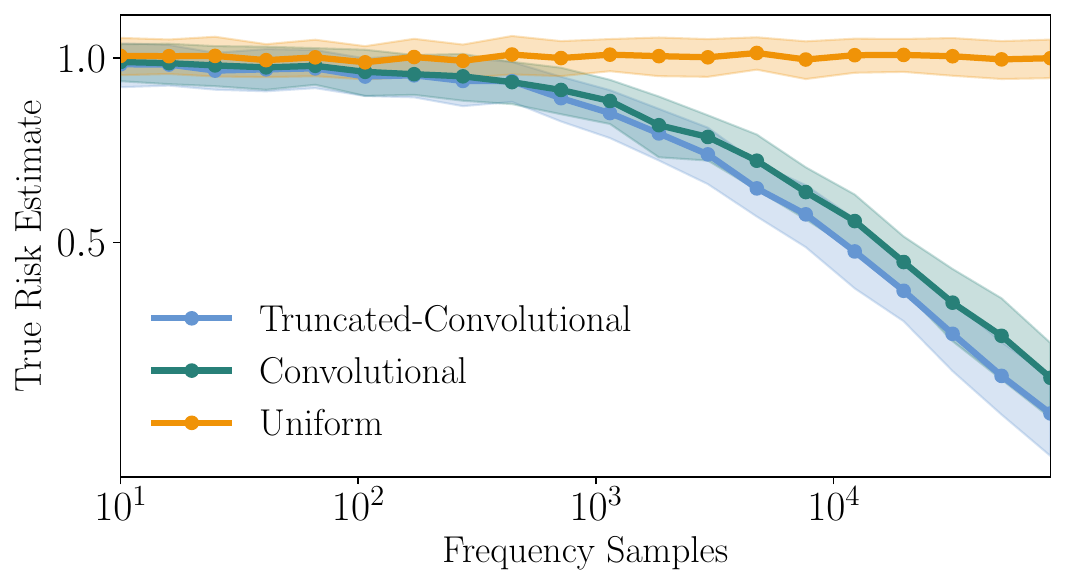}}
\caption{\textbf{Comparison of Sampling Strategies.} For a high-energy physics dataset, the true risk of a RFF-SVM is plotted against the number of frequency samples $D$ for the RFF, for three different distributions. The solid lines are the average performance over 60 runs of the RFF algorithm, and the shaded regions represent the standard deviation. More details can be found in App.~\ref{app:dataset}.} 
\label{fig:sampling comp}
\end{center}
\vskip -0.2in
\end{figure}

We consider three separable RFF sampling distributions.
First, we use the ``uniform'' distribution over the set of all possible frequencies. 
Secondly, we use the fact that for QKs, lower frequencies are, on average, more likely to appear than higher ones. More precisely, the frequency profile has been shown to be the result of the convolution of the spectrum of the generators of the encoding~\cite{barthe2024gradients}. For the kernel considered here, this yields the line from Pascal's triangle with the same number of elements as the number of frequencies available for each dimension. We call this ``convolutional'' sampling (equivalent to Tree sampling in Ref.~\cite{landman2022classicallycorrect}). Thirdly, we cut the upper half of the frequency support from the convolutional distribution, and call it ``truncated-convolutional'' sampling. A sketch of these three sampling strategies is provided in Fig.~\ref{fig:sampling} in the Appendices.  

\begin{figure}
\vskip 0.2in
\begin{center}
\centerline{\includegraphics[width=\columnwidth]{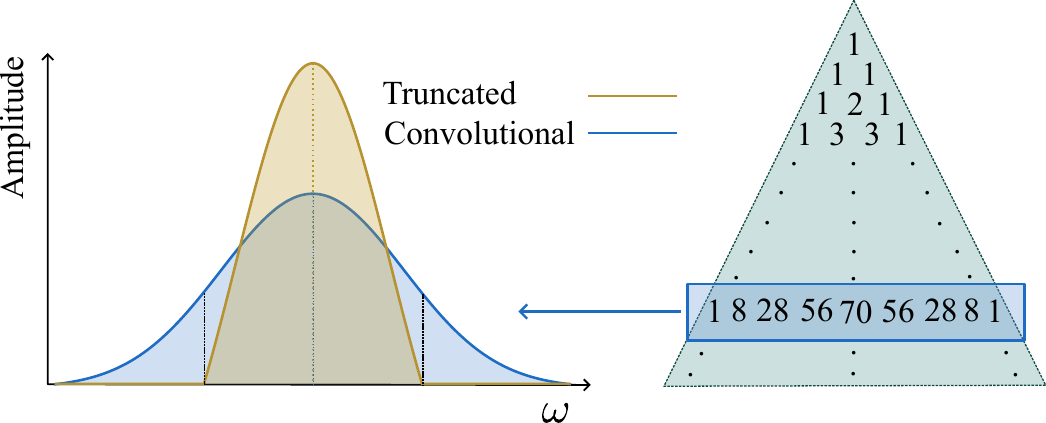}}
\caption{\textbf{Convolutional and Truncated Sampling Strategies used in Numerical Experiments (Section~\ref{subsec:numerics})} Convolutional sampling (Green) inspired by the distribution of Fourier coefficients in QNNs~\cite{barthe2024gradients} comes from the rows of Pascal triangle. Truncated sampling (Golden) is sets the probability of the higher half of the frequencies to zero and for the lower half is proportional to the convolutional distribution.}
\label{fig:sampling}
\end{center}
\vskip -0.2in
\end{figure}

Fig.~\ref{fig:sampling comp} shows a comparison between these three distributions. 
The uniform distribution fails, as predicted by Ref.~\cite{sweke2023potential}. Our proposed method of truncated-convolutional sampling method performs slightly better than the convolutional method on this data. 
  The reason is that the number of frequencies in the frequency support $\Omega$ increases exponentially with $d$ and so the support of the truncated distribution is exponentially smaller than the full support. Consequently, if the dataset consists of low frequencies, the truncated support has a much higher chance of sampling the relevant frequencies. This implies that a dataset that can not be represented only by low frequencies is a better candidate to establish the need for quantum computers in classification (See App.~\ref{app:truncated}). 
  
Fig.~\ref{fig:QSVM vs RFF}, compares the performance of truncated sampling with QSVM. RFF with truncated sampling can outperform the QSVM with different levels of shot noise. We see that with small shot counts, as expected, much fewer RFF samples are needed to achieve a better performance. Intriguingly, Fig.~\ref{fig:QSVM vs RFF} shows a better performance for QSVM with 100 shots than the exact kernel evaluation. This can be attributed to low levels of shot noise acting like regularization in classical models preventing overfitting. A more detailed discussion on the regularization effect of the shot noise in QML models is provided in  App.~\ref{app:shotnoise}

Any claim of dequantization as defined in Def.~\ref{def:rff_dequantization} requires a scaling analysis over the dimension of the input data. Fig.~\ref{fig:RFF Complexity} shows the minimum number of frequency samples for which RFF reaches the estimated risk of 0.15 with respect to the dimension of input data. The fitted line shows a polynomial growth of the resources required by the RFF algorithm, and hence indicates the efficiency of this algorithm on this dataset. 

\begin{figure}
\begin{center}
\centerline{\includegraphics[width=\columnwidth]{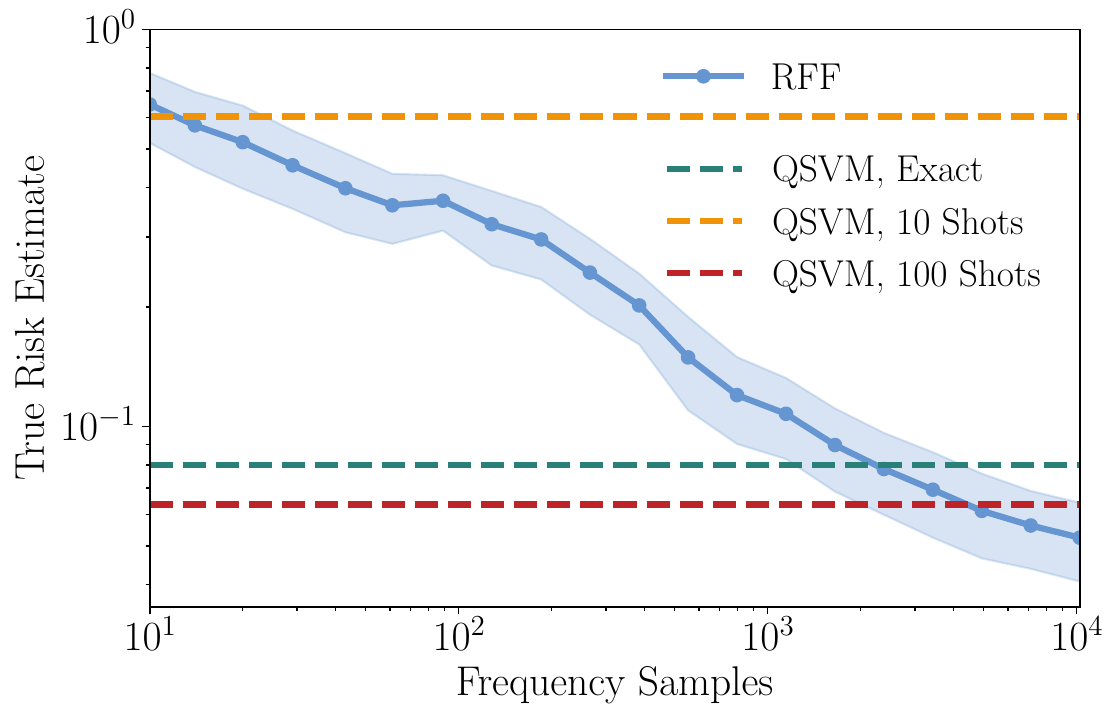}}
\caption{\textbf{RFF-SVM vs QSVM.} The true risk of SVM against the number of frequency samples $D$ for the RFF-SVM using a truncated-convolutional sampling distribution. The dashed lines are the performance of a 5-layer 16-qubit QK, with different levels of shot noise. More details can be found in App.~\ref{app:dataset}.}
\label{fig:QSVM vs RFF}
\end{center}
\vskip -0.2in
\end{figure}
\begin{figure}
\begin{center}
\centerline{\includegraphics[width=\columnwidth]{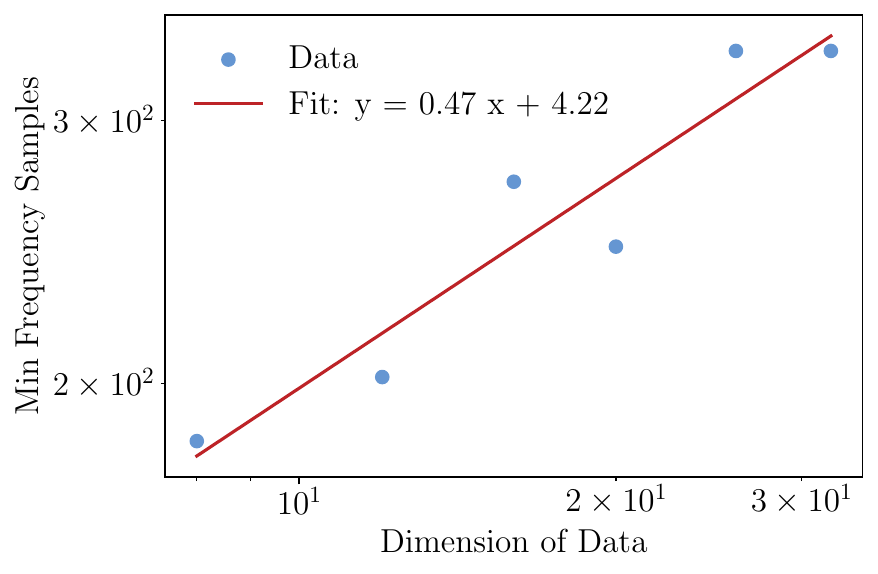}}
\caption{\textbf{Complexity of RFF.} Minimum number of frequency samples required by the RFF algorithm to reach an estimated risk of 0.15 with respect to the dimension of the input data. Different data dimensions are achieved by performing PCA on 64 dimensional data. 1000 training points and 200 test points are used. The performance is averaged over 20 runs of the algorithm. The slope of the fitted line is shown on the legend, and has a regression score of 0.89.}
\label{fig:RFF Complexity}
\end{center}
\vskip -0.2in
\end{figure}
\section{Conclusion}
In this work, we examine the question of whether RFF algorithms can dequantize supervised QML models.
Existing work has focused solely on the dequantization of QNN regression, while we extend it to QK methods and SVM classification problems, as summarized in Table~\ref{tab:main}. 
Specifically, our Props.~\ref{prop:1},~\ref{prop:loose},~\ref{prop:tight} provide conditions under which classical algorithms (i.e., RFF) could be comparable to QML  models in terms of true risks. Overall our results indicate two conditions for dequantizable quantum machine learning tasks namely  \textit{alignment} and \textit{concentration}. These results help chisel out the types of problems for which we can hope to find potential quantum advantages for practical learning tasks.
The Alignment condition implies that the Fourier spectrum of either the QK (for QK methods) or the QNN optimal decision function (for QNNs) determines the optimal RFF sampling distribution. Therefore, if this optimal distribution is efficient to find or sample from, dequantization can be possible. The concentration condition states that this optimal sampling distribution should be concentrated in some sense to ensure dequantization. For regression tasks this concentration condition manifests as a lower bound on the infinity norm of the distribution and, for QSVM, it puts an upper bound on 1/2-Renyi entropy of the distribution. 

While we propose analytical bounds and conditions on dequantization of various quantum models using information from their Fourier transform, our ability to verify these conditions is limited by the computational cost of the Fourier transform. In other words, a full knowledge of the optimal RFF sampling distribution is not practically possible for large circuits. 
Whether this information can be obtained without the full Fourier transform remains an open question. Nevertheless, even if this optimal distribution can be obtained efficiently, its sampling complexity could still hinder dequantization.

To circumvent this practical issue, we use simple distributions that are easy to sample from. Using a task independent distribution bypasses the need for the Fourier analysis of the quantum circuit. Moreover, the efficiency of the RFF sampling subroutine is ensured by choosing these simple distributions to be separable and easy to sample from. Using practical compressed datasets from high-energy physics, we numerically show that these simplified distributions can outperform QSVM.
While this theoretical work compares noiseless quantum models to RFF, numerical experiments show that a low level of shot noise may increase generalization performance. The interplay between shot noise and RFF dequantization remains unexplored.

An exciting direction for future work is extending our results to other types of data encoding. 
This study focuses on Hamiltonian encoding, a key ingredient to express QML models as finite Fourier series.
However, alternative encodings, such as amplitude encoding, result in a continuous frequency space and thus fall outside the scope of our methods.
Notably, a well-known QML model with a provable quantum advantage~\cite{liu2021rigorous} relies on a QK using non-Hamiltonian encoding.
Hence, exploring the dequantizability of QML models across various encoding strategies could offer deeper insight into power of RFF for dequantization and potential of quantum advantage.  

\medskip
\section*{Acknowledgments}
Authors are grateful to Sofiene Jerbi, Daniel Stilck Franca and Elies Gil-Fuster for their valuable feedback on this work. We appreciate fruitful discussions with Hela Mhiri and Slimane Thabet.
MS, AB and MG are supported by CERN through the CERN Quantum Technology Initiative.
AB is supported by the quantum computing for earth observation (QC4EO) initiative of ESA $\phi$-lab, partially funded under contract 4000135723/21/I-DT-lr, in the FutureEO program. ZH acknowledges support from the Sandoz Family Foundation-Monique de Meuron program for Academic Promotion. 
\bibliography{quantum.bib,otherbib.bib}
\onecolumngrid
\setcounter{theorem}{0}
\clearpage

\appendix
\part{Appendix}
\parttoc 

\section{Related Work} \label{sec:relatedworks}
\paragraph*{RFF.} Ref.~\cite{rahimi2007random} first introduced RFF as a method for approximation of shift-invariant kernels and provided probabilistic bounds on the maximum point-wise approximation error of shift-invariant kernels. Later, the authors gave bounds on the true risk of random Fourier feature map for learning tasks with Lipschitz loss function~\cite{rahimi2008weighted}. 
Ref.~\cite{sutherland2015error} bounds the difference between decision functions achieved by the original and RFF-approximated kernels for tasks such as ridge regression and SVM using the proximity of kernel matrices. Ref.~\cite{rudi2017generalization,avron2017random} directly bounds the true risk of RFF for regression. Ref.~\cite{li2019towards} offer a framework that works for both regression and SVM. Ref.~\cite{liu2021random}, provides an exhaustive review on the development of RFF methods. 
 
\paragraph*{Dequantization via RFF.} In the context of QML, RFF was the first method used to dequantize QNN regression by Ref.~\cite{landman2022classicallycorrect}. Their analytical results use the bounds of Ref.~\cite{sutherland2015error}. The authors also suggest practical sampling strategies for RFF. Ref.~\cite{sweke2023potential} extends the results to non-universal QNNs building up from Ref.~\cite{rudi2017generalization}. The idea of random features has been used for quantum kernels but in a slightly different context. Ref.~\cite{nakaji2022deterministic} uses random features to approximate a class of QKs called projected QKs~\cite{huang2021power,jerbi2021quantum}. Ref.~\cite{gil2024expressivity} uses RFF to approximate classical kernels with QKs.

Ref.~\cite{sweke2025kernel} explores the possibility of efficient and exact evaluation of QNN kernels with matrix product state reweighting. It is shown that when the eigen spectrum of the kernel forms a matrix product state, it can be evaluated efficiently using tensors without the need for random feature sampling. A generalization comparison between this tensor network dequantization method and RFF-dequantization is left for future work.  

\section{Background and Framework}
\subsection{Useful Notions from Kernel Theory} \label{app:Defs}
This section is largely based on Chapter 4 of Ref.~\cite{steinwart2008support}. 
\begin{definition}[Kernel Function]
Let $\XC$ be a non-empty set, then a function $k:\XC \times \XC \rightarrow \mathbb{F}$ is called a kernel on $\XC$ if there exists a $\mathbb{F}$-Hilbert space $\HC$ and a map $\Phi: \XC \rightarrow \HC$ such that $\forall x, x' \in \XC$,
\be 
k(x,x') = \expval{\Phi(x'),\Phi(x)} 
\ee
We call $\Phi$ and $\HC$ a feature map and a feature space, respectively. $\mathbb{F}$ can be the set of real or complex numbers.
\end{definition}

\begin{definition}[Reproducing Kernel Hilbert Space RKHS] \label{def:RKHS}
Let $k$ be a kernel on $\XC$ and $\HC_0$ be a Hilbert space for the feature map $\phi_0$ associated with the kernel. Then, the Hilbert space:
\be
\HC := \{f |  \ \exists \nu \in \HC_0  \ \text{with} \  f(x) = \expval{\nu, \phi_0(x)}_{\HC_0} \forall x \in \XC  \} 
\ee
with norm: 
\be 
\norm{f}_{k} := \inf\{\norm{\nu}_{\HC_0} | f = \expval{\nu, \phi_0} \}
\ee
is called the RKHS of $k$ and its norm is called the RKHS norm.
\end{definition}
\begin{definition}[Kernel Integral Operator] \label{def:appKIO}
Given a kernel $k: \mathcal{X} \times \mathcal{X} \rightarrow \mathbb{R}$ and a probability distribution $\pi$ over $\mathcal{X}$, we consider the space of square-integrable functions with respect to $\pi$:
$$
L^2\left(\mathcal{X}, \pi\right)=\left\{f \in \mathbb{R}^{\mathcal{X}} \text { such that } \int_{\mathcal{X}}|f(x)|^2 \mathrm{~d} \pi(x)<\infty\right\} .
$$
The kernel integral operator $T_k: L^2\left(\mathcal{X}, \pi\right) \rightarrow L^2\left(\mathcal{X}, \pi\right)$ is then defined as
$$
\left(T_k g\right)(x)=\int_{\mathcal{X}} k\left(x, x^{\prime}\right) g\left(x^{\prime}\right) \mathrm{d} \pi\left(x^{\prime}\right)
$$
for all $g \in L^2\left(\mathcal{X},\pi\right)$.
\end{definition}

Bochner's theorem provides some insights into the spectrum of shift-invariant kernels, which are targets of RFF.
\begin{theorem}[Bochner's Theorem~\cite{rudin1990fourier}] \label{th:boch}
    A continuous function $k(\delta)$  on $\Rbb^d$ is a shift-invariant kernel if and only if it is the Fourier transform of a non-negative measure. That is, every shift-invariant kernel $k$ can be expressed as
\be \label{eq:bochner}
k(\xv-\yv) = \int_{\Omega} p(\omv) e^{j\omv\cdot (\xv-\yv) } d\omega
\ee

where $p$ is unique and non-negative. Moreover, by assuming  $k(0) =1$ without loss of generality, we have $\int_{\Omega} p(\omv) d\omv=1$ and thus $p$ is a probability density function. 
\end{theorem}
Thm.~\ref{th:boch} shows a one-to-one correspondence between probability distributions and shift-invariant kernels. That is, the Fourier transform of every normalized shift-invariant kernel is a probability distribution.
In this work, we consider quantum kernels {(QKs)}, which are non-stationary in general.
To do so, we focus on the following theorem.
\begin{theorem}[Yaglom 1987~\cite{yaglom1987correlation,samo2015generalized,remes2017non}]\label{th:yaglom}
A complex-valued continuous bounded function $k$ on $\Rbb^d \times \Rbb^d$ is the covariance function of a mean
square continuous complex-valued random process on $\Rbb^d$. 
Namely, a non-stationary continuous bounded kernel $k$ is positive semi-definite if and only if it can be represented as 
\be \label{eq:yaglom}
k(\xv,\yv) = \int_{\Rbb^d \times \Rbb^d} e^{i(\omv_1 \cdot \xv- \omv_2 \cdot \yv)} F(\omv_1, \omv_2)  d\omv_1 d\omv_2
\ee
where $F$ is a positive semi-definite function.
\end{theorem}
In analogy to Eq.~\eqref{eq:bochner}, Eq.~\eqref{eq:yaglom} shows that $F(\omv_1,\omv_2)$, is the Fourier transform of $k(\xv,-\yv)$. 
Note that if $k$ is real, as is the case throughout this thesis, $F(-\omv_1, -\omv_2) = F^*(\omv_1, \omv_2)$.
Additionally, if $k$ is symmetric, i.e., $k(\xv,\yv) = k(\yv,\xv)$, then $F$ is Hermitian and $F(\omv_1, \omv_2) = F^*(\omv_2, \omv_1)$. 
If $F$ is diagonal, i.e. $F(\omv_1,\omv_2) = 0 \ \forall \omv_1\neq \omv_2 $.
all the elements are non-negative and $k$ is shift invariant according to~\eqref{eq:yaglom}, with which we can reproduce Bochner's theorem (Thm.~\ref{th:boch}). Note the technical detail that Thm.~\ref{th:yaglom} holds for a class of kernels called \textit{harmonizable} kernels. However, except some highly contrived examples, most of complex-valued continuous bounded kernels are harmonizable. For more details refer to App.~D of Ref.~\cite{samo2015generalized}.  

\subsection{Some Useful Lemmas} 
\begin{lemma}[Cholesky Decomposition]\label{lem:chol}
    A positive definite (semi-definite) Hermitian matrix $A$ can be written as $A = LL^\dagger$ where $L$ is a lower triangular matrix with positive (non-negative) diagonal elements. 
\end{lemma}
\begin{corollary}[Reverse Cholesky Decomposition ] \label{col:choleskyupper}
    A positive definite (semi-definite) Hermitian matrix $A$ can be written as $A = UU^\dagger$ where $U$ is an upper triangular matrix with positive (non-negative) diagonal elements.
\end{corollary}
\begin{proof}
   Consider the permutation matrix $R$ with the same dimension as $A$ ($n\times n$) such that $\forall i, \ R_{n,n-i} = 1$ and all of its other elements are $0$. 
    Together with Lem.~\ref{lem:chol}, we get $RAR = LL^\dagger$.
    Moreover, by inserting $RR^{\dagger}=\mathbb{1}$ into $LL^{\dagger}$, we have $A = (RLR)(RLR)^\dagger = UU^\dagger$.
\end{proof}

\begin{corollary}\label{col:cholesky}
    A positive definite (semi-definite) Hermitian matrix $A$ can be written as $A = U P U^\dagger$ where $U$ is an upper triangular matrix with normal columns and $P$ is diagonal with positive (non-negative) elements. Moreover, $\Tr[A] = \Tr[P]$.
\end{corollary}
\begin{proof}
    We are given a positive semi-definite matrix $A$. According to Cor.~\ref{col:choleskyupper}, there exists an upper triangular matrix $V$ such that $A = VV^\dagger$. Then we can write $A = \sum_i \bm{v}_i \bm{v}_i^\dagger$ where $\bm{v}_i$'s columns of $V$. Take $p_i = \norm{\bm{v}_i}_2^2$, then we have $A = \sum_i \uv_i p_i \uv_i^\dagger$ where $\uv_i = \bm{v}_i/\norm{\bm{v}_i}_2$. Define $P$ to be diagonal with $p_i$'s as its main diagonal and $U$ to be a matrix with normal vectors $\uv_i$'s as its columns. Then, we have $A = UPU^\dagger$. Moreover, $\Tr[A] = \Tr[\sum_i \bm{v}_i \bm{v}_i^\dagger] =   \sum_i \Tr[  \bm{v}_i^\dagger \bm{v}_i] = \sum_i p_i = \Tr[P]$. 
\end{proof}

\subsection{Quantum Kernels and Quantum Neural Networks} \label{app:QK}
\begin{definition}[Fidelity Quantum Kernel] \label{def:appqk}
    Given an $n$-qubit encoding unitary $U(\xv)$ where $\xv \in \XC \subset \Rbb^d$, we define an encoded state as $\rho(\xv) = U(\xv) \ketbra{0}{0} U^\dagger (\xv)$. Then the fidelity quantum kernel $k_Q$ is a kernel on $\XC$ with: 
    \be  \label{eq:app_qk_def}
        k_Q(\xv,\yv) = \Tr[\rho(\xv) \rho(\yv)].
    \ee
\end{definition}
In this work, we consider the \textit{Hamiltonian } encoding (Assumption~\ref{ass:Hencoding}), i.e., quantum circuits whose gate in the $l$-th layer that encodes $j$-th element of data $\xv$ has the form $V_l(\xv) = \exp (iH^{(j)}_l\xv_j)$ for a Hermitian matrix $H^{(j)}_l$.
Concretely, the encoding unitary operator is represented as
\be 
U(\bm{x}) = \prod_{l=1}^L \prod_{j=1}^d S_l^{(j)}e^{iH_l^{(j)}x_j}T_l^{(j)}, 
\ee
where $S_l^{(j)}$ and $T_l^{(j)}$ are fixed unitaries (i.e., non-variational unitary operators) and $L$ is the number of layers.

We also denote the $i$'th eigenvalues of $H_l^{(j)}$ by $\lambda^{(l,j)}_i$ and define
\be \label{eq:tempor}
{\Lambda}^{(j)}_{\mathbf{i}} = \lambda^{(1,j)}_{i_1} + \lambda^{(2,j)}_{i_2} + \cdots + \lambda^{(L,j)}_{i_{L}}
\ee 
for a vector index $\mathbf{i} = (i_1,\cdots, i_{L})$. 
For fidelity QKs using a Hamiltonian encoding the following Lemma holds.

\begin{lemma} \label{lem:kernelfreqs}\label{lem:appQKYaglom}
    Given a Hamiltonian encoding, a quantum kernel $k_Q$ in Eq.~\eqref{eq:app_qk_def} can be expressed as 
\be \label{eq:qk_QKYaglom}
k_Q(\bm{x},\bm{y}) = \sum_{\bm{\omega}, \bm{\nu} \in {\Omega}} F_{\bm{\omega}\bm{\nu}} e^{i(\bm{\omega} \cdot\bm{x} - \bm{\nu}\cdot\bm{y})},
\ee 
where $F$ is positive semi-definite ($F_{\bm{\omega}\bm{\nu}} = F^*_{\bm{\nu}\bm{\omega}}$). Moreover, $F_{-\bm{\omega}, -\bm{\nu}} = F^*_{\bm{\omega}\bm{\nu}}$ since $k_Q$ is real-valued. Here, $\Omega = \Omega_1 \times \Omega_2 \times \cdots \Omega_d$ is defined as 
\be
\Omega_j = \{ {\Lambda}^{(j)}_{\mathbf{s}} - {\Lambda}^{(j)}_{\mathbf{t}} | \ \forall \mathbf{s},\mathbf{t}  \}.
\ee
where $\Lambda^{(j)}_{\mathbf{s}}$ is defined in Eq.~\eqref{eq:tempor}. 
\end{lemma}
\begin{proof}
    The proof for one-dimensional input $x$ is provided in Ref.~\cite{schuld2021supervised}.
    This result was extended to multi-dimensional inputs by taking the same steps for each dimension as done in Ref.~\cite{schuld2021effect}. 
    Since $k_Q$ can be written as a finite Fourier series,
    Thm.~\ref{th:yaglom} can be applied to prove the 
    positive semi-definiteness of $F$.
    The Hermitian property of $F$ is derived from the symmetry of the fidelity quantum kernels i.e. $k_Q(\bm{x},\bm{y})= k_Q(\bm{y},\bm{x})$.
    To be more precise, we have
    \be
k_Q(\bm{y},\bm{x}) = \sum_{\bm{\omega}, \bm{\nu} \in {\Omega}} F_{\bm{\omega}\bm{\nu}} e^{i(\bm{\omega}\cdot \bm{y} - \bm{\nu}\cdot\bm{x})} \overset{(1)}{=} \sum_{\bm{\omega}, \bm{\nu} \in {\Omega}} F_{\bm{\nu}\bm{\omega}} e^{-i(\bm{\omega}\cdot\bm{x} - \bm{\nu}\cdot\bm{y})} \overset{(2)}{=} \sum_{\bm{\omega}, \bm{\nu} \in {\Omega}} F^*_{\bm{\nu}\bm{\omega}} e^{i(\bm{\omega}\cdot\bm{x} - \bm{\nu}\cdot\bm{y})} = k_Q(\bm{x},\bm{y}) = \sum_{\bm{\omega}, \bm{\nu} \in {\Omega}} F_{\bm{\omega}\bm{\nu}} e^{i(\bm{\omega}\cdot\bm{x} - \bm{\nu}\cdot\bm{y})} \nonumber,
    \ee
where we swapped the dummy variables $\omv, \bm{\nu}$ in step (1) and in step (2) we used the fact that $k_Q$ is real-valued and equal to its conjugate.
By comparing the first and the last terms, we conclude that we have $F_{\bm{\omega} \bm{\nu}} = F^*_{\bm{\nu}\bm{\omega}}$ for all $\bm{x},\bm{y}$, which completes the proof.
\end{proof}

Note that the frequency support of encoding, $\Omega$, is symmetric around zero and always contains $0$, meaning  $\forall \omega \in \Omega, \ -\omega \in \Omega$. 
Thus, we can define a set $\Omega_+$ such that $\Omega = \Omega_+ \cup -\Omega_+ \cup \{0\}$ and $\Omega_+ \cap -\Omega_+ =\varnothing $. We define $\Omega_+$ more formally. 
\begin{definition}[Positive Frequency Support]\label{def:PositiveFreqSup}
  For any quantum model (QNN or QK) that can be written in the form of a discrete Fourier sum (Eqs.~\eqref{eq:pqcft} and~\eqref{eq:qkft}) on a set $\Omega$, we call any set $\Omega_{+}$ a \textit{positive frequency support} of the model if $\Omega = \Omega_+ \cup -\Omega_+ \cup \{0\}$ and $\Omega_+ \cap -\Omega_+ =\varnothing $.  
\end{definition}

\begin{corollary} \label{prop:qk}
    For any Hamiltonian encoding (Assumption~\ref{ass:Hencoding}) and the fidelity quantum kernel $k_{Q}$ in Eq.~\eqref{eq:qk_QKYaglom} with frequency support $\Omega$, there exist $|\Omega| \times |\Omega|$ lower triangular matrices $L$ and $\Lh$ such that: 
    \be \label{eq:heqk_feature}
        k(\xv,\yv) = z^\dagger(\yv) F z(\xv) = \zh^\dagger(\yv) \hat{F} \zh(\xv) 
    \ee
    where $F = LL^\dagger,\hat{F} = \Lh \Lh^\dagger$ are positive semi-definite Hermitian matrices, and where 
    \be \label{eq:zdef}
        z(\xv) \coloneqq \biggl ( 1 , e^{i\omv_1\cdot \xv}, e^{-i\omv_1 \cdot \xv}, e^{i\omv_2 \cdot\xv}, e^{-i\omv_2 \cdot \xv}, \cdots, e^{i\omv_{|\Omega_+|}\cdot\xv}, e^{-i\omv_{|\Omega_+|} \cdot\xv}      \biggr)^T_{\omv_i \in \Omega_+}
    \ee
    \be
        \zh(\xv) \coloneqq \sqrt{2} \biggl ( \frac{1}{\sqrt{2}} , \cos(\omv_1 \cdot\xv), \sin(\omv_1 \cdot\xv), \cos(\omv_2 \cdot\xv), \sin(\omv_2 \cdot\xv), \cdots, \cos(\omv_{|\Omega_+|} \cdot\xv) , \sin(\omv_{|\Omega_+|} \cdot\xv)      \biggr)^T_{\omv_i \in \Omega_+}
    \ee
\end{corollary}

\begin{proof}
    First, the fidelity quantum kernel $k_{Q}$ is bounded; $0\le k_{Q}(\xv,\yv)\le \sqrt{k_{Q}(\xv,\xv)k_{Q}(\yv,\yv)}\le 1$.
    Also, the Hamiltonian encoding is continuous (Assumption~\ref{ass:Hencoding}), and $k_{Q}$ has frequencies in $\Omega = \Omega_+ \cup -\Omega_+ \cup \{0\}$ from Lem~\ref{lem:kernelfreqs}.
    Thus, using Thm.~\ref{th:yaglom}, there exists a positive semi-definite matrix $F$ such that:
        \be
            k(\xv,\yv) =  \sum_{l,p}^{|\Omega|} e^{i(\omv_l\cdot \xv-\omv_p \cdot\yv)} F_{l,p}  = z^\dagger(\yv) F z(\xv) = z^\dagger(\yv) LL^\dagger z(\xv)
        \ee
        with $\omv_i \in \Omega$.
        In the second equality, we use the vectorized form of the features to define $z$ in Eq.~\eqref{eq:heqk_feature} and the fact that $\Om = \Omega_+ \cup -\Omega_+ \cup \{0\}$.
        Since $F$ is positive semi-definite according to Lem.~\ref{lem:chol}, there exists a lower triangular matrix with non-negative diagonal elements such that $F = LL^\dagger$ and the first part of the proposition is proved. 

        We remark that the map $S$ that brings $\zh$ to $z$ is isometric and reversible and thus a unitary. 
        Therefore, we have $k(\xv,\yv)=  \zh^\dagger(\yv) S^\dagger F S \zh(\xv)$.
        Because of the unitarity of $S$, $\hat{F} = S^\dagger F S$ is still positive semi-definite, which completes the proof. 
\end{proof}

\begin{lemma}\label{lem:intz}
    For any set $\Omega$, the associated map $z$ defined in Eq.~\eqref{eq:zdef} has the following property for $\lim_{T\rightarrow \infty}\frac{1}{T}\int_{-T/2}^{T/2}  z(x) z^\dagger(x) dx =  \mathbb{1}$ where $\mathbb{1}$ is the identity matrix. 
\end{lemma}

\begin{proof}
By substituting $z(\xv)$ in Eq.~\eqref{eq:heqk_feature} into the quantity, we get
\be
    \lim_{T\rightarrow \infty}\frac{1}{T}\int_{-T/2}^{T/2} [z(\xv) z^\dagger(\xv)]_{m,n} dx =\lim_{T\rightarrow \infty}\frac{1}{T}\int_{-T/2}^{T/2} e^{i(\omv_n-\omv_m)\cdot\xv} d\xv =  \delta_{m,n}.
\ee
\end{proof}

Lem.~\ref{lem:intz} yields the following property of the Fourier transform of fidelity quantum kernels with Hamiltonian encoding, which we give below.
\begin{corollary}\label{prop:Bnorm}
      Under the setting in Lem.~\ref{lem:kernelfreqs} and Cor.~\ref{prop:qk}, we have $\norm{L}_2^2 = \Tr[F]=1$ where $F$ is the Fourier transform of the QK and $L$ is a lower triangular matrix such that $F = LL^\dagger$.  
\end{corollary}
\begin{proof}
    From Lem.~\ref{lem:kernelfreqs} we have, $\norm{L}_2^2 = \Tr(LL^\dagger) = \Tr[F]$. 
    Then, we utilize Lem.~\ref{lem:kernelfreqs} and Lem.~\ref{lem:intz} to derive
    \begin{align}
     \lim_{T\rightarrow \infty}\frac{1}{T}\int_{-T/2}^{T/2} k(\xv,\xv) d\xv &= \lim_{T\rightarrow \infty}\frac{1}{T}\int_{-T/2}^{T/2} z^\dagger(\xv) F z(\xv) d\xv =  \lim_{T\rightarrow \infty}\frac{1}{T}\int_{-T/2}^{T/2} \Tr (F z(\xv)z^\dagger(\xv)) d\xv \nonumber \\ & = \Tr (F \lim_{T\rightarrow \infty}\frac{1}{T}\int_{-T/2}^{T/2} z(\xv)z^\dagger(\xv) d\xv) =  \Tr[F]. 
    \end{align}
    On the other hand, every fidelity quantum kernel with pure encoded state outputs one for the same inputs by definition, i.e., $k(\xv,\xv)= 1$, and $ \lim_{T\rightarrow \infty}\frac{1}{T}\int_{-T/2}^{T/2} k(\xv,\xv) d\xv = \Tr[F] = 1$.
    Thus we obtain $\norm{L}_2^2 =\Tr[F] = 1$. 
\end{proof}

\subsection{Support Vector Machines}\label{app:SVM}
We here consider the framework of statistical learning theory.
Suppose we have a set of i.i.d. samples $\MC = \{\xv_i,y_i\}_{i=1}^m$ drawn from some unknown distribution $P:\XC \times \YC \rightarrow \Rbb, \ \XC \subset \Rbb^d$, and a set of functions $\FC$, called \textit{hypothesis class}, mapping $\XC$ to $\YC$.
Also, let $\LC:\YC \times \YC \rightarrow [0,\infty)$ be a positive function. 
The main goal in statistical learning theory is to find a function in the hypothesis class $f \in \FC$ that minimizes the \textit{true risk} defined as 
\be
R[f] := \Ebb_{(X,Y) \sim P}[\LC(f(X),Y)].
\ee
However, since $P$ is unknown and thus the true risk is inaccessible, we usually optimize the \textit{empirical risk}, which is the sample mean of the true risk using dataset $\MC$, i.e.,
\be
\hat R[f] = \frac{1}{m} \sum_{i=1}^m \LC(f(\xv_i),y_i).
\ee
Some learning algorithms map the input data to a Hilbert space $\HC_0$ called \textit{feature space} using a \textit{feature map} $\phi: \XC \rightarrow \HC_0$, so that the structure of data can be captured better. In this case, we assume that $\FC$ is a subset of some RKHS associated with reproducing kernel $k = \expval{\phv(\xv_1), \phv(\xv_2)}$ i.e. for every function $f$ in $\FC$ there exists $\wv \in \HC_0$ such that $f(\xv)=\expval{\bm{w},\phv(\xv)}$.

For a SVM classification problem, we minimize the hinge loss $\LC(y_1,y_2) = \max\{0, 1-y_1y_2\}$ with the binary targets $y_1$ and $y_2$. 
To avoid overfitting, a regularization term can be added to the empirical risk. 
More concretely, if we consider the RKHS $\HC_k$ of a kernel $k$ as the hypothesis class, the optimization problem is expressed as 
\begin{align*}
\min_{f\in\FC} \quad & \frac{\lambda}{2} \norm{f}^2_{k} +  \frac{1}{m} \sum_{i=1}^m\max\{0, 1-y_if(\xv_i)\}, 
\end{align*}
with $\lambda\geq0$ being the so-called regularization strength and $\norm{.}_k$ denoting the RKHS norm. 

As $f$ is in the RKHS of $k$ by definition, it can also be written as $\expval{\bm{w},\phv(\xv)}$ for some $\bm{w} \in \HC_0$ (See Def.~\ref{def:RKHS}). 
Thus, we can re-express the above optimization formulation as 
\begin{align}\label{eq:hinge_pb}
\min_{\bm{w} \ s.t \  \expval{\bm{w},\phv} \in \FC} \quad & \frac{\lambda}{2} \norm{\bm{w}}^2_{\HC_0} +  \frac{1}{m} \sum_{i=1}^m\max\{0, 1-y_i \expval{\bm{w},\phv(\xv)}\}.
\end{align}

The above formulation requires access to the feature map. This is not possible for every kernel e.g. some kernels have infinite dimensional feature maps. When the feature map is not efficiently accessible, the problem can be solved in its dual form which only requires evaluation of the kernel function $k = \expval{\phv(\xv_1), \phv(\xv_2)}$. The dual problem solves the following optimization:
\begin{align} \label{eq:dual}
\max_{\alphav \in [0,\lambda^{-1}]^m} \quad & \sum_{i=1}^m \alphav_i -  \frac{1}{2} \sum_{i,j=1}^m \alphav_i \alphav_j y_iy_j k(\xv_i,\xv_j) .
\end{align} 
In other words, the dual problem has the advantage that we do not have to compute the feature maps explicitly, once the kernel for every pair of training data is computed (i.e., these values are stored in the Gram matrix $K$ with $K_{i,j} = k(\xv_i, \xv_j)$).
This is known as the ``kernel trick'' and enables the use of even infinite-dimensional feature maps. 

The main idea of SVM classification as shown in Fig.~\ref{fig:norm}, is to find a separating hyperplane that can maximize the distance between the samples of both classes, which is called \textit{margin}.
When misclassification is not allowed, i.e., $\lambda\approx0$ in Eq.~\eqref{eq:hinge_pb}, it is called hard margin SVM.
On the other hand, a classification scheme with large $\lambda$ allows for some trade-off between misclassification error and the margin and is called soft-margin SVM. 
Moreover, by adding another optimization parameter $b$ as an offset to $f$ so that the decision function is written as $f+b$, we can incorporate a single class SVM problem for anomaly detection into the same optimization problem as above.
We note that the decision boundary of SVM is given by the sign of the function $f$, as shown below.

\begin{definition}[Decision boundary of SVM]
The decision boundary of SVM for binary classification is a function $f:\XC \rightarrow \Rbb$ such that the predicted label for a test sample $\xv_t$ is computed as $y_t = \operatorname{sign}(f(\xv_t))$.
\end{definition}

Finally, we note that other loss functions can be used for different tasks.
For example, the squared error $\LC(y,y') = (y-y')^2$ is mostly used in regression tasks; in case the regularization is included with the loss, it is called a ridge regression problem.
When kernel methods are used for ridge regression, the model is classed kernel ridge regression.

\subsection{Random Fourier Features} \label{app:RFF}

Random Fourier features~\cite{rahimi2007random} is a method to approximate a shift-invariant kernel. The motivation comes from the fact that, although the kernel trick makes it possible to utilize infinite dimensional feature spaces, the computational complexity of solving the dual problem using  $m\times m$ Gram matrix is $\OC(m^2)$.
Here, $m$ is the number of training data points. 
This indicates that the computation becomes infeasible with increase in data size. 
A circumventing approach to this problem is to approximate the kernel $k$ using a $D$ dimensional feature space and then solve the primal problem in the space; this is useful for $m<D$ as the complexity is $\OC(mD)$. 

According to Bochner's theorem (Thm.~\ref{th:boch}), for every bounded, normalized, continuous and shift-invariant kernel we have 
\begin{align} \label{eq:2}
k(\xv-\yv) \overset{(1)}{=}& \int_{\Omega} p(\omv) \cos{ ( \omv \cdot (\xv-\yv))} d\omv = \int_{\Omega} p(\omv) \cos{( \omv \cdot \xv)}\cos{( \omv \cdot \yv)} + 
\sin{( \omv \cdot \xv)}\sin{( \omv \cdot \yv)} d\omv  \nonumber \\ 
=& \Ebb_{\omv \sim p} [\psv(\xv,\omv) \psv(\yv,\omv)] 
\end{align}
where in step (1) we use the fact that $k$ is symmetric to write Eq.~\eqref{eq:bochner} with cosine functions. $\psv(\xv,\omv)= \bigl(\cos(\omv\cdot \xv),  \sin(\omv\cdot \xv) \bigr)^T$ and $p$ is a probability distribution obtained by normalizing the Fourier transform of the kernel. 

RFF algorithm estimates this expectation using samples from the distribution. More precisely, given $D$ i.i.d. samples, $\omv_1, \cdots, \omv_D$,  from the distribution $p$, we can generate the  \textit{random Fourier feature map},
\begin{align} \label{eq:clrff}
\phv_{D}(\xv) = \frac{1}{\sqrt{D}}
\begin{pmatrix}
\cos(\omv_1 \cdot\xv) \\  \sin(\omv_1 \cdot\xv) \\ \vdots \\ \cos(\omv_D \cdot\xv) \\  \sin(\omv_D \cdot\xv) \, 
\end{pmatrix}.
\end{align}
 This method of approximating the kernel is called Random Fourier Features (RFF). 
RFF was first proposed in Ref.~\cite{rahimi2007random} together with a probabilistic bound on the approximation error that depends on the variance of distribution $p$. 
Followed by the original work, tighter bounds are derived in Ref.~\cite{sutherland2015error}. 
Both of these sources prove that one can approximate a kernel with RFF with infinity norm error less than $\epsilon$ with high probability using $D \in \OC(\frac{1}{\epsilon^2} \log (\frac{\sigma_p}{\epsilon}))$ frequency samples where $\sigma_p = \Ebb_p[\|\omv\|^2]$ is the variance of distribution $p$. 

We note that $\phi_{D}(\xv)$ is not the unique way of approximating such kernels. 
Different feature maps such as 
\begin{align}
    \phi'_{D} =& \frac{1}{\sqrt{D}}(e^{j\omv_1\cdot\xv}, \cdots, e^{j\omv_D \cdot \xv})^T \\
    \phi''_{D} = & \frac{1}{\sqrt{D}}(\sqrt{2}\cos(\omv_1 \cdot\xv + \gamma_1), \cdots, \sqrt{2}\cos(\omv_D \cdot\xv + \gamma_D))^T
\end{align}
can be used.
Here, both feature maps are constructed by sampling $\omv_i$'s from distribution $p$ and $\gamma_i$'s from uniform distribution between $0,2\pi$. 

We remark that, if the kernel $k$ is periodic then we can obtain an exact Fourier series representation and thus have discrete frequencies rather than continuous ones.
That is, instead of $\int_{\Omega} p(\omv) d\omv=1$, the periodic kernel has the distribution $p$ satisfying $\sum_{\omv\in\Omega}p(\omv)=1$.
 
In general, the approximation method can not be used for non-stationary kernels. 
As shown in Thm.~\ref{th:yaglom}, non-stationary kernels can be written in a form similar to the shift-invariant kernel.
However, $F(\omega_1, \omega_2)$ in Eq.~\eqref{eq:yaglom} is not guaranteed to be non-negative or even real and thus cannot be taken as a probability distribution in general.

\section{Approximation of Quantum Kernels with Random Features}\label{app:approx}
In this section, we propose a method for approximating fidelity quantum kernels as defined in Def.~\ref{def:appqk}. 
Specifically, we introduce an RFF-based method tailored to approximate a family of non-stationary kernels, which we refer to as \textit{Discrete Spectrum} kernels (Def.~\ref{def:appDisSpecKernel}).
Our proposed methods take as input a real-valued discrete spectrum kernel $k: \XC \times \XC \rightarrow \Rbb$ and return the approximate kernel $s(\xv,\yv) = \Phi(\yv)^\dagger \Phi(\xv)$ with the feature map $\Phi(\xv) \in \Cbb^D$ for the original kernel $k$.

The main idea of our proposal is to express the kernel as the expected value of the inner products of feature maps, as shown in Eq.~\eqref{eq:kp}, and estimate the output using its sample mean.
In this regard, our methods are closely related to the RFF framework (see Sec.~\ref{app:RFF}).
Similar to RFF, we have to obtain the probability distribution from which we sample the frequencies to construct the feature maps.
To do so, we here consider two approaches: (1) Cholesky decomposition (Alg.\ref{alg:ApproxKernel}) and (2) eigenvalue decomposition (Alg.~\ref{alg:ApproxKernelEVD}).
The key distinction between these methods lies in their trade-offs.
Namely, while Alg.\ref{alg:ApproxKernel} has tighter more interpretable error bounds , the bounds for Alg.~\ref{alg:ApproxKernelEVD} converge to the original RFF framework~\cite{rahimi2007random} bounds when we consider shift-invariant kernels. Therefore the bounds for Alg.~\ref{alg:ApproxKernelEVD} can be seen as a generalization of the results of Ref.~\cite{rahimi2007random}.  

We again emphasize that our method applies to a family of non-stationary kernels (not necessarily quantum), but we in this work we focus on the fidelity quantum kernel.

\subsection{Settings \& Definitions}
Before we introduce our proposal, we outline the settings and definitions that will be helpful in this section as well as the subsequent sections. 
Throughout this section, we assume $\XC$ is a compact set in $\Rbb^d$.

\begin{definition}[Discrete Spectrum Kernel] \label{def:appDisSpecKernel}
    A kernel $k: \XC \times \XC \rightarrow \Rbb$ is called a Discrete Spectrum Kernel if its Fourier decomposition can be written as a finite sum: 
    \be 
k(\bm{x},\bm{y}) = \sum_{\bm{\omega}, \bm{\nu} \in {\Omega}} F_{\bm{\omega}\bm{\nu}} e^{i(\bm{\omega}\cdot\bm{x} - \bm{\nu} \cdot \bm{y})} = \bm{z}^\dagger(\bm{y}) F \bm{z}(\bm{x}) ,
\ee 
where $\Omega$ is a finite set of frequencies we call \textit{frequency support} and 
\be \label{eq:appzdef}
        \zv(\xv) = \biggl ( e^{i\omv_1 \cdot \xv} , e^{i\omv_2 \cdot \xv}, e^{i\omv_3 \cdot \xv},\cdots, e^{i\omv_{|\Omega|} \cdot \xv} \biggr)^T .
    \ee
Without loss of generality, we assume that these kernels are normalized, i.e., $k(\xv,\xv) = 1$. 
Note that the Fourier transform matrix $F$, which we call Fourier transform of the kernel, is positive semi-definite and unit-trace. See Thm.~\ref{th:yaglom} for the details of these properties.

An example of a discrete spectrum kernel is a fidelity quantum kernel in Def.~\ref{def:appqk}.

\end{definition}
\begin{corollary}
    A fidelity quantum kernel defined in Def.~\ref{def:appqk},  with Hamiltonian encoding is a discrete spectrum kernel.  
\end{corollary}
\begin{proof}
    Lem.~\ref{lem:appQKYaglom} and Cor.~\ref{prop:qk} demonstrate that the fidelity quantum kernel is a discrete spectrum kernel. 
\end{proof}

For ease of understanding, we consider the following ordering of the frequencies in $\zv(\xv)$. 
\begin{definition}[Ascending Ordering of $\zv$]\label{def:appAscOrd}
    $\zv$ is in ascending ordering if the frequencies in $\zv$ are sorted based on the size of norm. More precisely, $\forall i,j= 1, \cdots$, if $i \leq j$, then $\norm{\omv_i}_2 \leq \norm{\omv_j}_2$, where $\omv_i, \omv_j$ are frequencies related to $i$th and $j$th elements of $\zv$. 
\end{definition}
Such ascending ordering is possible for every kernel without loss of generality.
Thus, we assume this ordering is satisfied in any decomposition of the kernel involving $\zv$ from now on.

Moreover, the positive semi-definiteness and unit-trace property of $F$ motivates us to derive the following probability distributions from it.
\begin{definition}[Diagonal Distribution of Discrete Spectrum Kernels]\label{def:appDiagonalPdf}
Any discrete spectrum  kernel $k:\XC\times \XC \rightarrow \Rbb$ has the Fourier transform $k(\bm{x},\bm{y})  = \bm{z}^\dagger(\bm{y}) F \bm{z}(\bm{x})$ with the assumption on the ascending order of $\zv$ (Def.~\ref{def:appAscOrd}). Diagonal elements of $F$ are positive and sum to one therefore they constitute a probability mass function. We denote this probability mass function with $q: \{1,\cdots,|\Omega|\} \rightarrow \Rbb^+$, and call it the diagonal probability distribution of the kernel. 
\end{definition}
\begin{definition}[Eigenvalue Distribution of Discrete Spectrum Kernels] \label{def:appEVDPdf}
Any discrete spectrum  kernel $k:\XC\times \XC \rightarrow \Rbb$ has the Fourier transform $k(\bm{x},\bm{y})  = \bm{z}^\dagger(\bm{y}) F \bm{z}(\bm{x})$ with the assumption on the ascending order of $\zv$ (Def.~\ref{def:appAscOrd}). Eigenvalues of $F$ are positive and sum to one therefore they constitute a probability mass function. We denote this probability mass function with $v: \{1,\cdots, |\Omega|\} \rightarrow \Rbb^+$, and call it the eigenvalue probability distribution of the kernel. 
\end{definition} 

\begin{definition}[Cholesky Distribution of Discrete Spectrum Kernels] \label{def:appCholeskyPdf}
Any discrete spectrum  kernel $k:\XC\times \XC \rightarrow \Rbb$ has the Fourier transform $k(\bm{x},\bm{y})  = \bm{z}^\dagger(\bm{y}) F \bm{z}(\bm{x})$ with the assumption on the ascending order of $\zv$ (Def.~\ref{def:appAscOrd}). According to Corollary~\ref{col:cholesky}, $F$ can be decomposed as $F = U P U^\dagger$ where $U$ is upper triangular with (2-norm) unit columns and $P$ is a diagonal matrix with positive elements. The diagonal elements of $P$ constitute a probability mass function. We denote this probability mass function with $p: \{1,\cdots, |\Omega|\} \rightarrow \Rbb^+$, and call it the Cholesky probability distribution of the kernel.
\end{definition}
In this section, we used eigenvalue distribution and Cholesky distribution to express non-stationary kernels as expectation values. 
We have previously defined the diagonal distribution in Def.~\ref{def:appDiagonalPdf} in the main text (Def.~\ref{def:FT}) where we call it the kernel distribution. This diagonal distribution plays an important role in the sufficient conditions (e.g. Prop~\ref{prop:loose}) we derive. We recall the definition here to emphasize its distinction from the other two distribution defined in Defs.~\ref{def:appEVDPdf},~\ref{def:appCholeskyPdf}.   

\subsection{Kernel Approximation using Cholesky Decomposition}
In what follows, we approximate a discrete spectrum kernel with its Cholesky decomposition. 

Given a discrete spectrum kernel $k$, we decompose its Fourier transform $F$ as $UPU^\dagger$ as in Def.~\ref{def:appCholeskyPdf}.
Here, we denote the $i$'th column of $U$ with $\bm{u}_i$ and the diagonal elements of $P$, $P_{ii}$ with $p_i$.
Then, the kernel can be written as $k(\bm{x},\bm{y}) = \sum_{i=1}^{|\Omega|} p_i \bm{z}^\dagger(\bm{y}) \uv_i \uv_i^\dagger \zv(\xv)$.
By defining $g(i,\xv) = \uv_i^\dagger \zv(\xv)$, we re-express the kernel as
\be
k(\bm{x},\bm{y}) = \sum_{n=1}^{|\Omega|} p_n g^*(n,\yv) g(n,\xv) = \Ebb_{N \sim p}[g^*(N,\yv) g(N,\xv)] \ . 
\ee
Similar to the RFF approach for shift-invariant kernels, we can express the discrete spectrum kernel as an expected value of the inner product of two functions. 
With this expression, We can now approximate the kernel with the following $D$ dimensional feature map 
\be 
\phv(\xv) = \frac{1}{\sqrt{D}} \biggl ( g(n_1, \xv), \cdots, g(n_D,\xv) \biggr)^T  
\ee
using i.i.d. samples from distribution $p$, $n_1,\cdots, n_D$, and the approximate kernel is given by
\be \label{eq:appSdef}
s(\xv,\yv) = \frac{1}{D} \sum_{i=1}^{D} g^*(n_i,\yv) g(n_i,\xv). 
\ee
We provide our algorithm for the approximation of discrete spectrum kernels in Alg.~\ref{alg:ApproxKernel}. 
\begin{algorithm}[t]
\caption{Random Features Approximation of Quantum Kernels (upper triangular)} \label{alg:ApproxKernel}
\begin{algorithmic} 
\STATE{\textbf{Input:} A quantum kernel $k_Q: \XC \times \XC \rightarrow \Rbb$ with Hamiltonian encoding, An integer $D$ indicating the dimension of the approximated feature vector}
\STATE{\textbf{Output:} A randomized feature map $\phv: \Rbb^d\rightarrow \Cbb^D$ such that $s(\xv,\yv)= \phv(\yv)^\dagger\phv(\xv) \approx k_Q(\xv,\yv)$}
\vspace{5pt}
\noindent \hrule
\vspace{5pt}
\STATE{\textbf{1:} Take the Fourier transform of $k_Q$ and write it as $k_Q(\bm{x},\bm{y})  = \bm{z}^\dagger(\bm{y}) F \bm{z}(\bm{x})$ as defined Def.~\ref{def:appDisSpecKernel} with $\zv$ following ascending ordering (Def.~\ref{def:appAscOrd}). }
\STATE{\textbf{2:} Decompose $F= U P U^\dagger $ where $U$ is an upper triangular matrix with normal columns and $P$ is diagonal. Diagonal elements of $P$ form a probability mass function $p: \{1,\cdots, |\Omega|\} \rightarrow \Rbb^+$ such that $p_i = P_{ii}$.}
\STATE{\textbf{3:} Sample $D$ i.i.d. integer samples $n_1,\cdots, n_D$ from $p$ }
\STATE{\textbf{4:} Construct $g(n_i,\xv) = \uv_{n_i}^\dagger \zv(\xv)$ where $\uv_{n_i}$ is the $n_i$'th column of $U$.}
\STATE{\textbf{5:} Return $\phv(\xv) = \frac{1}{\sqrt{D}} [g(n_1,\xv), \cdots, g(n_D, \xv)]^T$}
\end{algorithmic}
\end{algorithm}
Before we proceed to derive analytical bounds on the approximation method, we state the following Lemma about features $g$. 
\begin{lemma}\label{lem:appg}
   Consider the feature function $g: \{1,2,\cdots,|\Omega|\} \times \XC \rightarrow \Cbb$ defined as $g(n,\xv) = \uv_n^\dagger \zv(\xv)$ with $\zv(\xv)$ of Eq.~\eqref{eq:appzdef} in the ascending order and columns of an upper triangular matrix $\uv_i$'s that is normal for all $i$.
   Then $g$ possesses the following properties:   
    \begin{itemize}
        \item $\forall \xv \in \XC, n \in \{1,2,\cdots,|\Omega|\}, \quad  |g(n,\xv)| \leq \sqrt{n}$
        \item  $\forall \xv \in \XC, n \in \{1,2,\cdots,|\Omega|\}, \quad \norm{\nabla_{\xv}g(n,\xv)}_2 \leq n \norm{\omv_n}_2 $
    \end{itemize}
\end{lemma}
\begin{proof}
    To prove the first property, we start with the upper bound of the feature function;
    \be
        |g(n,\xv)| = \left|\sum_{i=1}^{n} [\uv_n]_i [\zv]_i(\xv)\right| \leq \norm{\uv_n}_2 \sqrt{\left(\sum_{i=1}^{n} | [\zv]_i(\xv)|^2\right)} = \sqrt{n}
    \ee
    where we used Cauchy-Schwartz inequality and the normality of $\uv_n$. Note that the first $n$ elements of $\uv_n$ are non-zero as they come from an upper triangular matrix.
    Next, we bound its gradient, 
    \begin{align}
        \norm{\nabla_{\xv}g(n,\xv)}_2 &= \norm{\sum_{j=1}^{n} i[\uv_n]_j [\zv]_j(\xv) \omv_j}_2 \leq \sum_{j=1}^{n} |[\uv_n]_j| |[\zv]_j(\xv)| \norm{\omv_j}_2 \nonumber\\ 
        &\leq \norm{\omv_n}_2 \sum_{j=1}^{n} |[\uv_n]_j| \leq  n\norm{\omv_n}_2.
    \end{align}
    where we use the fact that $\nabla_{\xv} \zv_j(\xv) = i\omv_j \exp(i\omv_j \xv )$ in the first equality, and the triangle inequality is applied to get the second inequality. Recall that the absolute value of the $j$'th element of $\zv$ is one, i.e. $|[\zv]_j(\xv)| = 1$.
    Lastly, we use the ascending ordering assumption on $\zv$ to bound all $\norm{\omv_j}$'s with $\norm{\omv_n}$
 \end{proof} 
Now, we state the theorem showing the convergence of Alg.~\ref{alg:ApproxKernel}
\begin{theorem}[Performance Guarantee of Alg.~\ref{alg:ApproxKernel}] \label{thm:appApproxAlg}
Given a discrete spectrum kernel (Def.~\ref{def:appDisSpecKernel}) $k: \XC \times \XC \rightarrow \Rbb$ with $\XC \in \Rbb^d$ a compact set of diameter $\text{diam}(\XC)$, Alg.~\ref{alg:ApproxKernel} is employed to obtain the feature vector $\phv(\xv)$. Then, the error of the approximated kernel $s(\xv,\yv) = \phv^\dagger(\yv) \phv(\xv)$ is given by
\be 
\Pr[\sup_{\xv,\yv \in \XC} |s(\xv,\yv) - k(\xv,\yv)| \geq \epsilon] \leq 2^7 \frac{\text{diam}^2(\XC) \Ebb_{N\sim p} \bigl[N^3 \norm{\omv_N}_2^2 \bigr] }{\epsilon^2} \exp(-\frac{2\epsilon^2 D}{|\Omega|^2 (d+1)})  \ ,
\ee
where $p$ is the Cholesky distribution of the kernel $k$ in Def.~\ref{def:appCholeskyPdf}, $\omv_n$ is the frequency corresponding to $n$-th row of $F$, $\Omega$ is the frequency support of $k$ and $N$ is an integer random variable sampled from distribution $p$. 

Moreover, $\sup_{\xv,\yv \in \XC} |s(\xv,\yv) - k(\xv,\yv)| \leq \epsilon$ with any fixed probability when $D \in \bm{\Omega}\Biggl(\frac{|\Omega|^2 d}{\epsilon^2}\log(\frac{\text{diam}^2(\XC)\Ebb \bigl[N^2 \norm{\omv_N}_2^2 \bigr]}{\epsilon^2})\Biggr)$
\end{theorem}
\begin{proof}
We follow the steps in Ref. \cite{rahimi2007random} to derive bounds on the error. 
First, let us define an error function as $\EC(\xv,\yv) = s(\xv,\yv) - k(\xv,\yv)$. 
For ease of understanding, we use the notation $\Delv = (\xv , \yv)$. 
We then write $\EC(\xv,\yv)$ as a function on $\XC^2= \XC \times \XC$ with respect to this new variable $\Delv$. Note that $\XC^2$ is a compact set in $\Rbb^{2d}$ with diameter $\text{diam}(\XC^2) = \sqrt{2} \text{diam}(\XC)$. According to Ref.~\cite{cucker2002mathematical}, there exists an $r$-net that covers $\XC^2$ with at most $T = (2\sqrt{2} \text{diam}(\XC)/r)^{2d}$ center points $\{\Delv_i\}_{i=1}^T$. This means that every $\delv \in \XC^2$ is closer than $r$ to one of these center points; that is, there exists a $\Delv_{i(\delv)}$ from the set of center points such that $\norm{\delv - \Delv_{i(\delv)}}_2 \leq r$.

We use the following Lemma to show the results. 
\begin{lemma} \label{lem:aux}
    Given a compact set $\mathcal{D} \subset \Rbb^d$, an $r$-net on $\mathcal{D}$ with center points $\{\Delv_i\}_{i=1}^T$ with $T = (2 \text{diam}(\mathcal{D})/r)^{d}$, for any Lipschitz function $\EC: \mathcal{D} \rightarrow \Rbb$ with Lipschitz constant $L_\EC$, and any $\epsilon \geq 0$, 
    if   
    \begin{equation}
        \EC(\Delv_i) \leq \epsilon/2 \quad \forall i \in \{1,2,\cdots,T\} \ ,
    \end{equation}
and 
\be
    L_\EC \leq \epsilon/2r \,
\ee
then we have
    \be
           |\EC(\delv)|\leq \epsilon \quad \forall \delv\in \mathcal{D}. 
    \ee
\end{lemma}
\begin{proof}
We are given a compact set $\mathcal{D}$ and a set of center points  $\{\Delv_i\}_{i=1}^T$ for an $r$-net on this set. We denote the closest center point to a point $\delv\in\mathcal{D}$ with $\Delv_{i(\delv)}$. Using the definition of an $r$-net we have $\norm{\delv - \Delv_{i(\delv)}}_2 \leq r$. Moreover, since $L_\EC$ is the Lipschitz constant of $\EC$, we have $|\EC(\delv) - \EC(\Delv_{i(\delv)})| \leq L_\EC \norm{\delv - \Delv_{i(\delv)}}_2 $. Therefore we have
\begin{align}\label{eq:moreexp}
  |\EC(\delv)| &= |\EC(\delv) - \EC(\Delv_{i(\delv)}) + \EC(\Delv_{i(\delv)})| \leq  |\EC(\delv) - \EC(\Delv_{i(\delv)})| + |\EC(\Delv_{i(\delv)})| \\ 
  & \leq L_\EC \norm{\delv - \Delv_{i(\delv)}}_2 + \max_{i\in \{1,2,\cdots,T\}}|\EC(\Delv_{i})| \leq L_\EC \cdot r + \max_{i\in \{1,2,\cdots,T\}}|\EC(\Delv_{i})| \,
\end{align}
where in the second line we used the Lipschitz property of $\EC$.
\end{proof}

We can apply Lemm.~\ref{lem:aux} to the error function $\EC$ defined above. Note that in this case the domain of the error function is $2d$ dimensional and $\text{diam}(\mathcal{D})= \sqrt{2}\text{diam}(\XC)$. Moreover, since $\EC$ is a random function depending on samples of distribution $p$, according to Lemm.~\ref{lem:aux} if events ($\mathrm{i}$) $L_\EC \leq \epsilon/2r$ and ($\mathrm{ii}$) $\EC(\Delv_i) \leq \epsilon/2$, $\forall  i \in \{1,\cdots, T\}$ both happen we can say that $ \EC(\delv) \leq \epsilon$, $\forall \delv \in \XC^2$. 
That is, the probability bounds of these two events provide the bound of the event $ \EC(\delv) \leq \epsilon$.
Therefore, in this proof we first bound the probability of each of these events and then bound the probability of both of them happening together. 

We first work on ($\mathrm{i}$) $\Pr(L_\EC \geq \ep/2r)$. 
Here, our strategy is to use Markov's inequality. Namely, we have: $\Pr[L_\EC^2 \geq t^2 ] \leq \Ebb[L_\EC^2]/t^2$. 
Next, we need to bound  $\Ebb[L_\EC^2]$. Note that, because $\EC$ is differentiable (both $s,k$ are differentiable), $L_\EC = \max_{\Delv} \norm{\nabla \EC(\Delv)}_2 =\norm{\nabla \EC(\Delv^*)}_2 $ where $\Delv^* =\argmax_{\Delv} \norm{\nabla \EC(\Delv)}_2 $. 
Hence, we get

\begin{align}
   \Ebb[L_\EC^2] = \Ebb[\norm{\nabla \EC(\Delv^*)}_2^2] \overset{(1)}{=} \Ebb[\norm{\nabla_{\xv} \EC(\Delv^*)}_2^2+ \norm{\nabla_{\yv} \EC(\Delv^*)}_2^2].
\end{align}
Here, we utilize the fact that $\Delv$ is the concatenation of $\xv, \yv$ in the step (1).
Now, due to the symmetry of the error function and linearity of the expectation, we only have to consider $\Ebb[\norm{\nabla_{\xv} \EC(\Delv^*)}_2^2]$. Then, we obtain
\begin{align}
    \Ebb[\norm{\nabla_{\xv} \EC(\Delv^*)}_2^2] &= \Ebb[\norm{\nabla_{\xv} s(\xv^*,\yv^*) -\nabla_{\xv} k(\xv^*,\yv^*)}_2^2 ]\overset{(1)}{\leq}  \Ebb[\norm{\nabla_{\xv} s(\xv^*,\yv^*)}_2^2] = \frac{1}{D^2} \Ebb \Biggl[\norm{\sum_{j=1}^D g^*(n_j, \yv^*) \nabla_{\xv} g(n_j, \xv^*)}_2^2 \Biggr] \nonumber \\ 
    & \overset{(2)}{\leq} \frac{1}{D}  \Ebb\biggl[ \sum_{j=1}^D \norm{ g^*(n_j, \yv^*) \nabla_{\xv} g(n_j, \xv^*)}_2^2\biggr] = \frac{1}{D}  \Ebb\biggl[ \sum_{j=1}^D |g^*(n_j, \yv^*)|^2\norm{  \nabla_{\xv} g(n_j, \xv^*)}_2^2\biggr] \overset{(3)}{\leq}  \frac{1}{D} \Ebb\biggl[ \sum_{j=1}^D n_j n_j^2 \norm{\omv_j}_2^2\biggr] \nonumber \\ 
    & = \Ebb[n^3 \norm{\omv_n}_2^2],
\end{align}
where for step (1) we bounded the variance of a random variable with its second moment. For step (2) we used the fact that $\norm{\sum_{i=1}^D \av_i}_2^2 \leq (\sum_{i=1}^D\norm{ \av_i}_2)^2 \leq D\sum_{i=1}^D\norm{ \av_i}_2^2$ (triangle inequality and Cauchy-Schwartz inequality). For step (3) we use Lem.~\ref{lem:appg}.

By defining $\kappa^2 = \Ebb_{N \sim p}[N^3 \norm{\omv_N}_2^2]$, we put the two equations above together and arrive at the following Markov's bound: 
\be\label{eq:appLf}
\Pr[L_\EC \geq \frac{\epsilon}{2r} ] \leq \biggl( \frac{2r\kappa}{\epsilon} \biggr)^2
\ee

Next, we move on to bound ($\mathrm{ii}$) $\Pr[|\EC(\Delv_i)| \geq \epsilon]$. 
For ease of later discussion, we briefly recall sub-Gaussian random variables and some of their properties.
\begin{definition}($\eta^2$ Sub-Gaussian Random Variables) \label{def:appsubG}
    A random variable $X$ is called $\eta^2$ sub-Gaussian if, for every $\lambda$, the following satisfies; 
    \be
    \Ebb[e^{X-\Ebb[X]}] \leq e^{\frac{\lambda^2 \eta^2}{2}}
    \ee
    Also, $\eta^2$ is called the variance proxy of this random variable.
\end{definition}
\begin{lemma}\label{lem:appsubG}
Sub-Gaussian random variables have the following properties.
\begin{enumerate}
    \item For every sub-Gaussian random variable $X$ with variance proxy $\eta^2$ we have
    \be
        \Pr[|X-\Ebb[X]|\geq \epsilon] \leq 2\exp(-\frac{\epsilon^2}{2\eta^2}) \ 
    \ee
    for every $\epsilon > 0$. \\
    \item A random variable that is bounded in $[a,b]$ is $(\frac{b-a}{2})^2$ sub-Gaussian. \\
    \item Given $n$ independent sub-Gaussian random variables $X_1, \cdots,X_n$ with variance proxy $\eta_1^2, \cdots, \eta_n^2$ respectively, the sum of these random variables is sub-Gaussian with variance proxy $\eta_1^2+\cdots+\eta_n^2$.\\
    \item From 1, 3. If $X_i$'s are $n$ independent sub-Gaussian random variables with variance proxies $\eta_i^2$ we have 
    \be
        \Pr[\lvert\frac{1}{n}\sum_{i=1}^{n} X_i-\Ebb[X_i]\rvert\geq \epsilon] \leq 2\exp(-\frac{n^2\epsilon^2}{2\sum \eta_i^2}) \ .
    \ee
\end{enumerate}
\end{lemma}
Note that, according to Lem.~\ref{lem:appg}, we have $g(N,\yv_i)^* g(N,\xv_i) \leq N$, where $N$ ranges between $1$ and $|\Omega|$; this indicates that $g(N,\yv_i)^* g(N,\xv_i)$ is at least $(\frac{|\Omega-1|}{2})^2$ sub-Gaussian (Def.~\ref{def:appsubG} and point 3 of Lem.~\ref{lem:appsubG}). However, depending on the kernel, $g(N,\yv_i)^* g(N,\xv_i)$ could be more concentrated and therefore $\eta^2$ sub-Gaussian for some $\eta < \frac{|\Omega-1|}{2}$.
For the sake of generality and possible tighter bounds, we keep a general form and assume $g(N,\yv_i)^* g(N,\xv_i)$ is $\eta^2$ sub-Gaussian. 
We then use Lem.~\ref{lem:appsubG}, to bound $s(\xv,\yv) = \frac{1}{D} \sum_{i=1}^{D} g(n_i,\yv)^* g(n_i,\xv)$:  
\be \label{eq:Et1}
\Pr[|\EC(\Delv_i)| \geq \epsilon] = \Pr[\left|\frac{1}{D} \sum_{i=1}^{D} g(n_i,\yv_i)^* g(n_i,\xv_i) - k(\Delv_i)\right| \geq \epsilon ] \leq 2\exp(-\frac{\epsilon^2 D}{2\eta^2}). 
\ee

Now we can bound the probability of the maximum error of our approximation.
\begin{align}
    \Pr[\max_{\xv,\yv \in \XC} |s(\xv,\yv) - k(\xv,\yv)| \geq \epsilon] &= 1-\Pr[\max_{\xv,\yv \in \XC} |s(\xv,\yv) - k(\xv,\yv)| < \epsilon] \\ 
    & \leq  1-\Pr[\cap_{i=1}^T |\EC(\Delv_i)| < \frac{\epsilon}{2} \cap L_{\EC} < \frac{\epsilon}{2r}] \\ 
    &= \Pr[\neg\left(\cap_{i=1}^T |\EC(\Delv_i)| < \frac{\epsilon}{2} \cap L_{\EC} < \frac{\epsilon}{2r}\right)] \\ 
    &= \Pr[\cup_{i=1}^T |\EC(\Delv_i)| \geq \frac{\epsilon}{2} \cup L_{\EC} \geq \frac{\epsilon}{2r}] \\ 
    &\leq 2T\exp(-\frac{\epsilon^2 D}{2\eta^2})  + \biggl( \frac{2r\kappa}{\epsilon} \biggr)^2\\
    &= 2\biggl(\frac{2\sqrt{2}\text{diam}(\XC)}{r} \biggr)^{2d}\exp(-\frac{\epsilon^2 D}{2\eta^2})  + \biggl( \frac{2r\kappa}{\epsilon} \biggr)^2,
\end{align}
where we use the complement of the event in the first, third and fourth lines.
In the fifth line, we used union bound and Eq.~\eqref{eq:appLf}, Eq.~\eqref{eq:Et1} are used for the last equality. 

Now, this bound has the form $c_1 r^{-2d}+ c_2 r^{2}$. We choose $r= (\frac{c_1}{c_2})^{\frac{1}{2d+2}}$ and substitute it in the bound. By doing so the bound will have the form $2c_1^{\frac{2}{2d+2}} c_2^{\frac{2d}{2d+2}}$ and the first part of the theorem is proven.
To prove the second part, we can fix any probability for the left-hand side and solve for $D$.
Then, we get
\be
D \in \bm{\Omega}\Biggl(\frac{\eta^2 d}{\epsilon^2}\log(\frac{d\Ebb[N^3 \norm{\omv_N}_2^2]}{\epsilon^2})\Biggr).
\ee
\end{proof} 
Fidelity quantum kernels, with integer frequencies i.e. $\Omega \subset \mathbb{Z}^d$  are periodic with a period of $2\pi$ on each dimension of their inputs. Consequently, they can be defined on $\XC= [0,2\pi)^d$ with the diameter of $\text{diam}(\XC)=\sqrt{\sum_{i=1}^d (2\pi)^2} = 2\pi \sqrt{d}$.
Therefore, we get 
\be
\Pr[\max_{\xv,\yv \in \XC} |s(\xv,\yv) - k(\xv,\yv)| \geq \epsilon] \leq 2^9 \pi^2 d\frac{\Ebb[N^3 \norm{\omv_N}_2^2]}{\epsilon^2} \exp(-\frac{\epsilon^2 D}{2\eta^2(d+1)}) .
\ee
Note that, as mentioned before, an upper bound for $\eta^2$ is $\frac{|\Omega|^2}{4}$. 
Since the number of frequencies scales exponentially in $d$ this upper bound suggests the hardness of the approximation in the number of dimension. However, we conjecture that for most fidelity quantum kernels, $\eta^2$ is actually lower. Indeed, Ref.~\cite{barthe2024gradients} suggests that most fidelity quantum kernels are low pass (the Fourier coefficient for high frequencies is relatively low), implying the potential for easier implementation.  

\subsection{Kernel Approximation using Eigendecomposition} 
Similar to the case for Cholesky decomposition, eigenvalue decomposition can also be used for the approximation. 
In this approach, instead of sampling the columns of the upper diagonal decomposition of $F$ from the Cholesky distribution (Def.~\ref{def:appCholeskyPdf}), we sample columns of the eigenvectors of the Fourier transform from eigenvalue distribution (Def.~\ref{def:appEVDPdf}). 
This method can reproduce the convergence bound derived for the shift-invariant kernels in Ref.~\cite{rahimi2007random}.

Given a discrete spectrum kernel $k$, we decompose its Fourier transform $F$ to eigenvalues and eigenvectors, such that $F = UVU^\dagger$. Then, the kernel can be written as $k(\bm{x},\bm{y}) = \sum_{i=1}^{|\Omega|} v_i \bm{z}^\dagger(\bm{y}) \uv_i \uv_i^\dagger \zv(\xv)$, where $\uv_i$ is the $i$'th column of $U$. By defining $g(i,\xv) = \uv_i^\dagger \zv(\xv)$, the kernel can be expressed as
\be
k(\bm{x},\bm{y}) = \sum_{n=1}^{|\Omega|} v_n g^*(n,\yv) g(n,\xv) = \Ebb_{N \sim v}[g^*(N,\yv) g(N,\xv)] \ . 
\ee
Thus, we can sample the features from $v$ to construct the kernel using its sample mean, which is provided in Alg.~\ref{alg:ApproxKernelEVD}. Note that the matrix $U$ in this section consists of eigenvectors and is different from the Cholesky decomposition matrix $U$ we used in the previous section. Therefore the properties of the feature function $g$ are different from what we have in Lemm.~\ref{lem:appg}. We state some properties of function $g$ in the following lemma. 

\begin{algorithm}[t]
\caption{Random Features Approximation of Quantum Kernels (eigenvectors)} \label{alg:ApproxKernelEVD}
\begin{algorithmic} 
\STATE{
\textbf{Input:} A discrete spectrum kernel $k: \XC \times \XC \rightarrow \Rbb$, An integer $D$ indicating the dimension of the approximated feature vector}
\STATE{\textbf{Output:} A randomized feature map $\phv: \Rbb^d\rightarrow \Cbb^D$ such that $s(\xv,\yv)= \phv(\yv)^\dagger\phv(\xv) \approx k(\xv,\yv)$}
\vspace{5pt}
\noindent \hrule
\vspace{5pt}
\STATE{\textbf{1:} Take the Fourier transform of $k$ and write it as $k(\bm{x},\bm{y})  = \bm{z}^\dagger(\bm{y}) F \bm{z}(\bm{x})$ as defined in Cor.~\ref{prop:qk} with $\zv$ following ascending ordering (Def.~\ref{def:appAscOrd}).}
\STATE{\textbf{2:} Decompose $F= U V U^\dagger $ where $U$ the eigenvector matrix of $F$ and $V$ is diagonal with eigenvalues in its diagonal. Diagonal elements of $V$ form  a probability mass function $v: \{1,\cdots, |\Omega|\} \rightarrow \Rbb^+$ such that $v_i = V_{ii}$.}
\STATE{\textbf{3:} Sample $D$ i.i.d. integer samples $n_1,\cdots, n_D$ from $v$ }
\STATE{\textbf{4:} Construct $g(n_i,\xv) = \uv_{n_i}^\dagger \zv(\xv)$ where $\uv_{n_i}$ is the $n_i$'th column of $U$.}
\STATE{\textbf{5:} Return $\phv(\xv) = \frac{1}{\sqrt{D}} [g(n_1,\xv), \cdots, g(n_D, \xv)]^T$}
\end{algorithmic}
\end{algorithm}

\begin{lemma}\label{lem:appgEV}
    Suppose the feature function $g: \{1,2,\cdots,|\Omega|\} \times \XC \rightarrow \Cbb$ defined as $g(n,\xv) = \uv_n^\dagger \zv(\xv)$  with $\zv(\xv)$ in Def.~\ref{def:appAscOrd} and columns of a unitary  $\uv_i$'s.  Then g has the following properties:  
    \begin{itemize}
        \item $\forall \xv \in \XC, n \in \{1,2,\cdots,|\Omega|\} \quad  |g(n,\xv)| \leq \norm{\uv_n}_1$
        \item  $\forall \xv \in \XC, n \in \{1,2,\cdots,|\Omega|\} \quad \norm{\nabla_{\xv}g(n,\xv)}_2 \leq |[\uv_{n}]_n| \norm{\omv_n}_2 + B(\norm{\uv_n}_1 - |[\uv_{n}]_n|)$
    \end{itemize}
    Where $B = \norm{\omv_{|\Omega|}}_2$ is the norm of the largest frequency that exists in $\zv$.  
\end{lemma}
\begin{proof}
     For the first part, using Cauchy–Schwarz inequality, we have
    \be
        |g(n,\xv)| = \left|\sum_{i=1}^{|\Omega|} [\uv_n]_i [\zv]_i(\xv)\right| \leq \sum_{i=1}^{|\Omega|} |[\uv_n]_i| |[\zv]_i(\xv)| =  \norm{\uv_n}_1,
    \ee 
    where we use $|[\zv]_i| = 1 $. Now, let us bound the gradient of the feature function as follows; 
    \begin{align}
        \norm{\nabla_{\xv}g(n,\xv)}_2 &= \norm{\sum_{j=1}^{|\Omega|} i[\uv_n]_j [\zv]_j(\xv) \omv_j}_2 \leq \sum_{j=1}^{|\Omega|} |[\uv_n]_j| |[\zv]_j(\xv)| \norm{\omv_j}_2 \nonumber\\ 
        &= \sum_{j=1}^{|\Omega|} |[\uv_n]_j| \norm{\omv_j}_2  = |[\uv_n]_n| \norm{\omv_n}_2 + \sum_{j=1, j\neq n}^{|\Omega|} |[\uv_n]_j| \norm{\omv_j}_2 \leq |[\uv_n]_n| \norm{\omv_n}_2 + B \sum_{j=1, j\neq n}^{|\Omega|} |[\uv_n]_j| \nonumber \\ 
        & = |\uv_{nn}| \norm{\omv_n}_2 + B(\norm{\uv_n}_1 - |\uv_{nn}|)
    \end{align}
    where we first use the fact that $\zv_j(\xv) = \exp(i\omv_j \xv )$, and we then use the triangle inequality. We have denoted $[\uv_n]_n$ with $\uv_{nn}$
    Lastly, we apply the ordering assumption on $\zv$ to introduce $B$. 
 \end{proof} 
Previously, we have discussed that, if the Fourier transform of a discrete spectrum kernel is diagonal, that kernel is shift invariant and vice versa. This motivates us to define a measure of non-stationarity for these types of kernels. 
 \begin{definition}[Non-stationarity of a Discrete Spectrum Kernel] \label{def:appNSMeasure} 
     Given a discrete spectrum kernel $k$ and its Fourier transform $F$ such that $k(\xv,\yv) = \zv^\dagger(\yv) F \zv(\xv)$, we define its $n$'th non-stationarity measure as 
     \be 
        \zeta_n := \norm{\uv_n}_1 \bigl[|U_{nn}| \norm{\omv_n}_2 + B(\norm{\uv_n}_1 - |U_{nn}|)\bigr]
     \ee
     where $\uv_n$'s are the orthonormal eigenvectors of $F$ and $U$ is a unitary with $n$'th column $\uv_n$. Moreover $\omv_n$ is the frequency associated with the $n$'th element of $\zv$ and $B = \max_{\omv \in \Omega}\norm{\omv}_2$. 
 \end{definition}
To see how this measure is related to non-stationarity, we consider two cases. The first example is a shift-invariant kernel with diagonal $F$. 
In this case, $U= \mathbb{1}$ and then $\forall n, U_{nn} = \|\uv_n\|_1=1$.  
As a result $\zeta_n= \norm{\omv_n}_2$ which is its minimum. 
On the other hand, if the kernel is maximally non-diagonal such that $\forall i,n, \  |(\uv_n)_i| =  1/\sqrt{|\Omega|}$, then  $\zeta_n =  \norm{\omv_n}_2+ B({|\Omega|}-1)$, which is a large number because of $|\Omega|$ and $B$. 
This shows that $\zeta_n$ captures ``non-diagonalness" (i.e., how much the matrix is far from the diagonal one) of the kernel well. 

Now, we state the convergence bounds for Alg.~\ref{alg:ApproxKernelEVD}, which depends on this non-stationarity parameters.
\begin{theorem}[Performance Guarantee Alg.~\ref{alg:ApproxKernelEVD}] \label{thm:appApproxAlgEVD}
Given a discrete spectrum kernel (Def.~\ref{def:appDisSpecKernel}) $k: \XC \times \XC \rightarrow \Rbb$ with $\XC \in \Rbb^d$ a compact set of diameter $\text{diam}(\XC)$, we apply Alg.~\ref{alg:ApproxKernelEVD} with $D$ samples and obtain the feature vector $\phv(\xv)$. 
Then, the error of the approximated kernel $s(\xv,\yv) = \phv^\dagger(\yv) \phv(\xv)$ is given by 
\be 
\Pr[\sup_{\xv,\yv \in \XC} |s(\xv,\yv) - k(\xv,\yv)| \geq \epsilon] \leq 2^7 \frac{\text{diam}^2(\XC) \Ebb_{N\sim v} \bigl[ \zeta_N^2  \bigr] }{\epsilon^2} \exp(-\frac{2\epsilon^2 D}{|\Omega|^2 (d+1)})  \ ,
\ee
where $v$ is the eigenvalue distribution of the kernel $k$ defined in Def.~\ref{def:appEVDPdf} and $\zeta_n$ is the $n$'th non-stationarity measure of the kernel defined in Def.~\ref{def:appNSMeasure}.
Here, $\omv_n$ is the frequency corresponding to $n$'th row of $F$ and $\Omega$ is the frequency support of $k$. 

Moreover, $\sup_{\xv,\yv \in \XC} |s(\xv,\yv) - k(\xv,\yv)| \leq \epsilon$ with any fixed probability when $D \in \bm{\Omega}\Biggl(\frac{|\Omega|^2 d}{\epsilon^2}\log(\frac{\text{diam}^2(\XC)\Ebb_{N\sim v} \bigl[ \zeta_N^2  \bigr]}{\epsilon^2})\Biggr)$
\end{theorem}
\begin{proof}
The proof is almost identical to the proof of Thm.~\ref{thm:appApproxAlg}, except that the bound for the second moment of Lipschitz constant changes according to Lem.~\ref{lem:appgEV}; 

\begin{align}
    \Ebb[\norm{\nabla_{\xv} \EC(\Delv^*)}_2^2] &= \Ebb[\norm{\nabla_{\xv} s(\xv^*,\yv^*) -\nabla_{\xv} k(\xv^*,\yv^*)}_2^2 ] \\ 
    &\leq  \Ebb[\norm{\nabla_{\xv} s(\xv^*,\yv^*)}_2^2] \\ 
    &= \frac{1}{D^2} \Ebb \Biggl[\norm{\sum_{j=1}^D g^*(n_j, \yv^*) \nabla_{\xv} g(n_j, \xv)}_2^2 \Biggr]  \\ 
    & = \frac{1}{D}  \Ebb\biggl[ \sum_{j=1}^D \norm{ g^*(n_j, \yv^*) \nabla_{\xv} g(n_j, \xv)}_2^2\biggr] \\
    & \leq \frac{1}{D}  \Ebb\biggl[ \sum_{j=1}^D |g^*(n_j, \yv^*)|^2\norm{  \nabla_{\xv} g(n_j, \xv)}_2^2\biggr]  \\ 
    &\leq  \frac{1}{D} \Ebb\biggl[ \sum_{j=1}^D  \norm{\uv_{n_j}}^2_1 \bigl (|U_{n_j,n_j}| \norm{\omv_{n_j}}_2 + B(\norm{\uv_{n_j}}_1 - |U_{n_j,n_j}|) \bigr)^2\biggr]  \\ 
    & = \Ebb[\zeta^2_n] .
\end{align}
Here, the second line utilizes the fact that the second moment of a random variable is always larger that the variance and Lem.~\ref{lem:appgEV} is used in the last line. 
$\zeta_n$ is defined as $n$'th non-stationarity measure of kernel $k$ (Def.~\ref{def:appNSMeasure}). 
The rest of the proof follows the same procedure as the one for Thm.~\ref{thm:appApproxAlg}.
\end{proof}
Thm.~\ref{thm:appApproxAlgEVD} shows that the number of frequency samples scales with $|\Omega|^2 \log(\Ebb[\zeta^2_N])$. This bound suffers from the same problem as the bounds in Thm.~\ref{thm:appApproxAlg}; that is, the exponential scaling of $|\Om|$ with $d$. 
However, this bound is more interpretable because, if the kernel is shift-invariant, $\Ebb[\zeta^2_N] = \Ebb[\norm{\omv}_2^2]$ and we retrieve the bounds in Ref.~\cite{rahimi2007random}. 

Overall, this section provides a method to approximate discrete spectrum kernels using the Cholesky decomposition and eigenvalue decomposition of their Fourier transform. We also derived upper bounds on the probability of point-wise error. Comparing the two upper bounds, Alg.~\ref{alg:ApproxKernel} shows a better scaling with $D$ if $g^*(N,.)g(N,.)$ is $\eta^2$ sub-Gaussian with $\eta^2\leq |\Omega/4|$.  

\subsection{Numerical Experiment} 
Finally, in this section, we support our theoretical results numerically. To show the convergence of Algs.~\ref{alg:ApproxKernel},~\ref{alg:ApproxKernelEVD}, we consider a simple kernel of 2 qubits and 2 layers with the feature map in Fig.~\ref{fig:kernel}. As shown in Fig.~\ref{fig:kernelapprox}, the approximation error for both algorithms converges to zero. We see that Alg.~\ref{alg:ApproxKernel}, performs slightly better than Alg.~\ref{alg:ApproxKernelEVD}. The advantage seems to be in the prefactor of the exponential decay over $D$. This could indicate that our bounds for Alg.~\ref{alg:ApproxKernelEVD} are looser compared to Alg.~\ref{alg:ApproxKernel} despite its interpretability.

\begin{figure}
\vskip 0.2in
\begin{center}
\centerline{\includegraphics[width=0.70\linewidth]{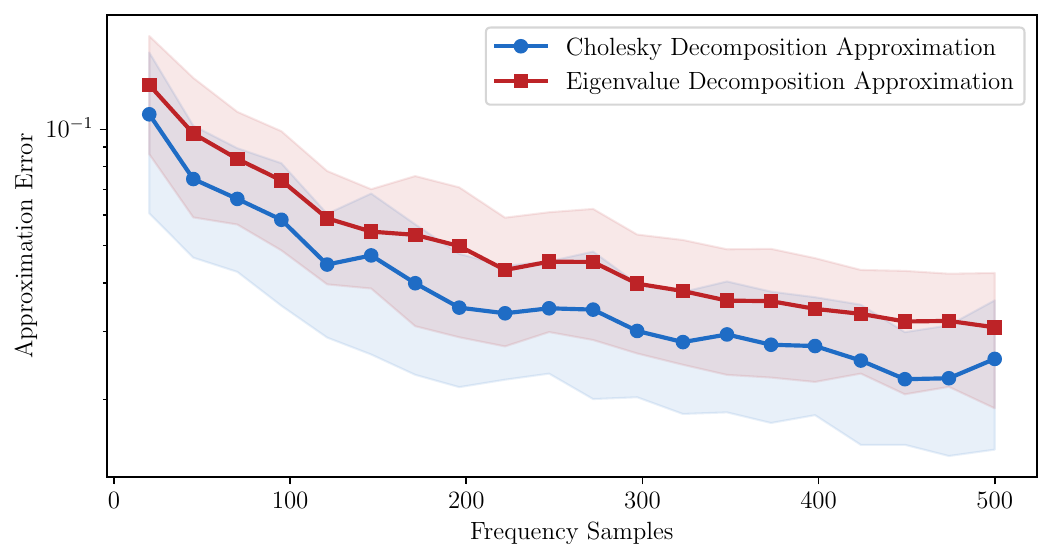}}
\caption{\textbf{Alg.~\ref{alg:ApproxKernel} vs Alg.~\ref{alg:ApproxKernelEVD}.} The approximation error for our proposed algorithms for a 2 qubit kernel of 2 layers of the feature map in Fig.~\ref{fig:kernel} over the number of random features. The plots average over 50 runs of the algorithms, with standard deviation shown as shaded areas. The approximation error is the average error over 10 uniformly randomly chosen points between $[0,2\pi)$.}
\label{fig:kernelapprox}
\end{center}
\vskip -0.2in
\end{figure}

\section{Alternative Alignment Condition of Ref.~\texorpdfstring{\cite{sweke2023potential}}{Lg}} \label{app:alignment}

The main result of Ref.~\cite{sweke2023potential} (Thm. 1) states sufficient conditions under which RFF can reach the performance of QNN regression. A slightly different version of this theorem is stated as Thm.~\ref{thm:appSweke1} and we will later discuss this theorem in more details. In this appendix we show how these sufficient conditions imply an alignment between the sampling distribution of the RFF algorithm and the optimal decision function of QNN. 

One of the sufficient conditions for dequantization provided in Ref.~\cite{sweke2023potential} is the \textit{alignment} condition; expressed as, 
$$\norm{f_{\operatorname{QNN}}}_{k_p} \in \OC(\poly(d))$$
where $f_{\operatorname{QNN}}$ (Eq.~\eqref{eq:pqcft}) is the output of the QNN model and $k_p$ is the shift-invariant kernel whose Fourier transform is distribution $p$ (defined in Eq.~\eqref{eq:kp}). 

This condition rises from the dependence of the number of required RFF frequency samples on the RKHS norm. To be more precise, the main result of Ref.~\cite{sweke2023potential} (Thm. 1) shows that RFF can reach the performance of QNN regression if $m,D \in \Omega(\poly(\norm{f_{\operatorname{QNN}}}_{k_p}))$ where $m,D$ are the number of data and frequency samples, respectively. 

Consequently, the value of this RKHS norm is important in determining the efficiency of RFF dequantization. Ideally we want to find a sampling distribution such that $\norm{f_{\operatorname{QNN}}}_{k_p}$ is relatively small. Ref.~\cite{sweke2023potential} suggests that a sampling distribution $p$ that is aligned with the Fourier transform of the optimal decision function $f_{\operatorname{QNN}}$ is a good choice for RFF, using some examples to support their claim. Here, we show that a distribution that is proportional to the Fourier transform of $f_{\operatorname{QNN}}$ is the minimizer of $\norm{f_{c}}_{k_p}$. 
In other words, we show the following:

Take any function $f_c$ in the RKHS. It can be described as a vector, as $f_c(\cdot) = \expval{\bm{c},\phi_{\operatorname{QNN}}(\cdot)}$ with $\phi_{\operatorname{QNN}}$ as in Eq.~\ref{eq:pqcfm}. The distribution $\bm{p}$ which minimizes $\norm{f_{c}}_{k_p}$ is such that $\bm{p} \propto \bm{c}$, i.e. the two vectors are aligned. 

Before we proceed, let us recall the settings and the definitions in Ref.~\cite{sweke2023potential}.
The kernel considered in Ref.~\cite{sweke2023potential} is a so-called re-weighted kernel for distribution $\bm{p}$, defined as $k_p(\xv,\xv') = \expval{\phv_p(\xv), \phv_p(\xv')} = \sum_{i = 0}^{|\Omega_+|} \bm{p}_i \cos(\omv_i \cdot (\xv-\xv'))$ with the feature map
\be
\phv_p(x) =   \bigl(\sqrt{\pvv_0}, \sqrt{\pvv_1} \cos(\omv_1 \cdot \xv), \sqrt{\pvv_1} \sin(\omv_1 \cdot \xv), \cdots, \sqrt{\pvv_{|\Omega_+|}} \cos(\omv_{|\Omega_+|} \cdot \xv), \sqrt{\pvv_{|\Omega_+|}}  \sin(\omv_{|\Omega_+|} \cdot \xv)\bigr).
\ee

Note that, throughout this section, $\bm{p}$ is a vector of size $|\Omega_+| + 1$ with positive elements that sum to one, while $\bm{p}_i$ and $\pvv_0$ denote the sampling probability for frequency $\omv_i$ and frequency zero, respectively. 

We recall from Def.~\ref{def:RKHS} that, for a kernel $k:\XC \times \XC \rightarrow \RC$, such that $k(\xv,\xv') = \expval{\phv(\xv'), \phv(\xv)}$, RKHS norm can be defined as: 
\be 
\norm{f}_k = \inf\{\norm{\bm{\nu}}_2 | \quad \bm{\nu} \in \XC 
\quad s.t \quad f(\cdot) = \expval{\bm{\nu},\phv(\xv)} \}.
\ee
We highlight two important facts.
First, since the basis functions i.e., the elements of $\phv(\xv)$ are linearly independent, the representation of any function with respect to the kernel is unique and the infimum of $\nu$ is simply the unique element that describes the function. 
Second, from Observation 1 in Ref.~\cite{sweke2023potential}, we know that the RKHS is invariant under reversible re-weightings. This means that $\|f_{\text{QNN}}\|_{k_p}$ can be reconstructed by any of the re-weighting kernels with any distribution $\pvv$ as long as $\forall i, \ \pvv_i \neq 0$.  

Now, we prove that the weights for such a kernel should be aligned with the frequency spectrum of the target function to minimize the corresponding RKHS norm.  
To do so, we define the feature map associated with the uniform weighting $u$ as $\phv_u(\xv)$. Then, we can express $f_{\operatorname{QNN}} = \expval{\bm{c}, \phv_u}$, where $\bm{c}$ consists of normalized Fourier coefficients of $f_{\operatorname{QNN}}$ (Eq.~\eqref{eq:pqc}).  

Given a re-weighted kernel $k_p$, we can also write $f_{{\operatorname{QNN}}} = \expval{\bm{v} , \phv_{p}}$, and thus we have 
\be
\bm{v} \odot \sqrt{\pvv}  = \bm{c}
\ee 
where $\odot$ represents point-wise multiplication. Alternatively, we have $\bm{v} = \bm{c}/\sqrt{\pvv}$ with point-wise division $/$. As a result, the minimization problem can be reduced to
\be \label{eq:cvopt}
\min_{\pvv} \sum_i \bm{v}_i^2 = \sum_i \frac{\bm{c}_i^2}{\pvv_i} \quad s.t\quad \pvv_i\geq 0,\ \sum \pvv_i = 1.
\ee

This is a convex optimization problem over $\pvv$. We recall that an optimization problem of the form 
\begin{align*}
    \min_{\xv}  \quad & f_0(\xv) \\ 
    s.t. \quad & f_i(\xv) \leq 0  \quad i=1,\cdots,n \\
    & h_i(x) = 0  \quad i=1,\cdots,m
\end{align*}
with convex $f_i$'s and affine $h_i$'s is called a convex optimization problem. We define the Lagrangian of this problem as 
\be 
L(\pvv,\alpha, \lambda) = f_0(\xv) - \sum_{i} \alpha_i f_i(\xv) + \sum_i \lambda_i (h_i(\xv)) \ .
\ee
The following conditions known as KKT conditions~\cite{kuhn1951nonlinear},~\cite[Ch.~5, p.~241-249]{boyd2004convexcorrect} are necessary and sufficient for points to be optimal. In other words, if there exists $\xv^\star, \alphav^\star, \vec{\lambda}^\star$ such that: 
\begin{align}
\begin{cases}
    \partial_{\xv}L(\xv^\star, \alphav^\star,\vec{\lambda}^\star) & = 0  \nonumber \\ 
    f_i(\xv^\star) &\leq 0  \quad  i=1,\cdots,n \nonumber\\
    h_i(\xv^\star) &= 0 \quad  i=1,\cdots,m  \nonumber\\ 
    \alphav_i^\star &\geq 0 \quad  i=1,\cdots,n  \nonumber\\
    \alphav_i^\star  f_i(\xv^\star)&= 0 \quad  i=1,\cdots,n  \nonumber 
\end{cases}
\end{align}
then $\xv^\star$ is the solution of the problem.

For our problem (Eq.\eqref{eq:cvopt}) the Lagrangian has the form 
\be 
L(\pvv,\alpha, \lambda) = \sum_i \frac{\bm{c}_i^2}{\pvv_i} - \sum_{i} \alpha_i \pvv_i + \lambda\bigl(\sum_i \pvv_i -1\bigr) \ .
\ee

By applying KKT conditions~\cite{kuhn1951nonlinear},~\cite[Ch.~5, p.~241-249]{boyd2004convexcorrect}, we have
\begin{align}
\begin{cases}
    \partial_{\pvv_i}L &= -\frac{\bm{c}_i^2}{\pvv_i^2}-\alpha_i + \lambda = 0  \nonumber \\ 
    \pvv_i &> 0 \nonumber \\ 
    \alpha_i\pvv_i &= 0 \rightarrow \alpha_i = 0\nonumber \\ 
    \alpha_i &\geq 0 \nonumber \\ 
    \sum_i \pvv_i &= 1 \nonumber \\ 
    \pvv_i  & \neq  0.
\end{cases}
\rightarrow  \pvv_i^2 = \bm{c}_i^2/\lambda
\end{align}
This concludes that the solution to this optimization problem is a distribution that is proportional to $|\bm{c}|$, i.e., the Fourier transform of  $f_{\text{QNN}}$. 
Hence, the distribution aligned with the spectrum of the target function gives us the lowest value for minimum number of data samples needed for dequantization. 
In the main text, we recast the alignment condition based on this derivation.

\section{Kernel Ridge Regression and RFF} \label{app:regression}
In this section, we prove Prop.~\ref{prop:1} of the main text as Cor.~\ref{col:prop1}. First, we recall the definition of dequantization. Then, we recall the theorem from Ref.~\cite{sweke2023potential} and~\cite{rudi2017generalization}. Next, we state Lem.~\ref{lem:AppDiagNonDiag} about the relation between the RKHS norm of a function with respect to two different kernels. Using these as preliminaries, we prove Thm.~\ref{thm:appKRR} about the generalization of QK regression. Finally, from Thm.~\ref{thm:appKRR}, we derive Cor.~\ref{col:prop1}, which corresponds to Prop.~\ref{prop:1} of the main text. Finally we propose Lem.~\ref{lem:appFeP} on the interpretation of the alignment condition.

We defined dequantization in Def.~\ref{def:rff_dequantization} of the main text. For convenience we state this definition again. 
\begin{definition}[RFF Dequantization] \label{def:apprff_dequantization}
Consider a supervised learning task with a training dataset $\{\xv_i,y_i\}_{i=1}^m  \  \text{s.t.} \ \xv\in \XC \subset \Rbb^d$. 
Given a quantum model (either a QNN or QK) with hypothesis class $\QC$, $f_q$ denotes the true risk minimizer of the quantum model i.e. $f_q=\min_{f \in \QC} \RC[f]$.
We say the task is \textit{RFF-dequantized}, if there exists a distribution $p:\Omega\rightarrow \Rbb$ such that for some $\epsilon > 0$, with $m, D \in \OC(\poly(d, \epsilon^{-1}))$, the following is true with high probability:
\be 
\RC[f_{D,p}] - \RC[f_q] \leq \epsilon.
\ee
Here, $f_{D,p}$ is the decision function of the task trained with a $D$ dimensional RFF feature map (Eq.~\ref{eq:clrff2}) sampled from $p$ and $\RC$ denotes the true risk. 
\end{definition}
In short, we say a quantum learning task is dequantized when there exists a frequency sampling distribution for which RFF works almost as well as the quantum model in terms of true risk. Therefore, our objective in this section is to bound the difference between the true risk of RFF method and the minimum true risk in the RKHS of the quantum kernel.

Generalization of RFF is well studied in a classical setting and we can use the existing results to prove dequantization under certain conditions. In this work, similar to Ref.~\cite{sweke2023potential}, we base our proofs on the results provided by Ref.~\cite{rudi2017generalization}. Ref.~\cite{rudi2017generalization} shows that under certain conditions, RFF regression performs almost as well as the exact kernel ridge regression. The following theorem summarizes these results.

\begin{theorem}[Thm.~3 From Ref.~\cite{sweke2023potential} ] \label{thm:appSweke1}
Consider a regression problem $ \mathcal{X} \times \mathbb{R}$. Let $k: \mathcal{X} \times \mathcal{X} \rightarrow \mathbb{R}$ be a kernel, and $\mathcal{H}_k^C$ be the subset of the RKHS $\mathcal{H}_k$ consisting of functions with an RKHS norm upper bounded by some constant $C$. 
We also denote the true risk minimizer over $\mathcal{H}_k^C$ as $f_{\mathcal{H}_k^C}^*$.  
Now, we assume the following:\\
1. The kernel $k$ has an integral representation
$$
k\left(x, x^{\prime}\right)=\int_{\Phi} \psi(x, v) \psi\left(x^{\prime}, v\right) \mathrm{d} \pi(v).
$$\\
2. The function $\psi$ is continuous in both variables and satisfies $|\psi(x, v)| \leq \kappa$ almost surely, for some $\kappa \in[1, \infty)$.\\
3. For some $b>0$, $|y| \leq b$ almost surely when $(x, y) \sim P$.

Additionally, define
$$
\begin{aligned}
& \bar{B}:=2 b+2 \kappa \max \left\{1,\left\|f_{\mathcal{H}_k^C}^*\right\|_k\right\}, \\
& \bar{\sigma}:=2 b+2 \kappa \sqrt{\max \left\{1,\left\|f_{\mathcal{H}_k^C}^*\right\|_k\right\}},
\end{aligned}
$$
and
$$
\begin{aligned}
& m_0:=\max \left\{4\left\|T_k\right\|^2,\left(264 \kappa^2 \log \frac{556 \kappa^3}{\delta}\right)^2\right\}, \\
& c_0:=9\left(3+4 \kappa^2+\frac{4 \kappa^2}{\left\|T_k\right\|}+\frac{\kappa^4}{4}\right), \\
& c_1:=8\left(\bar{B} \kappa+\bar{\sigma} \kappa+\max \left\{1,\left\|f_{\mathcal{H}_k^C}^*\right\|_k\right\}\right) .
\end{aligned}
$$
Here, $\|T_k\|$ is the operator norm of the kernel integral operator (Def.~\ref{def:appKIO}) and  $\delta \in(0,1]$.
Then, we assume $\lambda_m=1 / \sqrt{m}$ by defining the number of data samples as $m \geq m_0$.
Also, let $\hat{f}_{D, \lambda_m}$ be the output of $\lambda_m$-regularized linear regression with respect to the feature map,
$$
\phi_{D}(x)=\frac{1}{\sqrt{D}}\left[\psi\left(x, v_1\right), \ldots \psi\left(x, v_D\right)\right]
$$
constructed from the integral representation of $k$ by sampling $D$ elements from $\pi$. Then,
$$
D \geq c_0 \sqrt{m} \log \biggl(\frac{108 \kappa^2 \sqrt{m}}{\delta} \biggr)
$$
is enough to guarantee, with probability at least $1-\delta$, that
$$
R\left(\hat{f}_{D, \lambda_m}\right)-R\left(f_{\mathcal{H}_k^C}^*\right) \leq \frac{c_1 \log ^2 \frac{1}{\delta}}{\sqrt{m}}.
$$
\end{theorem}

This theorem bounds the difference between the true risk of RFF regression and the minimum possible true risk in $\HC_k^C$ which is the set of all functions in the RKHS of kernel $k$ with RKHS norm less than a constant $C$. The theorem states that, if the number of frequency samples $D$ is larger than a threshold value of $\OC(\sqrt{m}\log(\sqrt{m}))$, then the true risk gap between RFF and the minimum true risk over $\HC_k^C$ decays with rate $\OC(m^{-1/2})$ with the number of data samples $m$.   

This theorem can be applied to a QK method by finding $C$ such that the true risk minimizer of the QK method falls inside $\HC_k^C$. To do so, we stablish a connection between the RKHS of a QK and the RKHS of the RFF kernel in the following lemma. 
Specifically, we provide a bound on the RKHS norm of $f$ with respect to two different kernels, which are associated with the same RKHS. 

\begin{lemma} \label{lem:AppDiagNonDiag} 
Consider a real-valued shift-invariant kernel $k_D$ with a diagonal Fourier transform $P$ (such that no diagonal element is zero) and a frequency support $\Omega$. Consider another real-valued kernel $k_Q$ with a non-diagonal Fourier transform $F$ on the same frequency support. That is, $k_D(\xv,\yv) = \zv(\yv)^\dagger P \zv(\xv)$ and $k_Q(\xv,\yv) = \zv(\yv)^\dagger F \zv(\xv)$, where $\zv$ is the trigonometric feature vector defined in Eq.~\eqref{eq:appzdef}. For any function $f$ that lies in the RKHS of both kernels, we have: 
\be 
\norm{f}_{k_D} \leq \norm{\sqrt{P}^{-1}\sqrt{F}}_\infty \norm{f}_{k_Q}\,.
\ee
\end{lemma}

\begin{proof}
For any set of frequencies $\Omega$, we are given a shift-invariant kernel $k_D$ with diagonal Fourier transform $P$ (such that no diagonal element is zero), and another kernel $k_Q$ with non-diagonal Fourier transform $F$. That is, $k_D(\xv,\yv) = \zv(\yv)^\dagger P \zv(\xv)$ and $k_Q(\xv,\yv) = \zv(\yv)^\dagger F \zv(\xv)$, where $\zv$ is the trigonometric feature vector defined in Eq.~\eqref{eq:appzdef}. These kernels' corresponding feature maps are $\phi_D = \sqrt{P}\zv$ and $\phi_Q = \sqrt{F}\zv$, respectively.

We recall that the definition of the RKHS norm of a function $f$ with respect to a kernel $k(x,y)= \expval{\phi(y), \phi(x)}$ is $\norm{f}_{k} = \min \norm{a}_2$ for $f(\cdot) = \expval{a, \phi(\cdot)}$. 
Now, we define the vector with the minimal RKHS norm for $k_Q$ as $a_Q^*$.
We then have 
    \be
    f(\cdot) = \expval{a_Q^*,\sqrt{F} \zv(\cdot)} = \expval{ \sqrt{F} a_Q^*,\zv(\cdot)}  =  \expval{\sqrt{P}^{-1} \sqrt{F} a_Q^*,\sqrt{P} \zv(\cdot)},
    \ee
    where we used the fact that $F$ is Hermitian because of the symmetry of the kernel. From the rightmost expression, we have that $f(\cdot) = \expval{a_D,\phi_D}$ with $a_D \coloneqq \sqrt{P}^{-1} \sqrt{F} a_Q^* $. 
    
    We remark that $P$ is invertible from the assumption that all diagonal elements are non-zero.
    As a result, we have
    \begin{equation}
        \norm{f}_{k_D} \leq \norm{a_D}_2 = \norm{\sqrt{P}^{-1} \sqrt{F} a_Q^*}_2 \leq \norm{\sqrt{P}^{-1}\sqrt{F}}_{\infty}  \norm{a_Q^*}_2 = \norm{\sqrt{P}^{-1}\sqrt{F}}_{\infty}  \norm{f}_{k_Q}.
    \end{equation}
\end{proof}

Combining Lem.~\ref{lem:AppDiagNonDiag} and Thm.~\ref{thm:appSweke1} together, we get the following theorem:
\begin{theorem}[Quantum Kernel Ridge Regression vs. Random Fourier Features] \label{thm:appKRR}
Consider a regression problem $\pi: \mathcal{X} \times \mathbb{R}$. Let $k_q: \mathcal{X} \times \mathcal{X} \rightarrow \mathbb{R}$ be a fidelity quantum kernel with Fourier transform $F$ and frequency support $\Omega$ (Def.~\ref{def:FT}). 
We also denote $f_q$ as the true risk minimizer of quantum kernel ridge regression with kernel $k_q$.

Let $P$ be a diagonal matrix whose elements are given by a distribution $p:\Omega \rightarrow \Rbb^+$ on that is not zero anywhere on its support. 
We then sample $D$ i.i.d. samples $(\omv_i)$ from the distribution ${p}$ and construct the following feature map: 
$$
\begin{aligned}
\phv_{D}(x)=\frac{1}{\sqrt{D}}\left[\cos\left(\omv_1 \cdot \xv\right), \cdots, \cos\left(\omv_D \cdot \xv\right),  \sin\left(\omv_1 \cdot \xv\right), \cdots, \sin\left(\omv_D \cdot \xv\right) \right].
\end{aligned}
$$
 Define $\|T_{k_p} \|$ as the spectral norm of the integral operator of this kernel (Def.~\ref{def:appKIO}). 
Supposing that there exists $S_1$ and $S_2$ such that,
 
1. 
\begin{equation} \label{eq:cond_1_qk_reg}
    \norm{f_q}_{k_q} \leq S_1 \ , 
\end{equation}
2. 
\begin{equation} \label{eq:cond_2_qk_reg}
    \norm{\sqrt{P}^{-1} \sqrt{F}}_{\infty} \leq S_2\ .
\end{equation}
We define: 
$$
\begin{aligned}
& \bar{B}\coloneq 2 b+2 \kappa \max \left\{1,S_1 S_2\right\} \\
& \bar{\sigma} \coloneq  2 b+2 \kappa \sqrt{\max \left\{1,S_1 S_2\right\}}
\end{aligned}
$$
and
$$
\begin{aligned}
& m_0\coloneq \max \left\{4\|T_{k_p} \|^2,\left(264 \kappa^2 \log \frac{556 \kappa^3}{\delta}\right)^2\right\}, \\
& c_0\coloneq 9\left(3+4 \kappa^2+\frac{4 \kappa^2}{\|T_{k_p} \|}+\frac{\kappa^4}{4}\right), \\
& c_1\coloneq 8\left(\bar{B} \kappa+\bar{\sigma} \kappa+\max \left\{1, S_1 S_2\right\}\right) ,
\end{aligned}
$$ 
where $\delta \in(0,1]$, $b$ and $\kappa$ are upper bounds for regression labels and the quantum kernel respectively. 
We choose $\lambda_m \coloneq 1 / \sqrt{m}$ where $m \geq m_0$ is the number of data samples.
We call $\hat{f}_{D, \lambda_m}$ the output of $\lambda_m$-regularized linear regression with respect to the feature map $\phv_D$.
Then, 
$$
D \geq c_0 \sqrt{m} \log \biggl(\frac{108 \kappa^2 \sqrt{m}}{\delta} \biggr)
$$
is enough to guarantee with probability at least $1-\delta$, that
$$
R\left(\hat{f}_{D, \lambda_m}\right)-R\left(f_q\right) \leq \frac{c_1 \log ^2 \frac{1}{\delta}}{\sqrt{m}}.
$$
\end{theorem}
\begin{proof}
    We are given a regression task, a quantum kernel $k_p$ with frequency support $\Omega$, a distribution $p$ and two real numbers $S_1$ and $S_2$ such that Eqs.~\eqref{eq:cond_1_qk_reg} and~\eqref{eq:cond_2_qk_reg} are true.

    From distribution $p$, we define the shift-invariant kernel $k_p(x,y) \coloneq \sum_{\omv \in \Omega} {p}(\omv) \cos(\omv \cdot (\xv - \yv))$. The Fourier transform of $k_p$ is a diagonal matrix $P$ with its diagonal elements given by the distribution $p$. $k_p$ and $k_q$ both have the same frequency support $\Omega$. Since, $P$ is non zero on the whole support, the RKHS of $k_p$ spans all the frequencies in $\Omega$ and therefore $f_q$ is in the RKHS of $k_p$. 
    
    Utilizing Lem.~\ref{lem:AppDiagNonDiag} together with the conditions on $p$ (Eqs.~\eqref{eq:cond_1_qk_reg}~\eqref{eq:cond_2_qk_reg}) we have 
    \be
    \norm{f_q}_{k_p} \leq S_1 S_2.
    \ee  
    Finally, by applying Thm.~\ref{thm:appSweke1} to $f_q$ and $k_p$, the proof is complete. 
\end{proof}

Similar to Thm.~\ref{thm:appSweke1}, this theorem implies a $1/\sqrt{m}$ convergence rate between the true risk of RFF and QSVM as long as $D$ is at least $\OC(\sqrt{m}\log(\sqrt{m}))$. To gain more intuition about how the scaling works with $D$, we can assume that we choose $D\in\OC(m)$. In this case the scaling of the gap with respect to $D$ would also be $\OC(1/\sqrt{D})$ suggesting that with large enough number of training samples, RFF matches the performance of QSVM as we increase $D$ (the number of frequency samples).

From Thm.~\ref{thm:appKRR}, we derive the number of data points $m$ and frequency samples $D$ to achieve the generalization gap $\epsilon$.
\begin{corollary} [Data and Frequency Samples Required for RFF Regression]\label{col:KRR}
    Given the setup and conditions of Thm.~\ref{thm:appKRR}, for every $\epsilon > 0, \delta \in (0,1)$, with probability at least $1-\delta$ we have: 
    \be 
R\left(\hat{f}_{D, \lambda_m}\right)-R\left(f_q\right) \leq \epsilon
    \ee
if 
\begin{align}
m &\geq  \left(\frac{c_1 \log^2  \frac{1}{\delta}}{\epsilon}\right) ^2 \\
    D&\geq c_0 \sqrt{m} \log \frac{108 \kappa^2 \sqrt{m}}{\delta} 
\end{align}
\end{corollary}

\begin{proof}
    The proof follows from Thm.~\ref{thm:appKRR} by bounding the RKHS of the true risk to be less than $\epsilon$.  
\end{proof}

\begin{corollary}[Sufficient Conditions for RFF Dequantization of QK Regression] \label{col:prop1}
    Consider a regression task using a quantum kernel $k_q$ with Fourier transform $F$. 
    Then, the linear ridge regression with the RFF feature map (Eq.~\eqref{eq:clrff2}) sampled according to distribution $p$ dequantizes the quantum kernel ridge regression if:  
    \begin{itemize}
        \item \textbf{Concentration:} $p_{\text{max}}^{-1} \in \OC(\poly(d))$ 
        \item \textbf{Alignment:} $\norm{\sqrt{P}^{-1} \sqrt{F}}_\infty \in \OC(\poly(d))$ where  $P$ is a diagonal matrix with $p$ as its diagonal
            \item \textbf{Bounded RKHS norm:} $\norm{f_q}_{k_q} \in \OC(\poly(d))$ where $f_q$ is the true risk minimizer of kernel regression.
    \end{itemize}
\end{corollary}

\begin{proof}
    We are given a quantum kernel $k_q$ and a distribution $p$ that satisfy the concentration, the alignment and the bounded RKHS norm conditions. As $\norm{f_q}_{k_q} \in \OC(\poly(d))$ and $\norm{\sqrt{P}^{-1} \sqrt{F}}_\infty \in \OC(\poly(d))$,  therefore there exists $S_1, S_2$ such that $S_1, S_2 \in \OC(\poly(d))$ and $\norm{f_q}_{k_q} \leq S_1$, $\norm{\sqrt{P}^{-1} \sqrt{F}}_\infty \leq S_2$. We are in the condition of the Thm~\ref{thm:appKRR}, and we can apply Corollary~\ref{col:KRR}.

    In the bound for the number of data samples $m$ appears $c_1$, using its definition in Thm.~\ref{thm:appKRR}, we know $c_1 \in \OC(S_1S_2)$. Consequently, $c_1 \in \OC(\poly(d))$. Note that we assume that $b$ (the maximum value of labels of the regression problem is constant in $d$). 

    In the bound for the number of frequency samples $D$ appears $c_0$, using its definition $c_0$ in Thm.~\ref{thm:appKRR} we know that it is related to the inverse of operator norm of the kernel integral operator $\|T_{k_p} \|$. In Lem.~1 of Ref.~\cite{sweke2023potential} it is shown to be proportional to $p_{\text{max}}^{-1}$. Therefore, using the properties of $p$, we have $p_{\text{max}}^{-1} \in \OC(\poly(d))$. Therefore $c_0 \in \OC(\poly(d))$. 

    This means that there exists $m,D \in \OC(\poly(d,\epsilon^{-1}))$, for which the conditions in Cor.~\ref{col:KRR} are satisfied and thus $R\left(\hat{f}_{D, \lambda_m}\right)-R\left(f_q\right) \leq \epsilon$ holds and the QK regression task is dequantized according to Def.~\ref{def:rff_dequantization}. 
\end{proof}

This concludes the proof of one of our main results (Prop.~\ref{prop:1}). 
Here, we provide additional comments on the $\norm{\sqrt{P}^{-1}\sqrt{F} }_1$, which quantifies how two matrices are aligned.

\begin{lemma}\label{lem:appFeP}
    Given a real-valued, positive definite, diagonal and unit-trace matrix $F$, $P=F$ minimizes $\norm{\sqrt{P}^{-1}\sqrt{F} }_1$ in the space of real-valued, positive definite, diagonal and unit-trace matrices. 
\end{lemma}
\begin{proof}
    The spectral norm is the largest singular value by definition. Since $F$ and $P$ are diagonal and positive, the singular values are directly the diagonal elements, which we call $\bm{f}, \bm{p}$ respectively. The minimization problem can then be written as: 
    \begin{align*}
        \min_{\bm{p}} & \max_i \sqrt{\frac{\bm{f}_i}{\bm{p}_i}} \\  
        \text{s.t} \quad & \bm{p}_i > 0 \quad \forall i \\
         \quad & \sum_i \bm{p}_i = 1. \\
    \end{align*}
    First, we note that, if $P=F$, then all $\frac{\bm{f}_i}{\bm{p}_i}$'s are equal to one and $ \min_{\bm{p}}  \max_i \sqrt{\frac{\bm{f}_i}{\bm{p}_i}}=1$. 
    Secondly, we show that for every $\bm{p}$ that satisfies the constraints of the optimization, we have $\frac{\bm{f}_i}{\bm{p}_i} \geq 1$ for one index $i$ at least. 
    We prove this by contradiction: Consider the opposite is true i.e., there exists a $\bm{p}$ such that for all $i$,  $\frac{\bm{f}_i}{\bm{p}_i} < 1$ and satisfies the constraints. If for all $i$,  $\frac{\bm{f}_i}{\bm{p}_i} < 1$, then ${\bm{f}_i} < {\bm{p}_i}$ and $\sum_i \bm{p}_i > \sum_i \bm{f}_i = 1$ which is a contradiction. So for every $\bm{p}$ that satisfies the constraints of the optimization, the maximum of $\frac{\bm{f}_i}{\bm{p}_i}$ is greater or equal to 1. Finally, we conclude that $P=F$ is the minimizer. 
\end{proof}

\section{Bounds for RFF-SVM (Alg.~\ref{alg:mainSVM})} \label{app:QKSVM}

In this section, we establish the foundation for the main results. 
Specifically, we provide the error bound on the true risk for Alg.~\ref{alg:mainSVM}, based on the results for the so-called  kitchen sink algorithm by Ref.~\cite{rahimi2008weighted}. 

\begin{algorithm} 
\caption{Fitting Procedure of the Weighted Sum of Random Kitchen Sinks~\cite{rahimi2008weighted}}.\label{alg:RR}
\begin{algorithmic}

\STATE\textbf{Input:} A dataset $\{\xv_i,y_i\}_{i=1}^m$ of $m$ points, a bounded feature function $|\phi(\xv, \wv)| \leq 1$, an integer
$D$, a scalar $C$, and a probability distribution $p$ on second parameters of $\phi$, $\wv$. A $L$-Lipschitz loss function $\LC$. 
\STATE \textbf{Output:} A function $\hat{f}(\xv) = \sum_{j=1}^D \alpha_j\phi(\xv, \wv_j)$
\STATE \textbf{1:} Draw $\wv_1,\cdots,\wv_D$ i.i.d. from $p$.
\STATE \textbf{2:} Featurize the input: $\zv_i= (\phi(\xv_i, \wv_1), \cdots, \phi(\xv_i, \wv_D))^T$.
\STATE \textbf{3:} With $\wv$ fixed, solve the empirical risk minimization problem: 
\begin{align*}
\min_{\alphav \in \Rbb^D} \quad & \frac{1}{m}\sum_{i=1}^m c(\alphav^T \zv_i,y_i) \\
s.t. & \quad \norm{\alphav}_\infty \leq C/D
\end{align*}
\end{algorithmic}
\end{algorithm}

To begin with, we present the performance guarantee on Alg.~\ref{alg:RR} introduced in Ref.~\cite{rahimi2008weighted}. 
\begin{theorem}[Main Result of Ref.~\cite{rahimi2008weighted}] \label{thm:RR}
    Let $p$ be a distribution on a set $\Omega$ and $\phi$ be a feature function that satisfies $\sup_{\wv,\xv} |\phi(\xv, \wv)| \leq 1$. Then, we define the set
    \be
        \FC_{p,C} = \{f(\xv)= \int_\Omega \alpha(\wv)\phi(\xv,\wv) d\wv :|\alpha(\wv)|\leq Cp(\wv)\},
    \ee
    with a constant $C$.
    Suppose a $L$-Lipschitz loss function such that its value depends on the multiplication of the two inputs i.e. $\LC(y, y_0) = \LC(yy_0)$. Then for any $\delta > 0$, if the training data
$\{\xv_i, y_i\}_{i=1}^m$ are drawn i.i.d. from some distribution , Alg.~\ref{alg:RR} returns a function $\hat{f}$ that satisfies: 
\be
R[\hat{f}]- \min_{f \in \FC_{p,C}} \ R[f] \in \OC \biggl ( \biggl ( \frac{1}{\sqrt{m}} + \frac{1}{\sqrt{D}} \biggr ) LC \sqrt{\log(\frac{1}{\delta})}\biggr) 
\ee
with probability at least $1-2\delta$ over the training set and the choice of $D$ samples $\wv_1,\cdots, \wv_D$.
\end{theorem}

This theorem states that the performance of Alg.~\ref{alg:RR} cannot be much worse than the best possible true risk within a set of hypotheses $\FC_{p,C}$. In this appendix, we show that our proposed algorithm for RFF-SVM (Alg.~\ref{alg:RFF}) is a special case of Alg.~\ref{alg:RR} introduced above.  We state the adapted generalization bounds for Alg.~\ref{alg:mainSVM} in Cor.~\ref{col:hardm}. 

Later in App.~\ref{app:main results}, we employ Cor.~\ref{col:hardm} to prove our dequantization results. We set $C$ such that the true risk optimizer of the QK-SVM is inside $\FC_{p,C}$. Consequently, we bound the difference between the true risk of RFF and the minimum true risk achieved by QK-SVM and prove our dequantization results (Def.~\ref{def:apprff_dequantization}).  

Here we bring the main idea for the proof of Thm.~\ref{thm:RR} in Ref.~\cite{rahimi2008weighted}. The output of the Alg.~\ref{alg:RR} is in the set,  
\be \label{eq:Fw}
\hat{\mathcal{F}}_w \equiv\left\{f(\xv)=\sum_{k=1}^D \alphav_k \phi\left(\xv ; \wv_k\right) \ \text{s.t.} |\alphav_k| \leq \frac{C}{D}\right\} .
\ee
We denote by $f^*$ the true risk minimizer over $\FC_{p,C}$ and by $\hat{f}$ the empirical risk minimizer over  $\hat{\FC}_w$, which corresponds to the output of the algorithm.
Also, let $\hat{f}^*$ be the true risk optimizer over $\hat{\FC}_w$. 
Then, the following inequality can be derived;
\be
R[\hat{f}] -R[f^*] = R[\hat{f}] -R[\hat{f}^*] + R[\hat{f}^*] -R[f^*] \leq |R[\hat{f}] -R[\hat{f}^*]|  + R[\hat{f}^*] -R[f^*].
\ee

Based on the inequality, Thm.~\ref{thm:RR} bounds the two terms above separately. 
\begin{itemize}
    \item For the first term Ref.~\cite{rahimi2008weighted} uses an estimation error bound. They show that the true risk and empirical risk of every function in $\hat{\FC}_w$ are close. (Lem.~3 of Ref.~\cite{rahimi2008weighted})
    \item  For the second term Ref.~\cite{rahimi2008weighted} applies an approximation error bound. They show that the lowest true risk reached by $\hat{\FC}_w$ is not much larger from the lowest true risk of $\FC_{p,C}$ (Lems.~1 and 2 of Ref.~\cite{rahimi2008weighted}).
\end{itemize}
By putting these two bounds together Ref.\cite{rahimi2008weighted} proves Thm.\ref{thm:RR}. 

Next, we slightly modify Alg.~\ref{alg:RR} to better align with RFF.
As shown in Eq.~\eqref{eq:kp} in the main text, RFF conventionally utilizes the random feature maps consisting of $\cos(\omv_i \xv)$ and $\sin(\omv_i \xv)$ for the same $\omv_i$.
On the other hand, Alg.~\ref{alg:RR} deals with only one feature function $\phi(\xv,\omv)$.
Thus, we extend Alg.~\ref{alg:RR} to the case for two feature functions $\phi_1(\xv,\omv)$ and $\phi_2(\xv,\omv)$ (see Alg.~\ref{alg:RFF}). 
\begin{algorithm} 
\caption{Random Features}\label{alg:RFF}
\begin{algorithmic}
\STATE \textbf{Input:} A dataset $\{\xv_i,\yv_i\}_{i=1}^{m}$ of $m$ points, two bounded feature functions $|\phi_1(\xv, \bm{w})| \leq 1$ and $|\phi_2(\xv, \bm{w})| \leq 1$, an integer $D$, a scalar $C$, and a probability distribution $p(\bm{w})$ on the parameters of $\phi_1$ and $\phi_2$. A $L$-Lipschitz loss function $\LC$. 
\STATE{\textbf{Output:} A function $\hat{f}(\xv) = \sum_{j=1}^D \alphav_j\phi_1(\xv, {\wv}_j) + \alphav_{j+D} \phi_2(\xv,{\wv}_j)$} 
\STATE{\textbf{1:} Draw $W = \{ \wv_1,\cdots,\wv_j\}$ i.i.d. from $p$.}
\STATE{\textbf{2:} Featurize the input: $\bm{z}_i= (\phi_1(\xv_i, \wv_1), \cdots, \phi_1(\xv_i, \wv_D),\phi_2(\xv_i, \wv_1), \cdots, \phi_2(\xv_i, \wv_D) )^T$.}
\STATE{\textbf{3:} With $W$ fixed, solve the empirical risk minimization problem: 
\begin{align*}
\min_{\alphav \in \Rbb^D} \quad & \frac{1}{m}\sum_{i=1}^m \LC(\alphav^T \zv_i,y_i) \\
s.t. & \quad \norm{\alphav}_\infty \leq C/D
\end{align*}}
\end{algorithmic}
\end{algorithm}

Now, we show the performance guarantee on Alg.~\ref{alg:RFF}.
\begin{theorem}[Performance of Alg.~\ref{alg:RFF}] \label{thm:RFF}
    Let $p$ be a distribution on a set $\Omega$ and $\phi_1$ and $\phi_2$ be feature functions that satisfy $\sup_{\bm{w},\xv} |\phi_i(\xv, \bm{w})| \leq 1 ,i=1,2$. Then, we define the set
    \be \label{eq:fpc}
        \FC_{p,C,\phi_1,\phi_2} = \{f(\xv)= \int_\Omega \alpha_1(\bm{w})\phi_1(\xv,\bm{w}) +\alpha_2(\bm{w}) \phi_2(\xv,\bm{w}) d\bm{w} :|\alpha_i(\bm{w})|\leq Cp(\bm{w}), i=1,2\}.
    \ee
    with a constant $C$. 
    Suppose a $L$-Lipschitz loss function $\LC(y; y_0) = \LC(yy_0)$. Then for any $\delta > 0$, if the training data $\{\xv_i, y_i\}_{i=1}^m$ are drawn i.i.d. from some distribution $P$, Alg.~\ref{alg:RFF} returns a function $\hat{f}$ that satisfies
\be
R[\hat{f}]- \min_{f \in \FC_{p,C,\phi_1,\phi_2}} \ R[f] \in \OC \biggl ( \biggl ( \frac{1}{\sqrt{m}} + \frac{1}{\sqrt{D}} \biggr ) LC \sqrt{\log(\frac{1}{\delta})}\biggr) 
\ee
with probability at least $1-2\delta$ over the training set and the choice of $D$ samples  $\wv_1,\cdots, \wv_D$.
\end{theorem}

To prove this theorem we use the following Lemma which are extensions of Lem.~1 and Lem.~2 in Ref.~\cite{rahimi2008weighted}.
\begin{lemma}[Extension of Lemma 1 from Ref.~\cite{rahimi2008weighted}] \label{lem:1}
Consider  
\be 
\hat{\mathcal{F}}_w \equiv\left\{f(x)=\sum_{j=1}^D \alpha_{j} \phi_1\left(\xv , w_j\right)+ \alpha_{j+D} \phi_2\left(\xv, w_j\right) : | \alpha_{k} |\leq \frac{C}{D},\ \forall k\in\{1,\ldots,2D\}\right\} .
\ee
Let $\mu$ be a measure on $\XC$, and $f^*$ be a function in $\FC_{p,C,\phi_1, \phi_2}$ (Defined in Eq.~\eqref{eq:fpc}). If $\{w_1,\cdots, w_D\}z$ are drawn i.i.d. from $p$,
then for any $\delta >0 $, with probability at least $1-\delta$ over $\{w_1,\cdots, w_D\}$, there exists a function $\hat{f} \in \hat{\FC}_w$ so that 
\be
\sqrt{\int_{\mathcal{X}}\left(\hat{f}(x)-f^*(x)\right)^2 d \mu(x)} \leq \frac{2C}{\sqrt{D}}\left(1+\sqrt{2 \log \frac{1}{\delta}}\right) .
\ee
\end{lemma}
\begin{proof}
    Lem.~1 in Ref.~\cite{rahimi2008weighted} focuses on only one feature function $\phi(x,w)$, whereas we here extend it to two feature functions, i.e., $\phi_1$ and $\phi_2$. 
    We prove this Lemma by following the steps of the proof for Lem.~1 in Ref.~\cite{rahimi2008weighted}. 
    Since $f^* \in \FC_{p,C,\phi_1, \phi_2}$, the true risk minimizer can be expressed as $$f^* (x) = \int_\Omega \alpha_1(w)\phi_1(x, w) + \alpha_2(w)\phi_2(x, w) dw .$$ Now we consider the functions $f_j(\cdot) =c_j \phi_1(\cdot, w_j) + d_j \phi_2(\cdot, w_j), j = 1, \cdots, D$, with $c_j = \frac{\alpha_1(w_j)}
{p(w_j)} ,  \ d_j = \frac{\alpha_2(w_j)}
{p(w_j)} $, so that $\Ebb[f_j] = f^*$. Let $\hat{f}(x) = \frac{1}{D}\sum_{j=1}^D f_j$ be the sample average of these functions. Then $\hat{f} \in \hat{\FC}_w$ because $|c_j|/D \leq C/D,\  |d_j|/D \leq C/D$. 
Also,  $\norm{c_j \phi_1(\cdot \ , w_j)+ d_j \phi_2(\cdot \ , w_j)} \leq \norm{c_j \phi_1(\cdot \ ,  w_j)} + \norm{d_j \phi_2( \cdot \ , w_j)} \leq 2C$. 
On the other hand, assuming $K$ i.i.d. random variables in a Hilbert space ($X = \{x_1,\cdots, x_K\}$), Lem.~4 of Ref.~\cite{rahimi2008weighted} bounds the distance between the average ($\bar{X} = \frac{1}{K}\sum x_k$) with the mean i.e. $\|\bar{X}- \Ebb[\bar{X}]\|$, with high probability ($1-\delta$).
The proof is completed by applying Lem.
4 of Ref.~\cite{rahimi2008weighted} to $f_1,\cdots, f_D$. 
\end{proof} 

\begin{lemma} (Bound on the approximation error (Lemma 2 of Ref.~\cite{rahimi2008weighted}))\label{lem:approx} 
Suppose $c\left(y, y^{\prime}\right)$ is L-Lipschitz in its first argument. Let $f^*$ be a fixed function in $\mathcal{F}_p$. If $w_1, \ldots, w_K$ are drawn i.i.d from $p$, then for any $\delta>0$, with probability at least $1-\delta$ over $w_1, \ldots, w_K$, there exists a function $\hat{f} \in \hat{\mathcal{F}}_w$ that satisfies
$$
\mathbf{R}[\hat{f}] \leq \mathbf{R}\left[f^*\right]+\frac{2L C}{\sqrt{D}}\left(1+\sqrt{2 \log \frac{1}{\delta}}\right)
$$
\end{lemma}
\begin{proof}
The proof is identical to the proof provided in Ref.~\cite{rahimi2008weighted}. For any two functions $f$ and $g$, the Lipschitz condition on $c$ followed by the concavity of square root gives
$$
\begin{aligned}
\mathbf{R}[f]-\mathbf{R}[g] & =\mathbb{E} c(f(x), y)-c(g(x), y) \leq \mathbb{E}|c(f(x), y)-c(g(x), y)| \\
& \leq L \mathbb{E}|f(x)-g(x)| \leq L \sqrt{\mathbb{E}(f(x)-g(x))^2}
\end{aligned}
$$

The lemma then follows from Lem.~\ref{lem:1}.
\end{proof}

Now we proceed to prove Thm.~\ref{thm:RFF}
\begin{proof}    
   We prove this result by adapting the proof in Ref.~\cite{rahimi2008weighted} with a slight modification.
    There are two minor changes: one in Lem.~1,2 and the other in Thm.~2 of the original proof.
    The modified version of the former is already stated in~Lems.~\ref{lem:1},~\ref{lem:approx}.
    As for Thm.~2 of Ref.~\cite{rahimi2008weighted}, the Rademacher complexity is bounded by $2C/\sqrt{m}$ instead of $C/\sqrt{m}$, as we consider two-dimensional feature functions bounded by 1, i.e., $\phi_1$ and $\phi_2$.
    With these modification, we adapt the proof in Ref.~\cite{rahimi2008weighted} to our case.
\end{proof}

We note that Thm.~\ref{thm:RFF} can apply to a continuous probability distribution $p$, whereas the discrete distribution is used for RFF in practice. 
However, the proof of Thm.~\ref{thm:RR} in Ref.~\cite{rahimi2008weighted} can be extended to the discrete case.

Finally, we build upon Alg.~\ref{alg:RFF} and provide Alg.~\ref{alg:SVM}, which corresponds to  RFF-SVM (Alg.~\ref{alg:mainSVM} in the main text).
Concretely, we specify the loss function $\LC$ as the regularized hinge loss in Eq.~\eqref{eq:hinge_pb} and set the feature functions as $\phi_1(\xv,\omv) = \cos(\bm{w}\xv), \ \phi_2(\xv,\bm{w}) = \sin(\bm{w}\xv)$. 
Similar to Alg.~\ref{alg:RR} and Alg.~\ref{alg:RFF}, we present the performance guarantee of Alg.~\ref{alg:SVM}.

\begin{algorithm}[t]
\caption{RFF-SVM (Alg.~\ref{alg:mainSVM})} \label{alg:SVM}
\begin{algorithmic}
\STATE {\bfseries Input:} A training dataset $\{\xv_i,y_i\}_{i=1}^m$, an integer
$D$, a scalar $C$ as search radius, a scalar $\lambda$ as regularization parameter and a probability distribution $p$ on a set $\Omega$.
\STATE{\bfseries Output:} A function  $$\hat{f}(\xv) = \sum_{j=1}^D \betav_j \cos(\omv \cdot \xv) + \betav_{j+D} \sin(\omv \cdot \xv)$$
\noindent \hrule
\vspace{5pt}
\STATE \textbf{1:} Draw $\omv_1,\cdots,\omv_D$ i.i.d. from $p$
\STATE \textbf{2:} Construct the feature map: 
$\zv(\xv_i)= \sqrt{D} \phv_D(\xv)$ with $\phv_D(\xv)$ defined in Eq.~\eqref{eq:clrff2}.
\STATE \textbf{3:} With $\omv$'s fixed, solve the empirical risk minimization problem: 
\begin{align*}
\min_{\betav \in \Rbb^D} \quad &  \frac{\lambda}{2}\norm{\betav}_2^2 + \frac{1}{m}\sum_{i=1}^m \max\{0,1-\betav^T \zv(\xv_i)y_i\} \\ 
s.t. & \quad \norm{\betav}_\infty \leq C/D
\end{align*}
\end{algorithmic}
\end{algorithm}

\begin{corollary}[Performance of Alg.~\ref{alg:SVM} (same as Alg.~\ref{alg:mainSVM})] \label{col:hardm}
    Let $p$ be a distribution on $\Omega$, which is discrete with finite cardinality. 
    We also define the set:
    \be \label{eq:gdef}
        \GC_{p,C} \coloneq \{f(\xv)= \sum_{\omv \in \Omega}  \alpha(\omv) \cos(\omv \cdot \xv)+ \beta(\omv)\sin(\omv \cdot\xv)  :|\beta(\omv)| , |\alpha(\omv)|\leq Cp(\omv)\}
    \ee
 Then, for any $\delta > 0$, if the training data
$\{\xv_i, y_i\}_{i=1}^m$ are drawn i.i.d. from a distribution, applying 
Alg.~\ref{alg:mainSVM}
to this dataset with a regularization constant $\lambda$, a number of frequency samples $D$ and a search radius C returns a function $\hat{f}$ that satisfies: 
\be
R[\hat{f}]- \min_{f \in \GC_{p,C}} \ R[f] \in \OC \biggl ( \biggl ( \frac{1}{\sqrt{m}} + \frac{1}{\sqrt{D}} \biggr ) C \sqrt{\log(\frac{1}{\delta})} +  \lambda \frac{C^2}{D}\biggr) 
\ee
with probability at least $1-2\delta$ over the training set and frequency samples.
\end{corollary}
\begin{proof}
   We apply Thm.~\ref{thm:RFF} to get this result.
First, we note that the hinge loss i.e. $c(xy) = \max\{0,1-xy\}$ is $L$-Lipschitz with $L=1$ and therefore satisfies the assumption for the loss function in Thm.~\ref{thm:RR}. Moreover, for $\lambda = 0$, Alg.~\ref{alg:SVM} is similar to Alg.~\ref{alg:RFF} with $\phi_1(\cdot) = \sin(\cdot), \ \phi_2(\cdot) = \cos(\cdot)$. In other words, we have  $\GC_{p,C} = \FC_{p,C,\cos, \sin}$. 
    Hence, Thm.~\ref{thm:RFF} is directly applicable to the case of $\lambda = 0$. 
    When $\lambda \neq 0$, there is an additional regularization term $\lambda/2 \norm{\alpha}_2^2$ in the empirical risk. This term is bounded by $\lambda/2 \norm{\alpha}_2^2 \leq \lambda/2 C^2/D$ because of the condition on $\norm{\alpha}_\infty$. This additional bound can be added to Lem.~3 of Ref.~\cite{rahimi2008weighted}, which is about bounding the generalization error of every function in $\hat{F}_w$ (Eq~\eqref{eq:Fw}). The rest of the proof for Thm.~\ref{thm:RFF} remains unchanged. 
    Therefore, by incorporating this bound on the regularization term $\lambda C^2/D$ into Thm.~\ref{thm:RFF}, we can derive the final bound. 
    \end{proof}

Cor.~\ref{col:hardm} indicates that we can maintain the $\frac{1}{\sqrt{D}}$ growth of the bound by choosing $\lambda < \sqrt{D}/C$.
Recall that $\lambda=0$ corresponds to the hard margin SVM, while $\lambda\neq0$ indicates the soft margin SVM (see App.~\ref{app:SVM}).

\section{Dequantization of QSVM and QNN-SVM}\label{app:main results}

In this section, we present the detailed proofs of our main results, Prop.~\ref{prop:loose} as Prop.~\ref{prop:apploose} and Prop.~\ref{prop:tight} as  Prop.~\ref{prop:tightapp}.
\subsection{Proof of Prop.~\ref{prop:loose}}
We start with the sufficient conditions for RFF-dequantization of QSVM (Prop.~\ref{prop:loose}).
First, we recall the setup of the QSVM functions.

 \begin{definition}[Quantum SVM True Risk Minimizer] \label{def:setup}
Let $\{\xv_i, y_i\}_{i=1}^m$ be training data drawn i.i.d. from a distribution $\pi: \XC \times \{-1,1\}\rightarrow \Rbb, \quad \XC \subset \Rbb^d$.
We consider a quantum kernel using Hamiltonian encoding $k_q$ with frequency support $\Omega$, positive frequency support (Def.~\ref{def:PositiveFreqSup}) $\Omega_+$ and its Fourier transform $F$ as in Def.~\ref{def:FT}.
We also formulate the classification problem in its dual form defined in Eq.~\eqref{eq:dual}. Define $f_q$ the true risk minimizer of the QSVM problem in the RKHS of the kernel denoted with $\HC_{k_q}$ i.e. $f_q:= \min_{f\in \HC_{k_q}} \RC[f]$. Since $f_q$ is in the RKHS of the quantum kernel it can be written as:

\begin{align}\label{eq:fq:dec}
     f_q &= \Tilde{c}_0 + \sum_{j=1}^{|\Omega_+|} \Tilde{c}_j e^{i\omv_j \cdot \xv} + \Tilde{c}^*_j e^{-i\omv_j \cdot \xv} \nonumber \\ 
    &= \Tilde{c}_0 + \sum_{j=1}^{|\Omega_+|} 2 \Re(\Tilde{c}_j) \cos(\omv_j \cdot \xv) -2\Im (\Tilde{c}_j) \sin(\omv_j \cdot \xv) 
\end{align}
where the $\tilde{c}_{\omv_j}$ are complex numbers. 
\end{definition}

With this setting, we prove the gap of the generalization performance between QSVM and RFF. 
\begin{theorem}[Performance Comparison RFF vs. QSVM] \label{thm:mainapp}
    Assume $f_q$ to be the true risk optimizer of QSM with a quantum kernel $k_p$ (the setup of Def.~\ref{def:setup}). 
    Further, assume we are given a distribution $\tilde{p}:\Omega_+ \cup \{0\}  \rightarrow \Rbb$. From $\tilde{p}$ define another distribution $p:\Omega \rightarrow \Rbb$ such that $p_0= \Tilde{p}_0, \quad \forall \omv \in \Omega_+: p_{\omv} = p_{-\omv} = \tilde{p}_{\omv}/2$. If there exists $C_1, C_2 \in \Rbb$ such that 
    \begin{align} \label{eq:QSVMcond}
        \min_{\omv} \frac{p_{\omv}}{\sqrt{F_{{\omv},\omv}}}  &\geq \frac{1}{C_1} \\
        \norm{f_q}_{k_q} &\leq C_2  \label{eq:QSVMcond2}
    \end{align}
    then the following is true: 
    \be \label{eq:temp3app}
    \RC[\hat{f}] - \RC[f_q] \in \OC \biggl ( \biggl ( \frac{1}{\sqrt{m}} + \frac{1}{\sqrt{D}} \biggr ) C_1C_2 \sqrt{\log(\frac{1}{\delta})}\biggr) .
    \ee 
    where $\hat{f}$ is the output of Alg.~\ref{alg:mainSVM} with the number of frequency samples $D$, search radius $C = C_1C_2$ and distribution $\tilde{p}$ as inputs.
\end{theorem}
\begin{proof}
We are given a SVM task, a quantum kernel $k_q$, the true risk optimizer of the QSVM (Def.~\ref{def:setup}), two reals $C_1$ and $C_2$, and a distribution $\tilde{p}$ such that Eqs.~\eqref{eq:QSVMcond} and~\eqref{eq:QSVMcond2}  are satisfied. Cor.~\ref{col:hardm} bounds the generalization gap between the output of Alg.~\ref{alg:mainSVM} and the true risk minimizer in set $\GC_{\tilde{p},C}$ where $C=C_1C_2$ as follows
    \be 
     \RC[\hat{f}] - \min_{f \in \GC_{\tilde{p},C}} \RC[f] \in \OC \biggl ( \biggl ( \frac{1}{\sqrt{m}} + \frac{1}{\sqrt{D}} \biggr ) C_1C_2 \sqrt{\log(\frac{1}{\delta})}\biggr),   
\ee
where  $\GC_{\tilde{p},C}$ is defined in Eq.~\eqref{eq:gdef}, i.e.,
$$\GC_{p,C} = \{f(\xv)= \sum_{\omv \in \Omega}  \alpha(\omv) \cos(\omv \cdot\xv)+ \beta(\omv)\sin(\omv \cdot \xv)  :|\beta(\omv)| , |\alpha(\omv)|\leq Cp(\omv)\}.$$
If we have $f_q \in \GC_{\tilde{p},C}$, then $\RC[\hat{f}] - \RC[f_q] \leq \RC[\hat{f}] - \min_{f \in \GC_{\tilde{p},C}} \RC[f]$ by definition and the upper bound of Cor.~\ref{col:hardm} will also be an upper bound for $\RC[\hat{f}] - \RC[f_q]$. Thus, to prove the result, we show that the QSVM decision function $f_q $ is in $\GC_{\tilde{p},C}$.

First, according to Cor.~\ref{prop:qk}, $k_q(\xv,\yv) =  \zv^\dagger(\yv) LL^\dagger \zv(\xv)$ with $F = LL^\dagger$. We define $\Phi: \XC\rightarrow \HC_0$ as $\Phi(\xv) = L^\dagger \zv(\xv) $ where $\HC_0$ is a $|\Om|$ dimensional Hilbert space with Euclidean norm. Then, $k_q = \expval{\Phi(\yv),\Phi(\xv)}$ with the corresponding  feature map $\Phi(\xv)$.  As $f_q$ is in the RKHS of $k_q$, there exist at least one  $\av \in \HC_0$ such that 
    \be \label{eq:temp2app}
        f_q(\xv) = \expval{\av,\Phi(\xv)}_{\HC_0} .  
    \ee
Among all possible $\av$'s such that $f_q(\xv) = \expval{\av,\Phi(\xv)}_{\HC_0}$, we define $\av^\star$ to be the one with the least $l_2$ norm. In this case, the definition of RKHS norm (Def.~\ref{def:RKHS}) reveals that
    \be \label{eq:temp5app}
    \norm{f_q}_{k_q} = \norm{\av^\star}_2. 
    \ee
Therefore we have
    \be \label{eq:temp21app}
        f_q(\xv) = \expval{\av^\star,\Phi(\xv)}_{\HC_0} = \expval{L\av^\star,\zv(\xv)}  .
\ee
On the other hand, the Fourier representation of $f_q$ in Eq.~\eqref{eq:fq:dec} shows that 
    \be \label{eq:temp22app}
       f_q(\xv) = \expval{\bm{c},\zv(\xv)} \ ,
    \ee
    where $\bm{c} = [\Tilde{c}_0,\Tilde{c}_1, \Tilde{c}^*_1, \cdots, \Tilde{c}_{|\Omega_+|}, \Tilde{c}^*_{|\Omega_+|}]^T$ and the feature map $\zv$ defined in Eq.~\eqref{eq:zdef}.
    By comparing  Eq.~\eqref{eq:temp21app} and Eq.~\eqref{eq:temp22app}, we conclude that 
    \be \label{eq:temp23app}
    \bm{c} = L\av^\star.  
    \ee    
Now, we start from the assumptions in Eqs.~\eqref{eq:QSVMcond} and~\eqref{eq:QSVMcond2} and show that $f_q \in \GC_{\tilde{p},C_0}$. 
We denote $j$'th row of $L$ with $L_j^\dagger$. From the assumptions we have
     \begin{align} 
    C \min_{\omv} \frac{p_{\omv}}{\sqrt{F_{\omv, {\omv}}}}& \geq \frac{C}{C_1} = C_2
     \geq \norm{f_q}_{f_q}  \nonumber \\ \Rightarrow
    C \min_{\omv} \frac{p_{\omv}}{\sqrt{F_{{\omv}, {\omv}}}} & \geq \norm{f_q}_{f_q} \overset{(1)}{=} \norm{\av^\star}_2   \nonumber \\\Rightarrow
     C \frac{p_{\omv_i}}{\sqrt{F_{{\omv}_i, {\omv}_i}}} & \geq \norm{\av^\star}_2 \quad \forall  i= 1,2,\cdots, |\Omega|  \nonumber \\ \Rightarrow 
    C p_{\omv_i}  \geq \sqrt{F_{{\omv}_i, {\omv}_i}}\norm{{\av^\star}}_2 &\overset{(2)}{=} \norm{L_i}_2\norm{{\av^\star}}_2  \overset{(3)}{\geq} |\expval{L_i,\av^\star}| \overset{(4)}{=} |{\bm{c}_i}| \ .
    \label{eq:ineq}
\end{align}
Here, Eq.~\eqref{eq:temp21app} is used in step (1) and step (2) utilizes $F= LL^\dagger$ and $F_{i,i} = L_i^{\dagger} L_i = \norm{L_i}_2^2$. 
In step (3), Cauchy-Schwartz inequality is used.
Lastly, we use Eq.~\eqref{eq:temp23app} in step (4).

Note that the index $i$ in the last inequality of Eq.~\eqref{eq:ineq} runs from $1$ to $|\Omega|$. 
For frequency $0$, we simply have $C\Tilde{p}(0) \geq |\tilde{c}_0|$.
Otherwise, by summing the two corresponding inequalities of $\omv_j, -\omv_j$ for $\omv_j \in \Omega_+$, we get $C(p_{\omv_j}+p_{-\omv_j}) \geq 2|\Tilde{c}_j|$. Since $\Tilde{p}_{\omv} = p_{\omv}+ p_{-\omv}$ by definition, we have $C\Tilde{p}_{\omv_j} \geq 2|\Tilde{c}_j|$ for all $\omv_j \in \Om_+$. 

Now, we show that, if we have $2|\Tilde{c}_j| \leq C \Tilde{p}_{\omv_j}$ for all $\omv_j \in \Omega_+$ and $|\tilde{c}_0| \leq C\tilde{p}_0 $, then $f_q \in \GC_{\tilde{p},C}$.  
According to the definition of QSVM decision function $f_q$ in Eq.~\eqref{eq:fq:dec}, if we have $2|\Tilde{c}_j| \leq C\tilde{p}_{\omv_j}$ for all $j$ from $1$ to $|\Omega_+|$, then $|2 \Re(\Tilde{c}_j)|, |2 \Im(\Tilde{c}_j)| \leq C \tilde{p}_{\omv_j}$, meaning $f_q \in \GC_{\tilde{p},C}$ by definition. 
This completes the proof. 
\end{proof} 
This theorem states that the gap between the performance of RFF dequantization and QSVM scales $\OC(1/\sqrt{D}+ 1/\sqrt{m})$ with the number of data $m$ and frequency samples $D$. As a quick sanity check, we look at asymptotic limits. As discussed in App.~\ref{app:RFF}, RFF method estimates the kernel using $D$ random samples. Therefore, as $D$ increases, RFF-SVM converges to the kernel SVM method. In limit of large $D$, as RFF converges to kernel SVM, this bound reaches a generalization error of $\OC(1/\sqrt{m})$ which matches the $\OC(1/\sqrt{m})$ generalization error of kernel SVM.
Next, by extending Thm.~\ref{thm:mainapp}, we derive the condition of data points and frequency samples to achieve the precision $\epsilon$.
\begin{corollary} \label{col:epsilonapp}
    Given the settings of Thm.~\ref{thm:mainapp}, if there exists $C_1,C_2$ such that we have
    \begin{align}
        \min_{\omv} \frac{p_{\omv}}{\sqrt{F_{{\omv},\omv}}}  &\geq \frac{1}{C_1} \\
        \norm{f_q}_{k_q} &\leq C_2 
    \end{align}
    then to have $\RC[\hat{f}] - \RC[f_q] \leq \epsilon$ with probability $1-2\delta$ over data and frequency samples, it is sufficient to choose the number of frequency samples $D$ and data points $m$ such that
    \be
    \biggl(\frac{1}{\sqrt{m}}+\frac{1}{\sqrt{D}} \biggr)^{-1} \in \bm{\Omega}   \biggl ( C_1C_2 \sqrt{\log(\frac{1}{\delta})} /\epsilon \biggr)  
    \ee
     Note that $\bm{\Omega}$ denotes 'Big-$\Omega$' from complexity theory.
\end{corollary}
\begin{proof}
    The proof follows from Thm.~\ref{thm:mainapp}. 
 For the right-hand side of Eq.~\eqref{eq:temp3app} to be less than or equal to $\epsilon$, i.e., 
    \begin{equation}
        \OC \biggl ( \biggl ( \frac{1}{\sqrt{m}} + \frac{1}{\sqrt{D}} \biggr ) C_1C_2 \sqrt{\log(\frac{1}{\delta})}\biggr) \le \epsilon,
    \end{equation}
    we re-express the equation and obtain the bound on data and frequency samples.
\end{proof}

For convenience, we replace diagonal elements of $F$ with $q$, by following the notation of Def.~\ref{def:FT}. 
Hereafter, since we consider the case where both distributions $p$ and $q$ have the same frequency support $\Omega$, we omit explicit notation for $\omv$; that is, we simply denote $p_{\omv_j}$ and $q_{\omv_j}$ as $p_j$ and $q_j$, respectively.

Cor.~\ref{col:epsilonapp} states that 
the efficiency of the RFF-dequantization method for kernel SVM is determined by the scaling of $C_1$  and $C_2$, which implicitly depend on the number of qubits or dimension of input data.
As a result, if both of the following variables scale polynomially in these parameters, the quantum kernel SVM algorithm can be dequantized via RFF.
More concretely, our focus is on the following quantities;
\begin{itemize}
    \item Inverse of the minimum ratio of distributions $p$ and square root of distribution $q$: $\min_{i} \frac{p_i}{\sqrt{q_i}}$ 
    \item RKHS norm of the quantum decision function: $\norm{f_q}_{k_q}$
\end{itemize}

Here, we discuss the condition $ \min_\omega \frac{p_i}{\sqrt{q_i}} \geq \frac{1}{C_1}$ in detail. Notice that, in Alg.~\ref{alg:mainSVM}, the sampling distribution $p$ can be handled as a hyperparameter of the RFF algorithm, i.e., a parameter that we can choose arbitrarily. 
Based on this observation, we choose $p$ such that $C_1$ is the least possible value, so that we have the ideal scaling according to Cor.~\ref{col:epsilonapp}. 
We formulate the problem as $C_{1,opt}^{-1} = \max_{p} \min_i \frac{p_i}{\sqrt{q_i}}$. By putting $c = \min_i \frac{p_i}{\sqrt{q_i}}$, this problem can be formulated as a convex optimization problem: 
\begin{align*}
\min_{c,p} \quad & -c \\
s.t. \quad &  c \leq \frac{p_i}{\sqrt{q_i}} \quad \forall i \in \{1, \cdots, |\Omega|\} \\ 
\quad & p_i \geq 0 
\quad\forall i \in \{1, \cdots, |\Omega|\}\\ 
\quad & \sum_{i} p_i= 1.
\end{align*} 

We first compute the Lagrangian: 
\be 
L = -c + \sum_i \alpha_i \biggl( c-\frac{p_i}{\sqrt{q_i}} \biggr)   -  \sum_i \beta_i p_i + \lambda (1-\sum_i p_i )  
\ee
where $\alpha, \beta, \lambda$ are Lagrange multipliers. We then apply KKT conditions~\cite{kuhn1951nonlinear,boyd2004convexcorrect},  
\begin{align}
\begin{cases}
    \partial_{c} L &= -1 + \sum_i \alpha_i  = 0  \nonumber \\ 
    \partial_{p_i} L &= -\frac{\alpha_i}{\sqrt{q_i}} -  \beta_i - \lambda = 0\nonumber\\ 
    p_i & > 0 \nonumber \\ 
    \beta_i p_i &= 0 \rightarrow \beta_i = 0 \nonumber \\ 
    \alpha_i &\geq 0 \nonumber \\ 
    0&= \alpha_i ( c-\frac{p_i}{\sqrt{q_i}})
\end{cases}
\rightarrow  p^*_i = \frac{\sqrt{q_i}}{\int_\Omega \sqrt{q_i} }
\end{align}

By substituting the optimum $p^*$ in the objective function we have:
\be \label{eq:temp4app}
C_1^* = \sum_i \sqrt{q_i} 
\ee

For convex optimizations, KKT conditions~\cite{kuhn1951nonlinear},~\cite[Ch.~5, p.~241-249]{boyd2004convexcorrect} are necessary and sufficient for the solution.
Thus, this is the optimal solution.

Note that we also have a symmetry condition on $p$ in Thm.~\ref{thm:mainapp}, but we did not apply it as a constraint to the optimization problem here.
Yet, we remark that $p^*$ satisfies this symmetry condition.

The fact that this distribution $p^*$ characterizes the optimal resources is used for deriving Prop.~\ref{prop:loose} of the main text. Now we provide the detail of the sufficient condition for RFF-dequantization of QSVM and its proof.

\begin{proposition} [Sufficient Conditions for RFF dequantization of QSVM (Prop.~\ref{prop:loose} of the main text)] \label{prop:apploose}
Let $f_q$ be the true risk optimizer of QSVM (according to Def.~\ref{def:setup}) trained with a quantum kernel $k_q$, whose frequency support and diagonal distribution (Def.~\ref{def:FT}) are denoted as $\Omega$ and $q$, respectively. 
Then, Alg.~\ref{alg:mainSVM} with sampling probability $p: \{1,\cdots,|\Omega|\} \rightarrow \Rbb$ can dequantize the quantum model, if: 
\begin{itemize}
    \item \textbf{Concentration:} $\sum_i \sqrt{q_i} \in \OC(\poly(d))$
    \item \textbf{Alignment:} $p_i = \frac{\sqrt{q_i}}{\sum_j \sqrt{q_j}}$ 
    \item \textbf{Bounded RKHS norm: }$\norm{f_q}_{k_q} \in \OC(\poly(d))$.
\end{itemize} 
\end{proposition}
\begin{proof}
We are given a quantum kernel $k_q$, an SVM task and a probability distribution such that the three conditions above hold. Here, we demonstrate that RFF-SVM can have an $\epsilon$-small generalization gap with QSVM under these three conditions.
First, we have $p_i = \frac{\sqrt{q_i}}{\sum_j \sqrt{q_j}}$ (the alignment condition). Using the calculation for the optimal $C_1$ in Eq.~\eqref{eq:temp4app} reveals that
 \be \label{eq:C_1poly}
 \min_i \frac{p_i}{\sqrt{q_i}} \geq \frac{1}{C_1} \ , 
 \ee 
where $C_1 = \sum_i \sqrt{q_i}$. 
Then, since the first condition (concentration condition) assumes that $\sum_i \sqrt{q_i} \in \OC(\poly(d))$, we get $C_1 \in \OC(\poly(d))$. 
Finally the last condition (bounded RKHS norm) means that there exists a $C_2$ such that 
 \be \label{eq:C_2poly}
 \norm{f_q}_{k_q} \leq C_2 \in \OC(\poly(d)) \ .
 \ee

Corollary~\ref{col:epsilonapp} implies that $\bm{\Omega}((C_1C_2)^2/\epsilon^2)$ frequency and data samples are sufficient for RFF-SVM (with $p$ and $C= C_1C_2$ as inputs) to reach an $\epsilon$-small generalization gap.
Hence, together with $C_1,C_2 \in \OC(\poly(d))$ shown by the first two conditions, the requirements for frequency and data samples is polynomial in both $d, \epsilon^{-1}$, meaning the QSVM tasks can be dequantized by definition (Def.~\ref{def:rff_dequantization}).
This completes the proof of Prop.~\ref{prop:loose} of the main text. 
\end{proof} 
Now we state a slightly modified version of Thm.~\ref{thm:mainapp}, which will be useful to prove Prop.~\ref{prop:tight} in the next section. Similar to Thm.~\ref{thm:mainapp}, the following theorem explores the conditions required to reach an $\OC(1/\sqrt{m}+1/\sqrt{D})$ generalization gap between RFF dequantization and QSVM. However, unlike Thm.~\ref{thm:mainapp} where these conditions are on both the true risk minimizer and the quantum kernel, here the condition is \textit{only} on the {true risk minimizer}. This helps us prove our results regarding QNN-SVM (Prop.~\ref{prop:tight}), since QNNs do not explicitly use a kernel. 
\begin{theorem}[Performance Comparison (Alternative Form)] \label{thm:main2app}
    Assume $f_q$ to be the true risk optimizer of QSM with a quantum kernel $k_p$ (the setup of Def.~\ref{def:setup}).  $\tilde{p}:\Omega_+\rightarrow \Rbb$. If there exists $C_0$ and probability distribution $\tilde{p}:\Omega_+ \cup \{0\} \rightarrow \Rbb$, such that 
    \begin{align} \label{eq:ass_thmc5}
        \max\{|\Re(\Tilde{c}_j)|, |\Im(\Tilde{c}_j)|\} \leq C_0 \tilde{p}_{\omv_j}  \ \forall j\in \{0,\cdots,|\Omega_+|\}.
    \end{align}
    Then, for every $\delta \in [0,\frac{1}{2})$ with probability at least $1-2\delta$ over the data and frequency samples, the gap of the true risk is given by
    \be \label{eq:true_risk_c5}
    \RC[\hat{f}] - \RC[f_q] \in \OC \biggl ( \biggl ( \frac{1}{\sqrt{m}} + \frac{1}{\sqrt{D}} \biggr ) C_0 \sqrt{\log(\frac{1}{\delta})}\biggr) .
    \ee 
    where $\hat{f}$ is the decision function of RFF kernel SVM Alg.~\ref{alg:mainSVM}, with frequency samples $D$, search radius $C = C_0$ and distribution $\tilde{p}$ on $\Omega_+ \cup \{0\}$ as inputs.
\end{theorem}

\begin{proof}
We are given a SVM task, a quantum kernel $k_q$, the true risk optimizer of the QSVM (Def.~\ref{def:setup}), real number $C_0$, and a distribution $\tilde{p}$ such that Eq.~\eqref{eq:ass_thmc5} is satisfied where $\tilde{c}_j$'s are Fourier coefficients of $f_q$. Corollary~\ref{col:hardm} bounds the generalization gap between the output of Alg.~\ref{alg:mainSVM} and the true risk minimizer in set $\GC_{\tilde{p},C}$ where $C=C_0$ as follows
    \be \label{eq:temp1app}
     \RC[\hat{f}] - \min_{f \in \GC_{\tilde{p},C}} \RC[f] \in \OC \biggl ( \biggl ( \frac{1}{\sqrt{m}} + \frac{1}{\sqrt{D}} \biggr ) C_1C_2 \sqrt{\log(\frac{1}{\delta})}\biggr),   
\ee
where  $\GC_{\tilde{p},C}$ is defined in Eq.~\eqref{eq:gdef}, i.e.,
$$\GC_{p,C} = \{f(\xv)= \sum_{\omv \in \Omega}  \alpha(\omv) \cos(\omv \cdot \xv)+ \beta(\omv)\sin(\omv \cdot \xv)  :|\beta(\omv)| , |\alpha(\omv)|\leq Cp(\omv)\}.$$
Thus, to show the result, we need to show that the QSVM decision function $f_q$ is in $\GC_{\tilde{p},C}$. 
The Fourier transform of $f_q$ (Eq.~\eqref{eq:fq:dec}) and condition in Eq.~\eqref{eq:ass_thmc5} fit into the definition of $\GC_{p,C}$ and the proof is complete
\end{proof}
Similar to Thm.~\ref{thm:mainapp}, we follow Thm.~\ref{thm:main2app} to obtain the following corollary. 
\begin{corollary} \label{col:epsilonapp2}
    Given the settings of Thm.~\ref{thm:main2app}, if there exists $C_0$ such that
    \begin{align}
         \max\{|\Re(\Tilde{c}_j)|, |\Im(\Tilde{c}_j)|\} \leq C_0 \tilde{p}_{\omv_j}  \ \forall j\in \{0,\cdots,|\Omega_+|\}.
    \end{align}
    then to have $\RC[\hat{f}] - \RC[f_q] \leq \epsilon$ with probability $1-2\delta$ over data and frequency samples, it is sufficient to choose the number of frequency samples $D$ and data points $m$ such that
    \be
    \biggl(\frac{1}{\sqrt{m}}+\frac{1}{\sqrt{D}} \biggr)^{-1} \in \bm{\Omega}   \biggl ( C_0 \sqrt{\log(\frac{1}{\delta})} /\epsilon \biggr)  
    \ee
\end{corollary}
\begin{proof}
    The proof follows from Thm.~\ref{thm:main2app}. 
    We take the same step as Cor.~\ref{col:epsilonapp} and derive the condition.  
\end{proof}
Finally, we note one remark regarding Thms.~\ref{thm:mainapp} and~\ref{thm:main2app}. While these theorems are stated for the true risk minimizer of SVM tasks, the results can be extended to the empirical risk minimizer because the empirical risk minimizer of regularized kernel regression and SVM also exists in the RKHS of the kernel according to the representer theorem. 

\subsection{Proof of Prop.~\ref{prop:tight}}
We next move to the proof of Prop.~\ref{prop:tight} which is Prop.~\ref{prop:tightapp} here.
Thm.~\ref{thm:main2app} and its following corollary (Cor.~\ref{col:epsilonapp2}), that are formulated for QSVM also hold for the true risk minimizer of QNN-SVM using Hamiltonian encoding (the setup of the Prop.~\ref{prop:tight}).  The true risk optimizer of the QNN can be written as $f_q(\xv) = \expval{\bm{\nu}, \phv_{\text{QNN}}(\xv)}$ where $\phv_{\text{QNN}}$ is the QNN feature map defined in Eq.~\eqref{eq:pqcfm} $\bm{\nu}$ is a $|\Omega|$ dimensional vector. This means that $f_q$ has the form of Eq.~\eqref{eq:fq:dec} and as a result the proof of Thm.~\ref{thm:main2app} can be derived exactly the same way for the true risk optimizer of the QNN. 
 
Similar to the QSVM case, we derive the condition of data points and frequency samples to achieve the precision $\epsilon$.

Cor.~\ref{col:epsilonapp2} implies that, if $C_0$ scales polynomially with the parameter of interest, that is input data dimension $d$, then RFF SVM can efficiently perform as well as QNN-SVM, meaning the task is dequantized. 

We note that $p$ is an input of Alg.~\ref{alg:mainSVM} and should be chosen such that $C_0$ is as small as possible.
Here, we solve the following minimization problem:
\begin{align*}
\min_{p, C_0 } \quad & C_0 \\
s.t \quad &   1-\frac{C_0 p_j}{\max\{|\Re(\Tilde{c}_j)|, |\Im(\Tilde{c}_j)|\}} \leq 0 \quad \forall j \in \{1,2,\cdots,|\Omega|\} \\ 
\quad & p_j \geq 0 \quad \forall j \in \{1,2,\cdots,|\Omega|\}\\ 
\quad & \sum_j p_j = 1.
\end{align*}
This problem is convex in both $C_0$ and $p$ and thus can be solved by applying KKT conditions~\cite{kuhn1951nonlinear},~\cite[Ch.~5, p.~241-249]{boyd2004convexcorrect}. The Lagrangian is given by
\be 
L = C_0 + \sum_j \alpha_j \biggl( 1-\frac{C_0 p_j}{\max\{|\Re(\Tilde{c}_j)|, |\Im(\Tilde{c}_j)|\}} \biggr) -  \sum_j  \beta_j p_j + \lambda (1-\sum p_j)  .
\ee
Then, by applying KKT conditions, we get 
\begin{align}
\begin{cases}
    \partial_{C_0} L &= 1 - \sum_j \alpha_j\frac{p_j}{\max\{|\Re(\Tilde{c}_j)|, |\Im(\Tilde{c}_j)|\}}  = 0  \Rightarrow 1 = -\frac{\lambda}{C_0} \nonumber \\ 
    \partial_{p_j} L &= -\alpha_j \frac{C_0}{\max\{|\Re(\Tilde{c}_j)|, |\Im(\Tilde{c}_j)|\}}-  \beta_j - \lambda = 0 \nonumber \Rightarrow \frac{\alpha_j}{\max\{|\Re(\Tilde{c}_j)|, |\Im(\Tilde{c}_j)|\}} = -\frac{\lambda}{C_0}\\ 
    p_j & > 0 \nonumber \\ 
    \beta_j p_j &= 0 \rightarrow \beta_j = 0 \nonumber \\ 
    \alpha_j &\geq 0 \nonumber \\
    0&=\alpha_j( 1-\frac{C_0 p_j}{\max\{|\Re(\Tilde{c}_j)|, |\Im(\Tilde{c}_j)|\}}) \Rightarrow^{\alpha_j \neq 0} 1-\frac{C_0 p_j}{\max\{|\Re(\Tilde{c}_j)|, |\Im(\Tilde{c}_j)|\}} = 0
\end{cases}
\rightarrow  p^*_j = \frac{\max\{|\Re(\Tilde{c}_j)|, |\Im(\Tilde{c}_j)|\}}{C_0}.
\end{align}
We recall that $C_0^* = \sum_j \max\{|\Re(\Tilde{c}_j)|, |\Im(\Tilde{c}_j)|\}$ from the assumption. 

Again, the fact that this is the distribution that achieves the optimal resource is used to derive the following proposition for sufficient conditions for dequantization. 
\begin{proposition} [Sufficient Conditions for RFF dequantization of QNN-SVM (Prop.~\ref{prop:tight} of the main text)] \label{prop:tightapp}
Let $f_q$ be a true risk minimizer of QNN-SVM. Then, Alg.~\ref{alg:mainSVM} with sampling probability $p: \Omega \rightarrow \Rbb$ can dequantize the QML model, if
 \begin{itemize}
        \item  \textbf{Bounded Fourier Sum:}  $C_0 =  \sum_{\omv \in \Omega} |{c}_{\omv}| \in \OC(\poly(d))$ 
        \item \textbf{Alignment:} $p_{\omv} =  |{c}_{\omv}|/C_0 $  where $|{c}_{\omv}|$ are the Fourier coefficients of $f_q$ (Eq.~\eqref{eq:pqcft}).
    
 \end{itemize}
\end{proposition}

\begin{proof} 
Before we proceed to the proof, it is important to note a subtle difference between the notation of the main text and this appendix. In Eq.~\eqref{eq:pqcft} we denote the Fourier coefficients of the QNN outcomes with $c_{\omv}$ and $\omv$ runs over all of the frequency support. In this appendix, the Fourier coefficients in Eq.~\eqref{eq:fq:dec}, are denoted by $\tilde{c}_j$ where $j$ runs from $0$ to $|\Omega_+|$. In other words, if we represent the same function with these notations we will have: $\tilde{c}_0= c_0$ and $\forall \omv_j \in \Omega_+$, $\tilde{c}_j = c_{\omv_j} = c^*_{\omv_j}$. 

Now we proceed to the proof. We are given a QNN-SVM task with true risk minimizer $f_q$ and a distribution $p$ and a real number $C_0$ for which the two conditions above (alignment and bounded Fourier norm) hold. From distribution $p$ define another distribution $\tilde{p}: \{0\} \cup \Omega_+$ such that $\tilde{p}_0 = p_0 = |\tilde{c}_0|/C_0$ and $\forall \omv_j \in \Omega_+, \ \tilde{p}_{\omv_j} = p_{\omv_j} + p_{-\omv_j} = 2\tilde{c}_j/C_0$. Thus we have, $C_0\tilde{p_j} = 2|\tilde{c}_j| \geq \max\{|\Re(\Tilde{c}_j)|, |\Im(\Tilde{c}_j)|\}$. Therefore, Cor.~\ref{col:epsilonapp2}, tells that RFF-SVM with $\tilde{p}$ and $C_0$ as inputs, requires $\bm{\Omega}(\poly(C_0, \epsilon^{-1}))$  data and frequency samples to get $\epsilon$-close to the true risk of QNN-SVM. Since according to the bounded Fourier sum condition $C_0 \in \OC(\poly(d))$, the requirements for frequency and data samples is polynomial in both $d, \epsilon^{-1}$, meaning the QNN-SVM tasks can be dequantized by definition (Def.~\ref{def:rff_dequantization}).
This completes the proof of Prop.~\ref{prop:loose} of the main text.  
\end{proof}

\section{RFF Generalization Bounds and Sufficient Conditions for the Dequantization of QNN Ridge Regression}\label{app:CondReg}
In this appendix, we apply the results of Ref.~\cite{rahimi2008weighted} to linear regression to obtain guarantees about the performance of RFF regression. Next, in Sec.~\ref{subsec:QNNRR}, following an approach analogous to that used for the dequantization of QNN–SVM classification (App.~\ref{app:main results}),  we derive an alternative set of sufficient conditions for dequantization of QNN regression, and compare them to those presented by Ref.~\cite{sweke2023potential}. In particular, our concentration conditions are broader than the concentration condition of Ref.~\cite{sweke2023potential}. We attribute this difference to the fact that the statistical learning theory result used by Ref.~\cite{sweke2023potential}, relies on stronger assumptions and consequently yields tighter sufficient conditions.  

\subsection{Generalization Bounds for RFF Regression} \label{subsec:RFFRR}
We begin by specifying RFF ridge regression in Alg.~\ref{alg:RFFRR}. 
\begin{algorithm} 
\caption{Random Features Ridge Regression}\label{alg:RFFRR}
\begin{algorithmic}
\STATE \textbf{Input:} A dataset $\{\xv_i,\yv_i\}_{i=1}^{m}$ of $m$ points, two bounded feature functions $|\phi_1(\xv, \bm{w})| \leq 1$ and $|\phi_2(\xv, \bm{w})| \leq 1$, an integer $D$, a scalar $C$ and a probability distribution $p(\bm{w})$ on the parameters of $\phi_1$ and $\phi_2$.  
\STATE{\textbf{Output:} A function $\hat{f}(\xv) = \sum_{j=1}^D \alphav_j\phi_1(\xv, {\wv}_j) + \alphav_{j+D} \phi_2(\xv,{\wv}_j)$} 
\STATE{\textbf{1:} Draw $W = \{ \wv_1,\cdots,\wv_D\}$ i.i.d. from $p$.}
\STATE{\textbf{2:} Featurize the input: $\bm{z}_i= (\phi_1(\xv_i, \wv_1), \cdots, \phi_1(\xv_i, \wv_D),\phi_2(\xv_i, \wv_1), \cdots, \phi_2(\xv_i, \wv_D) )^T$.}
\STATE{\textbf{3:} With $W$ fixed, solve the mean square empirical risk minimization problem: 
\begin{align*}
\min_{\alphav \in \Rbb^D} \quad &   \frac{1}{m}\sum_{i=1}^m (\alphav^T\zv_i-y_i)^2 \\
s.t. & \quad \norm{\alphav}_\infty \leq C/D
\end{align*}}
\end{algorithmic}
\end{algorithm}

Before presenting the theoretical guarantees of the algorithm, we first recall some notation and definitions. In a regression task,  we are given a training set $\{\xv_i, y_i\}_{i=1}^m$ drawn i.i.d. from some distribution $P:\XC\times\YC\rightarrow \Rbb$. We assume that the labels are bounded by some positive real number $b$, i.e. $\forall y\in \YC, \ |y|\leq b$. The goal is to minimize the true risk defined as 
\be \label{eq:TRRR}
R[f]:= \Ebb_{(\mathbf{X},Y)\sim P}[(f(\mathbf{X})-Y)^2] \ . 
\ee
 Since the underlying data distribution is unknown, we resort to minimizing the empirical risk. The empirical risk is defined as the sample mean of the true risk for the training set, i.e., 

 \be \label{eq:ERRR}
 \hat{R}[f]:= \frac{1}{m} \sum_{i=1}^m (f(x_i)-y_i)^2 \ .   
 \ee
With these definitions in place, Alg.~\ref{alg:RFFRR} admits the following generalization guarantee. 
\begin{theorem}(Generalization performance of Alg.~\ref{alg:RFFRR}) \label{thm:RFFRR}
    Consider a regression task where we are given a set of training data $\{\xv_i, y_i\}_{i=1}^m$ drawn i.i.d. from some distribution $P:\XC\times\YC\rightarrow \Rbb$. Assume the labels are uniformly bounded, i.e., there exists a constant $b$ such that $\forall y\in \YC \ |y|\leq b$.
    Let $p$ be a distribution on a set $\Omega$ and $\phi_1$ and $\phi_2$ be feature functions that satisfy $\sup_{\bm{w},\xv} |\phi_i(\xv, \bm{w})| \leq 1 ,i=1,2$. Then, we define the set
    \be \label{eq:fpc2}
        \FC_{p,C,\phi_1,\phi_2} = \{f(\xv)= \int_\Omega \alpha_1(\bm{w})\phi_1(\xv,\bm{w}) +\alpha_2(\bm{w}) \phi_2(\xv,\bm{w}) d\bm{w} :|\alpha_i(\bm{w})|\leq Cp(\bm{w}), i=1,2\}.
    \ee
    with a constant $C$. 
     Then for any $\delta > 0$, Alg.~\ref{alg:RFFRR} outputs a function $\hat{f}$ that satisfies
\be
R[\hat{f}]- \min_{f \in \FC_{p,C,\phi_1,\phi_2}} \ R[f] \in \OC \biggl ( \biggl ( \frac{1}{\sqrt{m}} + \frac{1}{\sqrt{D}} \biggr ) bC \sqrt{\log(\frac{1}{\delta})}\biggr) 
\ee
with probability at least $1-2\delta$ over the training set and the choice of $D$ samples,  $\wv_1,\cdots, \wv_D$.
\end{theorem}
Note that, although we state the results for general feature functions $\phi_1$ and $\phi_2$, we often instantiate them as sine and cosine functions in practice, thereby recovering the standard random Fourier features regression framework.

The proof closely resembles Ref.~\cite{rahimi2008weighted}. The output of the Alg.~\ref{alg:RFFRR} is in the set,  
\be \label{eq:Fw1}
\hat{\mathcal{F}}_w \equiv\left\{f(\xv)=\sum_{k=1}^D \alphav_k \phi_1\left(\xv ; \wv_k\right)+ \betav_k \phi_2\left(\xv ; \wv_k\right)\ \text{s.t.} \ |\alphav_k|, |\betav_k| \leq \frac{C}{D}\right\} .
\ee
We denote by $f^*$ the true risk minimizer over $\FC_{p,C,\phi_1,\phi_2}$ and by $\hat{f}$ the empirical risk minimizer over  $\hat{\FC}_w$, which corresponds to the output of the algorithm.
Also, let $\hat{f}^*$ be the true risk optimizer over $\hat{\FC}_w$. 
Then, the following inequality can be derived;
\be \label{eq:bounds}
R[\hat{f}] -R[f^*] = R[\hat{f}] - R[\hat{f}^*] + R[\hat{f}^*] -R[f^*] \leq |R[\hat{f}] -R[\hat{f}^*]|  + R[\hat{f}^*] -R[f^*].
\ee

To prove Thm.~\ref{thm:RFFRR}, we bound these two terms separately. More precisely, we will derive the following:
\begin{itemize}
    \item For the first term, we derive an estimation error bound. Specifically, we rely on standard results from statistical learning theory to show that, for every function in $\hat{\FC}_w$, the true risk and the empirical risk are close to each other. 
    \item  For the second term, we apply an approximation error bound. This result ensures that the minimum true risk achievable within $\hat{\FC}w$ is not substantially larger than the minimum true risk attainable in $\FC{p,C}$ (see Lems.~1 and 2 of Ref.~\cite{rahimi2008weighted}).
\end{itemize}

 In what follows, we provide a detailed derivation for each term.
\paragraph*{Estimation Error Bound.}
To bound the first term we use the following Lemma. 
\begin{lemma}(Estimation Error Bound) \label{lem:estimation}
    Let $\w_1,\cdots, w_D$ be fixed. Then for any $\delta>0$ with probability at least $1-\delta$ over the choice of the training set, we have 
\be 
\forall f\in \hat{\FC_w} \ R[f] - \hat{R}[f] \leq  \epsilon_{\text{est}} = 16b\frac{C}{\sqrt{m}} + 4b^2\sqrt{\frac{\log(\frac{1}{\delta})}{2m}} \, .
\ee
\end{lemma}
Lem.~\ref{lem:estimation} is proven in Sec.~\ref{subsec:estimation}.
According to Lem.~\ref{lem:estimation}, we have that $R[\hat{f}] - \hat{R}[\hat{f}] \leq \epsilon_{\text{est}}$ for $\hat{f}$ and  $R[\hat{f}^*] - \hat{R}[\hat{f}^*] \leq \epsilon_{\text{est}}$ for $\hat{f}^*$. Therefore, we bound the first term of Eq.~\ref{eq:bounds}
\begin{align} \label{eq:estbound}
  |R[\hat{f}] - R[\hat{f}^*]| &= |R[\hat{f}] - \hat{R}[\hat{f}^*] + \hat{R}[\hat{f}^*] - R[\hat{f}^*]| \leq |R[\hat{f}] - \hat{R}[\hat{f}^*]| + |R[\hat{f}^*]-\hat{R}[\hat{f}^*]| \nonumber \\
  &\leq |R[\hat{f}] - \hat{R}[\hat{f}]| + |R[\hat{f}^*]-R[\hat{f}^*]| \leq 2\epsilon_{\text{est}} 
\end{align}
where in the first line we use the triangle inequality and in the second line we use the fact that $\hat{f}$ is the empirical risk minimizer over $\hat{\FC}_w$ and so $\hat{R}[\hat{f}] \leq \hat{R}[\hat{f}^*]$. 
\paragraph*{Approximation Error Bound.} We use Lem.~\ref{lem:approx} in App.~\ref{app:QKSVM} to bound the second term of Eq.~\ref{eq:bounds}.  
The square loss function in general is not Lipschitz continuous. However, if we bound its support to the set $\YC$ with radius $b$, it becomes Lipschitz with constant $L$, where 
\be
L = \sup_{y,y'\in \YC} \left|\frac{d}{dy} (y-y')^2 \right| = 4b \, .
\ee
Thus, by applying Lem.~\ref{lem:approx} to the square loss with Lipschitz constant $4b$, we will have 
\be \label{eq:approxbound}
\mathbf{R}[\hat{f}] \leq \mathbf{R}\left[f^*\right]+\frac{8b C}{\sqrt{D}}\left(1+\sqrt{2 \log \frac{1}{\delta}}\right) = \epsilon_{\text{app}}  \, .
\ee
Finally, following from Eq.~\eqref{eq:bounds}, we apply the estimation bound of Eq.~\eqref{eq:estbound} and the approximation bound of Eq.~\eqref{eq:approxbound}. By the union bound, with probability $1-2\delta$, we have 
\be 
 |R[\hat{f}] - R[\hat{f}^*]| \leq \epsilon_{\text{app}} + 2\epsilon_{\text{est}} = \frac{8b C}{\sqrt{D}}\left(1+\sqrt{2 \log \frac{1}{\delta}}\right) + 32b\frac{C}{\sqrt{m}} + 8b^2\sqrt{\frac{\log(\frac{1}{\delta})}{2m}}  \, .
\ee

\subsubsection{Proof of Lemma~\ref{lem:estimation}} \label{subsec:estimation}
To bound the difference between the true risk and the empirical risk of functions in the function class $\hat{\FC}_w$, we use the following theorem from section 11.2.2 of Ref.~\cite{mohri2018foundations}. 
\begin{theorem} (Theorem 11.3 of Ref~\cite{mohri2018foundations})\label{thm:Mohri}
Let $(\xv,y) \sim P$ define a regression problem where $P:\XC\times \YC \rightarrow \Rbb$ is an unknown probability distribution. Assume we are given a hypothesis class $\FC$ and a training sample of size $m$, $S = \{(\xv_i,y_i)\}_{i=1}^m$, for this problem. Then, for any $\delta > 0$ with probability at least $1-\delta$ over the training samples, the following holds for all $f\in \FC$: 
\be
R[f]-\hat{R}[f] \leq 8b\mathfrak{R}_m(\FC) + 4b^2\sqrt{\frac{\log(\frac{1}{\delta})}{2m}} \ , 
\ee
    where $b$ is the radius of $\YC$, and $\mathfrak{R}_m(\FC)$ is called the Rademacher Complexity of $\FC$ and is defined as 
    \be
        \mathfrak{R}_m(\FC) = \Ebb_{S\sim P^m} \left[\Ebb_{\sigma_1,\cdots, \sigma_m} \left[ \sup_{f\in \FC} \frac{1}{m} \sum_{i=1}^m \sigma_i f(\xv_i) \right] \right] \ ,
    \ee
    and $\sigma_i$'s are uniform random variables on the support $\{-1,1\}$. 
\end{theorem}
In App.~B of Ref.~\cite{rahimi2008weighted}, the Rademacher complexity of the function class $\hat{\FC}_w$ for one feature function $\phi$ is bounded by $C/\sqrt{m}$. Since, Alg.~\ref{alg:RFFRR} uses two feature functions $\phi_1, \phi_2$, the Rademacher complexity is upper bounded by $\mathfrak{R}_m(\FC) \leq 2C/\sqrt{m}$. This bound together with Thm.~\ref{thm:Mohri} proves Lem.~\ref{lem:estimation}. 

\subsection{Generalization Gap of QNN Ridge Regression and RFF Regression} \label{subsec:QNNRR}
The following theorem bounds the difference between the true risk of QNN regression and RFF regression. 
\begin{theorem}[Performance Comparison RFF Regression vs. QNN Ridge Regression] \label{thm:GenGapRR}
    Assume $f_q$ to be the true risk optimizer of a QNN Ridge Regression with $m$ training data. Take $\tilde{c}_j$s to be the Fourier coefficients of $f_q$ in its Fourier expansion namely,
    \begin{align}\label{eq:fq:dec4}
     f_q &= \Tilde{c}_0 + \sum_{j=1}^{|\Omega_+|} \Tilde{c}_j e^{i\omv_j \cdot \xv} + \Tilde{c}^*_j e^{-i\omv_j \cdot \xv} \nonumber \\
      &= \Tilde{c}_0 + \sum_{j=1}^{|\Omega_+|} 2 \Re(\Tilde{c}_j) \cos(\omv_j \cdot \xv) -2\Im (\Tilde{c}_j) \sin(\omv_j \cdot \xv) \, .
\end{align}
    Further, assume we are given a distribution $\tilde{p}:\Omega_+ \cup \{0\}  \rightarrow \Rbb$.   If there exists a real number $C_0$  such that 
    \begin{align} \label{eq:QSVMcond3}
        \forall \omv_j \in \Omega_+ \  2|\Tilde{c_j}| &\leq p(\omv_j) C_0 \, , \\
        |\Tilde{c}_0| \leq p(0) C_0 \, ,  \label{eq:QSVMcond4}
    \end{align}
    then the following is true: 
    \be \label{eq:temp3app1}
    \RC[\hat{f}] - \RC[f_q] \in \OC \biggl ( \biggl ( \frac{1}{\sqrt{m}} + \frac{1}{\sqrt{D}} \biggr ) b C_0 \sqrt{\log(\frac{1}{\delta})}\biggr) \ ,
    \ee 
    where $b$ is an upper bound on the set of labels and $\hat{f}$ is the output of Alg.~\ref{alg:RFFRR} with the number of frequency samples $D$, search radius $C= C_0$ and distribution $\tilde{p}$ as inputs.
\end{theorem}
\begin{proof}
We are given a regression task, the true risk optimizer of the QNN, $f_q$, the real number $C_0$, and a distribution $\tilde{p}$ such that Eqs.~\eqref{eq:QSVMcond3},\eqref{eq:QSVMcond4} are satisfied. Thm.~\ref{thm:RFFRR} bounds the generalization gap between the output of Alg.~\ref{alg:RFFRR} and the true risk minimizer in set $\FC_{\tilde{p},C_0}$ as follows;
    \be 
     \RC[\hat{f}] - \min_{f \in \FC_{\tilde{p},C_0}} \RC[f] \in \OC \biggl ( \biggl ( \frac{1}{\sqrt{m}} + \frac{1}{\sqrt{D}} \biggr ) C_0 \sqrt{\log(\frac{1}{\delta})}\biggr),   
\ee
where  $\FC_{\tilde{p},C_0}$ is defined in Eq.~\eqref{eq:fpc2} and we have replaced $\phi_1$, $\phi_2$ with sine and cosine functions, i.e.,
$$\FC_{p,C} = \{f(\xv)= \sum_{\omv \in \Omega}  \alpha(\omv) \cos(\omv \cdot\xv)+ \beta(\omv)\sin(\omv \cdot \xv)  :|\beta(\omv)| , |\alpha(\omv)|\leq Cp(\omv)\}.$$
If we have $f_q \in \FC_{\tilde{p},C_0}$, then $\RC[\hat{f}] - \RC[f_q] \leq \RC[\hat{f}] - \min_{f \in \GC_{\tilde{p},C_0}} \RC[f]$ by definition. Also, the upper bound of Thm.~\ref{thm:RFFRR} will also be an upper bound for $\RC[\hat{f}] - \RC[f_q]$. Thus, to prove the result, we show that the QNN decision function $f_q $ is in $\FC_{\tilde{p},C_0}$.

This follows almost immediately from the assumptions because we have $2\Im (\Tilde{c}_j)\leq 2|\Tilde{c}_j| \leq C_0 p(\omv_j)$ for the coefficients of sines in the expansion of Eq.~\eqref{eq:fq:dec4}; similarly, we have $2\Re (\Tilde{c}_j)\leq 2|\Tilde{c}_j| \leq C_0 p(\omv_j)$ for the coefficients of cosines. Moreover, for frequency zero we have $\Tilde{c}_0 \leq C_0p(0)$. Therefore, $f_q \in \FC_{\tilde{p},C_0}$ by definition and the proof is completed.
\end{proof} 
\begin{corollary} [Sufficient Conditions for RFF Dequantization]\label{prop:tightapp2}
Let $f_q$ be a true risk minimizer of QNN Regression. Then, Alg.~\ref{alg:RFFRR} with sampling probability $p: \Omega \rightarrow \Rbb$ can dequantize the QML model, if
 \begin{itemize}
        \item  \textbf{Bounded Fourier Sum:}  $C_0 =  \sum_{\omv \in \Omega} |{c}_{\omv}| \in \OC(\poly(d))$ 
        \item \textbf{Alignment:} $p_{\omv} =  |{c}_{\omv}|/C_0 $  where $|{c}_{\omv}|$ are the Fourier coefficients of $f_q$ (Eq.~\eqref{eq:pqcft}).
 \end{itemize}
\end{corollary}
The poof of this corollary is identical to the proof of Prop.~\ref{prop:tightapp}. 

Cor.~\ref{prop:tightapp2} provides a set of sufficient conditions for dequantization of QNN Regression. The alignment condition states that the RFF sampling distribution should be proportional to the Fourier coefficients of the optimal decision function. 

Moreover, the bounded Fourier sum condition can be seen as a form of concentration condition. To see this, note that $f_q$ is the output of a QNN, and thus its absolute value is bounded by one. Using Parseval's formula, we have 
$$
\sum_{\omv \in \Omega} |c_{\omv}|^2 = \frac{1}{2\pi}\int_{-\pi}^{\pi} |f_q(x)|^2 dx \leq \frac{1}{2\pi}\int_{-\pi}^{\pi} 1 dx =1 \ .
$$
This implies that the sum of the squares of the Fourier coefficients is bounded by one. Let us assume that the Fourier coefficients of $f_q$ saturate this bound and their squares sum to one. In this case, the sum of the absolute values of the coefficients, i.e. $\sum_{\omv \in \Omega} |c_{\omv}|$, is minimum value of one when all the coefficients except one are zero; that is, the coefficients are concentrated. Furthermore, it reaches its maximum value of $\sqrt{|\Omega|}$, when all the coefficients are uniformly spread and anti-concentrated.   

Although Ref.~\cite{sweke2023potential} presents the same alignment condition, their concentration condition requires $p^{-1}_{\text{max}} \in \OC(\poly(d))$, which is different from our proposed concentration condition. According to the alignment condition, we have $p_{\text{max}} = \frac{|c_{\text{max}}|}{\sum_{\omv \in \Omega} |c_{\omv}|}$. Since $|c_{\text{max}}|<1$, we have ${\sum_{\omv \in \Omega} |c_{\omv}|} \leq p^{-1}_{\text{max}}$. Therefore, whenever the concentration condition of Ref.~\cite{sweke2023potential} is satisfied, our concentration condition is satisfied as well, while the opposite is not necessarily true. This means that our set of sufficient conditions imply dequantization for a larger class of QML models that were not previously captured by the conditions proposed by Ref.~\cite{sweke2023potential}.   

\section{Numerical Experiments with Data from Ref.~\texorpdfstring{\cite{wozniak2023quantum}}{Lg}}\label{app:dataset}
In this section, we provide additional details on the numerical experiment comparing the performance of the QSVM to RFF-SVM for data from high-energy physics.
Our results show that RFF-SVM can outperform QSVM. We run an experiment on the dataset from Ref.~\cite{wozniak2023quantum}. 
The dataset consists of simulated proton-proton collisions at a center-of-mass energy of $\sqrt{13}$ TeV, representing dijet events. It includes a mix of Standard Model (SM) and Beyond Standard Model (BSM) processes, with each jet containing 100 particles. Jet dynamics are simulated using Monte Carlo methods and processed to emulate detector effects. The studied BSM processes include narrow and broad Randall-Sundrum gravitons decaying into W-bosons and a scalar boson $ A \to HZ$, representing potential new physics at the LHC. 

The raw data has a dimension of 600, due to computational restriction, an autoencoder is used to reduce the high-dimensional simulation data to a latent space of \(\ell = 4, 8, 16,\) or \(32\). The resulting dataset is structured as \(\{ (\bm{x}^{(i)}, y^{(i)}) \in \mathbb{R}^{2\ell} \times \{\textsc{sm}, \textsc{bsm}\}\}_{i=1}^p\), where the \(2\ell\) features arise from dijet events, with \(\ell\) features per jet.

Once it has been classically preprocessed, this data is encoded into a quantum state using the feature map in Fig.~\ref{fig:kernel}. The quantum kernel $k_q(x,x')$ is then obtained by repeating the feature map of Fig.~\ref{fig:kernel}, $L$ times for $\bm{x}$ and then applying the conjugate transpose of the same $L$  layer encoding for $x'$ and measuring the probability of measuring the all-zero state $\ketbra{0}{0}$. This is effectively a fidelity Hamiltonian encoding kernel. The number of qubits in the quantum circuit $n$ is equal to half of the number of input dimensions $d$. Each layer consists of two Pauli rotation gates for each dimension of input data and each Pauli rotation adds one to the maximum frequency expressed by the model~\cite{schuld2021supervised}. So for every dimension of input the frequency range is all the integer numbers from $-2L$ to $2L$ inclusive. This makes the total number of frequencies $|\Omega|= (4L+1)^d$.  

While Ref.~\cite{wozniak2023quantum} uses the dataset to perform anomaly detection via one-class SVM, we use both classes to perform a balanced two-class supervised classification with SVM. This is because we want to compare algorithms in terms of true risk and true risk is not well-defined for anomaly detection tasks. We used 1000 training and 200 test data points for this experiment. For each pair the quantum kernel is simulated and the Gram matrix is obtained. The simulation of quantum circuits are done using \texttt{PennyLane}~\cite{bergholm2018pennylane}. Shot noise is added by sampling from a binomial distribution and classical training and inference of models are done using \texttt{scikit-learn}~\cite{pedregosa2011scikit}. RFF optimization is performed by adding sampling methods to the already existing RFF package~\cite{rffpackage}. 

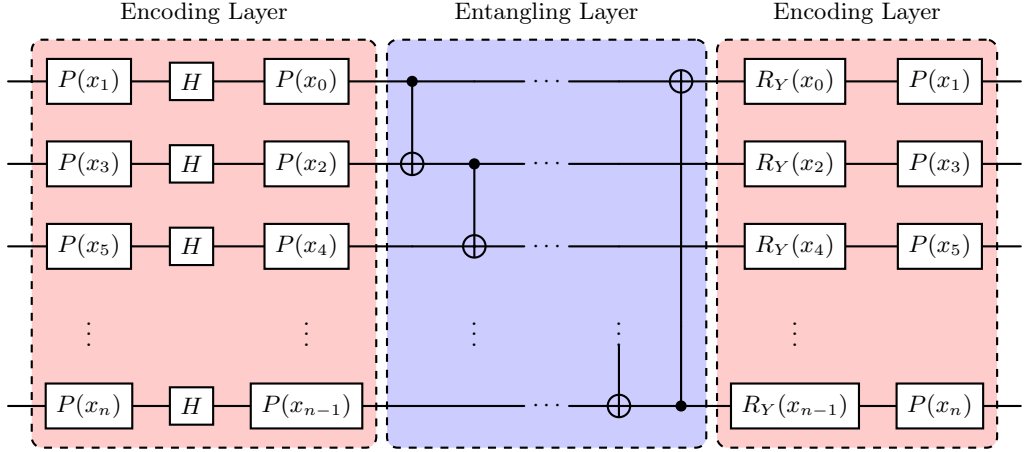
\begin{figure*}
    \centering
    \begin{quantikz}
    &\gate{P(x_1)}\gategroup[5,steps=3,style={dashed,rounded
corners,fill=red!20, inner
xsep=2pt},background,label style={anchor=north,yshift=0.4cm}]{Encoding Layer}&\gate{H}& \gate{P(x_0)}  &   \ctrl{1}\gategroup[5,steps=5,style={dashed,rounded
corners,fill=blue!20, inner
xsep=2pt},background,label style={anchor=north,yshift=0.4cm}]{Entangling Layer} &  &\ \cdots\ & &\targ{4} \wire[d]{q}& \gate{R_Y(x_0)}\gategroup[5,steps=2,style={dashed,rounded
corners,fill=red!20, inner
xsep=2pt},background,label style={anchor=north,yshift=0.4cm}]{Encoding Layer} & \gate{P(x_1)}&\\
    &\gate{P(x_3)}&\gate{H}& \gate{P(x_2)} &  \targ{} & \ctrl{1}&\ \cdots\ & & \wire[d]{q}& \gate{R_Y(x_2)} & \gate{P(x_3)}&\\ 
    &\gate{P(x_5)}&\gate{H}& \gate{P(x_4)} &   & \targ{}& \ \cdots \ & &\wire[d]{q}& \gate{R_Y(x_4)} & \gate{P(x_5)} &\\ 
     &\setwiretype{n} \ \vdots\ &&\ \vdots\ &&\ \vdots\ && \ \vdots\ \wire[d]{q}& \wire[d]{q}&  \ \vdots\ &\\ 
    &\gate{P(x_n)}&\gate{H}& \gate{P(x_{n-1})} &  & &\ \cdots\ &  \targ{}& \ctrl{0}& \gate{R_Y(x_{n-1})} & \gate{P(x_n)}  &
\end{quantikz}
    \caption{\textbf{One layer of the feature map of the kernel used.} This set of gates are repeated $L$ times for input $\bm{x}$ and then it is followed by the conjugate of the overall feature map for $\bm{x'}$ and at the end the computational basis zero is measured. $P$ gates indicate a phase shift and $R_Y$ shows Pauli $Y$ rotations. Every rotation and shift gate of  This makes the output be $\Tr(\U^\dagger(\bm{x'}) \U(\bm{x}) \ket{0}\bra{0} \U^\dagger(\bm{x})\U(\bm{x})\ket{0}\bra{0}) = \Tr(\ket{\psi(\bm{x})}\bra{\psi(\bm{x})} \ket{\psi(\bm{x'})}\bra{\psi(\bm{x'})} )=  |\bra{\psi(\bm{x})}\ket{\psi(\bm{x'})}|^2 = k(\bm{x},\bm{x'})$.   }
    \label{fig:kernel}
\end{figure*}

\subsection{Analytical Justification for Application of truncated convolutional sampling}  \label{app:truncated}
In our numerical results above (Fig.~\ref{fig:QSVM vs RFF}), we observed that RFF with truncated convolutional sampling can outperform QK-SVM with efficient resources despite being independent of the task and the QML model.  Intuitively, truncated sampling achieves this performance because it samples lower frequencies with higher probability. In this section, we show that most natural datasets do not contain very high frequencies. Consequently, when we randomly sample frequencies to perform the RFF algorithm, it is justified to prioritize lower frequencies over the higher ones (convolutional sampling) or even to discard higher frequencies altogether (truncated sampling).

To show that the ground truth of natural datasets does not contain very high frequencies, we use the following result from Fourier analysis. 
\begin{proposition}\label{prop:Fourier}
    Let $f$ be an absolutely integrable function that is $k$ times differentiable. Given that its $k$'th derivative $f^{(k)}$ is also absolutely integrable, then the following is true for the Fourier transform of $f$, $F$: 
    \be
    F(\omega) \in \OC(|\omega|^{-k})
    \ee
\end{proposition}
\begin{proof}
    The fact that the $k$'th derivative of $f$ is absolutely integrable entails that $f^{(k)}$ has a Fourier transform, which we denote with $F^{(k)}$. According to the differentiation property of the Fourier transform we have 
\be 
   F(\omega)  = \frac{1}{(2\pi \omega)^k} F^{(k)} (\omega) \ .
\ee
\end{proof}
Prop.~\ref{prop:Fourier} implies that the Fourier transform of a smooth, $k$-differentiable function has thin tails decaying with rate $|\omega|^{-k}$. Since most real-life datasets originate from such functions, a truncated or convolutional sampling scheme aligns well with their spectrum and hence these RFF sampling strategies are likely to perform well.   

\subsection{The Effect of the Shot Noise on Generalization of QML} \label{app:shotnoise}

In this appendix, we aim to provide some analytical evidence for our heuristic observation in Fig.~\ref{fig:QSVM vs RFF} that shot noise can serve as a form of regularization in QML. Regularization in machine learning is defined as "any modification applied to a learning algorithm aimed to reduce the generalization error without necessarily reducing the training error"~\cite{goodfellow2016deep}.  

The effect of noise as a regularization mechanism has been extensively studied in the context of classical neural networks. For instance, injecting noise to the inputs of a network can enhance generalization~\cite{sietsma1991creating}. In fact, this approach is shown to be equivalent to training the model without noise but with a specific regularization term called Tikhonov regularization~\cite{bishop1995training}. Similarly, adding noise to the weights of the network~\cite{murray1994enhanced} and to the outputs~\cite{chen2020investigation} during the training phase can improve generalization.

 Recall that given a kernel function $k(x,y)$ and a data set of $m$ points, we can construct the $m\times m$ Gram matrix $K_{ij}:= k(x_i, x_j)$. In QK methods with $t$ measurements, shot noise can be seen as a noise acting on the Gram matrix $K\in \Rbb^{m\times m}$. Under shot noise, each element of the Gram matrix $K_{ij}$ is replaced with $\hat{K}_{ij} = 1/t\sum_{l=1}^t V_{ij}^{(l)}$ where $V_{ij}^{(l)} \sim \text{Ber}(K_{ij})$'s are i.i.d samples drawn from a Bernoulli distribution with its mean equal to the true value of the kernel. By the law of large numbers, $\hat{K}_{ij}$ approaches a Gaussian random variable as $t$ grows. Therefore, in this appendix, we model the measurement noise as an additive Gaussian noise, namely, we assume $\hat{K}:= K+G$, where $K$ is the noiseless Gram matrix and $G$ is a symmetric matrix whose elements are i.i.d samples of $\mathcal{N}(0,1/t)$. In practice, shot noise also affects the inference step. However, following our numerical experiments and the previous theoretical analysis on the topic, we assume noiseless measurements during the inference process. 
 
The theoretical analysis of generalization in the presence of noise is considerably more tractable for ridge regression than SVM classification because ridge regression unlike SVM classification admits a closed form solution. For kernel ridge regression, this closed form is used to derive generalization bounds both in the ideal case~\cite{huang2021power} and with depolarizing and measurement noise~\cite{wang2021towards, beigi2022quantum, heyraud2022noisy, wang2025power}. We draw on the main result of Ref.~\cite{huang2021power} to justify how shot noise could work as a form of regularization. 
\begin{theorem}(Generalization Bound of Quantum Kernel Ridge Regression (Theorem 1 of~\cite{huang2021power} as presented in~\cite{beigi2022quantum}))\label{thm:robert}
 Let $f(x)$ be a function that maps the input data to $[-1,1]$. Suppose that we are given a training dataset $\left\{\left(x_i, y_i\right)\right.$: $i=1, \ldots, m\}$ with $y_i = f(x_i)$, where $x_i$ 's are drawn independently at random. Let $x \mapsto \phi_x$, for any data point $x$, be a feature map, where $\phi_x$ is some vector in a Hilbert space. Let $K$ be the associated kernel matrix with entries $W_{i j}=\left\langle\phi_{x_i}, \phi_{x_j}\right\rangle$. Also, let $\lambda \geq 0$ be a regularization parameter. Then, we can find a function $h(x)$ of form $\expval{w, \phi(x)}$ for some $w$ in the Hilber space of $\phi$, such that with probability at least $1-\delta$ over the choice of the training dataset, we have
$$
\mathbb{E}_x\left[\left|h(x)-f(x)\right|\right] \leq \mathcal{O}\left(\sqrt{\frac{\lambda' \mathbf{y}^T(W+\lambda I)^{-1} \mathbf{y}}{m}}+\sqrt{\frac{\log 1 / \delta}{m}}\right),
$$
where $\mathbf{y}=\left(y_1, \ldots, y_N\right)^T$ and $\lambda'= \max\{1,\lambda\}$.
\end{theorem}
This theorem bounds the expected error between the decision function $h(x)$ learned by kernel ridge regression using $W$ and the ground truth $f(x)$, and the bound scales with the square root of the inverse of the regularized Gram matrix. 
Originally, in Refs.~\cite{huang2021power,beigi2022quantum}, it is assumed that $f(x) = \Tr[\rho(x)O]$ for some encoding $\rho$ and observable $O$. However, the results hold for a general function $f(x)$ with range $[-1,1]$.

Note that there are no restrictions on the feature map $\phi_{x}$, and hence $W$ may be obtained from a noiseless QK or a QK with shot noise. This means that the theorem can be used to bound the generalization error in both scenarios: $W=K$ for a noiseless Gram matrix and $W=\hat{K}$ in the presence of shot noise. This applicability is ensured by the fact that the noisy Gram matrix $\hat{K}$ is positive semi-definite with high probability~\cite{beigi2022quantum}.

The following proposition provides insight into the effect on noise on generalization. 
\begin{proposition}[Shot noise can tighten the upper bound for the generalization gap] \label{app_prop:shot_noise}
    Given a dataset of size $m$, we obtain a noisy Gram matrix $\hat{K}$ from a kernel function $k_q$ with $t$ measurements. Take $K$ to be the noiseless Gram matrix obtained from the same QK and the same dataset. We denote the highest eigenvalue of the QK, $k_q$, with $\lambda_{\max}$. Then, if $\lambda_{\max} > \sqrt{\frac{t}{2m^3}}$, with high probability, the upper bound of Thm~\ref{thm:robert} is tighter for the noisy Gram matrix compared to the noiseless one.  
\end{proposition}
\begin{proof}
    Shot noise is approximated as an additive Gaussian noise $G$ with mean $0$ and variance $1/t$. Namely, we have $\hat{K} = K + G$ where $K$ is the noiseless Gram matrix obtained from the same QK and the same dataset. 

    We compare the upper bound of Thm.~\ref{thm:robert} for the noisy and noiseless cases. Concretely, we compare $(K+\lambda I)^{-1}$ with $(\hat{K}+\lambda I)^{-1}$ where $\lambda$ is the regularization parameter of the ridge regression. Therefore, we get 
    \begin{align}\label{eq:arb}
        (\hat{K}+\lambda I)^{-1} &=  ({K}+ G + \lambda I)^{-1} \nonumber \\ 
        &= (K+\lambda I)^{-1} -  (K+\lambda I)^{-1} (G^{-1}+ (K+\lambda I)^{-1})^{-1} (K+\lambda I)^{-1}
    \end{align}
    where in the first line we apply our assumption on shot noise ($\hat{K} = K+G$), and in the second line we use the Woodbury matrix inversion formula.
    Note that, if the second term in the right hand side of Eq.~\eqref{eq:arb} is positive semi-definite, we have $(\hat{K}+\lambda I)^{-1} \preceq (K+\lambda I)^{-1}$. This indicates that shot noise tightens the upper bound on generalization error provided in Thm.~\ref{thm:robert}.

    The remaining question is as to how often the random matrix $(G^{-1}+ (K+\lambda I)^{-1})$ is positive semi-definite.  The term $(K+\lambda I)^{-1}$ is inherently positive, however, the noise matrix $G$ can have negative eigenvalues. If the sum of the smallest eigenvalues of $(K+\lambda I)^{-1}$ and $G^{-1}$ is non-negative then $(G^{-1}+ (K+\lambda I)^{-1})$ is positive semi-definite. Therefore, we examine the spectrum of each term separately. 

    The elements of the $m\times m$ noise matrix $G$ are sampled from $\mathcal{N}(0,1/t)$. In the limit of large $m$, the eigenvalues of $G$ are confined to $[-\sqrt{2m/t}, \sqrt{2m/t}]$ and are independently distributed according to a semi-circular distribution on this interval (section~2.2 of Ref.~\cite{liu2000statistical}). Therefore, with high probability the eigenvalues of $G^{-1}$ lie in the interval $[-\sqrt{t/2m}, \sqrt{t/2m}]$. 
    Furthermore, while $K+\lambda I$ depends on the data in general, its eigenvalues are related to the eigenvalues of the quantum kernel. More concretely, the highest eigenvalue of the Gram matrix $K$ is $m$ times the highest eigenvalue of the quantum kernel function $k_q(x,y)$ (Lem.~12 of Ref.~\cite{ghojogh2021reproducing}). We denote the highest eigenvalue of the QK $k_q$ with $\lambda_{\text{max}}$. Thus, the highest eigenvalue of $K+\lambda I$ is of order $\OC(m\lambda_{\text{max}})$ and the lowest eigenvalue of $(K+\lambda I)^{-1}$ is of order $\OC(\frac{1}{m\lambda_{\text{max}}})$. 
    Thus, if $m\lambda_{\text{max}} > \sqrt{t/2m}$ then $(G^{-1}+ (K+\lambda I)^{-1})$ is positive semi-definite and shot noise tightens the upper bound on generalization. 
\end{proof}
Prop.~\ref{app_prop:shot_noise} states the conditions under which shot noise can reduce the generalization upper bound. Although this bound is not necessarily tight for every problem, one can see how noise could act as a form of regularization for QK ridge regression to potentially reduce the generalization error.  

Overall, in this appendix, we discussed the effect of shot noise on the generalization of QK methods. However, it is important to clarify that the setup and assumptions used in this appendix are different from the case of generalization improvement we observed in our numerical experiments. For example, we here analyzed ridge regression due to its analytical simplicity whereas numerically we perform an SVM classification. Furthermore, neither this analysis nor our numerical experiment account for noise during the inference phase which could negatively impact the generalization. Nevertheless, the discussion above provides valuable insight into the cases where shot noise can be seen as a form of regularization, enhancing the generalization performance.

\section{Relation to Barren Plateaus and Exponential Concentration}
\label{app:bp}

The barren plateau (BP) phenomenon refers to the exponential vanishing of gradients in quantum neural networks (QNNs), which hinders efficient training. Ref.~\cite{cerezo2023does} argued that, for QNNs with classical inputs and measurements, a proof of the absence of barren plateaus implies that the corresponding loss landscape can be simulated classically. This conjecture is supported by a myriad of works~\cite{lerch2024efficient, bermejo2024quantum, martinez2025efficient, angrisani2025simulating, fontana2023classical}. In this section, we will discuss the relationships between barren plateaus and RFF-dequantization.  
\medskip

\paragraph*{\textbf{Quantum Neural Networks.}}

We start by considering our results in the context of QNNs.
The notion of dequantization studied in this work concerns the distribution of Fourier frequencies for a fixed parameter configuration $\theta^*$ (the optimal QNN parameter), rather than the variation of the loss with respect to trainable parameters. In other words, in our analysis of dequantization we are not concerned about the trainability of the QNN. Instead, we compare the best true risk achievable by a QNN with that of the RFF method. Consequently, in the context of QNNs, the questions of barren plateaus and dequantization are largely orthogonal. BPs and RFF dequantizability both are caused by the architecture of the QNN, however, our analysis does not establish any direct relation between the existence (or absence) of barren plateaus and the (non-)dequantizability of a QNN via random Fourier features (RFF). 
\medskip

\paragraph*{\textbf{Quantum Kernels}}

 In the context of QKs, expressivity-induced exponential concentration arises over the distribution of training data rather than over a parameter landscape~\cite{thanasilp2022exponential, saem2025pitfalls, suzuki2024quantum, suzuki2024effect_}. Consequently, since exponential concentration and dequantization both concern the space of inputs, we believe they are related. However, this connection is subtle and non-trivial and certainly it is not the case that one entails the other one.  
To see this, in what follows, we provide two examples of kernels that do not satisfy our sufficient conditions for dequantization. The first example suffers form exponential concentration while the other does not.  

\paragraph*{Example 1: Exponentially Concentrated Kernel.}
Consider a shift-invariant kernel with a uniform Fourier spectrum, 
\begin{equation*}
    k(x,x') = \frac{1}{|\Omega|} \sum_{\omega \in \Omega} \cos(\omega(x-x')).
\end{equation*}
For uniformly distributed inputs $Y \sim \mathcal{U}[-\pi,\pi]$, its variance satisfies
\begin{align*}
    \Var_Y(k(Y)) 
    &\leq \Ebb_Y[k^2(Y)] 
    = \frac{1}{2|\Omega|^2} \sum_{\omega,\omega' \in \Omega} \Ebb[\cos((\omega-\omega')Y)] + \Ebb[\cos((\omega+\omega')Y)] = \frac{1}{|\Omega|},
\end{align*}
where $\Ebb[\cos((\omega-\omega')Y)] = 0$ for $\omega \neq \omega'$ under a symmetric distribution. Since $|\Omega|$ typically scales exponentially with the number of qubits, this kernel exhibits exponential concentration over its inputs. However, it fails to satisfy our dequantization condition due to an anti-concentrated Fourier spectrum.

\paragraph*{Example 2: Non-Exponentially-Concentrated Kernel.}
Consider a modified spectrum that assigns weight $1-\delta$ to one frequency $\omega_0$ and distributes the remaining weight $\delta$ uniformly:
\begin{equation*}
    k(x,x') = (1-\delta)\cos(\omega_0(x-x')) + \frac{\delta}{|\Omega|-1}\sum_{\omega \in \Omega}\cos(\omega(x-x')).
\end{equation*}
In this case, $\Var_Y(k(Y)) \approx (1-\delta)^2$, implying that the kernel is not exponentially concentrated. Furthermore, $\sum_{\omega} \sqrt{q_\omega} \in \mathcal{O}(\sqrt{|\Omega|})$, which is exponential in the input dimension, showing that the kernel also fails the concentration condition required for RFF-based dequantization\footnote{Note that this kernel satisfies the concentration condition of Ref.~\cite{sweke2023potential}. This does not contradict our results since Ref.~\cite{sweke2023potential} considers QNN-regression. However, this example showcases how difficult it is to discuss non-dequantizability with only sufficient conditions. Therefore, further work on the necessary conditions for dequantization would be interesting.}. 

Moreover, there are models for which the difficulty of RFF-dequantization lies not in the anti-concentration of their Fourier spectrum but in the complexity of sampling from it. In other words, there could be some distributions that are concentrated--such that their corresponding quantum kernel avoids exponential concentration--but at the same time are hard to sample from with polynomial resources and as a result, unfit for RFF-dequantization. One such possible family of distributions are peaked distributions~\cite{aaronson2024verifiable, zhang2025complexity}. While concentrated, these distributions are expected to be hard to sample from. As these distributions are concentrated they are likely to be non-exponentially concentrated. Moreover, if these distributions are proven to be hard to sample from, they will also be non-RFF-dequantizable. Consequently, they could be used to construct QKs with potential quantum advantage. As little is known about these distributions, this remains an open question for now.

Although these examples are constructed in a classical setting, the universality of quantum kernels \cite{gil2024expressivity} implies that such structures can be efficiently implemented on quantum hardware. Moreover, these examples are highly-contrived and are only provided to show that the connection between exponential concentration and dequantization is an intricate one. 

\medskip

We end by listing some further the subtleties that should be considered when discussing the relationship between our work and the findings in Ref.~\cite{cerezo2023does}. 
\begin{itemize}
    \item For classification tasks, exponential concentration is less relevant, as data distributions are rarely uniform.
    \item Exponentially concentrated kernels yield input-independent outputs, making RFF surrogation trivial; this parallels the observation that “guessing zero” approximates the loss on a barren plateau.
    \item In the standard cases studied in Ref.~\cite{cerezo2023does}, kernels that are not exponentially concentrated typically admit efficient classical simulation via tensor-network, light-cone, or Pauli propagation methods.
    \item Both in barren plateau–free QNNs and RFF-dequantizable kernels, classical simulability is theoretically guaranteed only when the respective sufficient conditions are provably satisfied; otherwise, simulation remains possible in principle but unconstructive in practice.
\end{itemize}

Overall, while dequantization and barren plateaus concern distinct aspects of quantum models—representation versus optimization—they share conceptual parallels regarding concentration phenomena and classical simulability.

\clearpage
\end{document}